\documentclass[11pt]{article}
\usepackage[vcentermath]{youngtab}

\newcommand{\cmp}{Comm. Math. Phys.~}

\newcommand{\jmp}{J. Math. Phys.~}

\newcommand{\jpa}{J. Phys. A~}

\newcommand{\njp}{New. J. Phys.~}

\newcommand{\prl}{Phys. Rev. Lett.~}
\newcommand{\pra}{Phys. Rev. A~}

\usepackage{young}
\usepackage{appendix}
\usepackage{epstopdf}
\usepackage{xcolor}
\usepackage{subfigure}
\usepackage[T1]{fontenc}
\usepackage[sc]{mathpazo}
\usepackage{amsmath}
\usepackage{amssymb}
\usepackage{graphicx,color}
\usepackage{framed}
\usepackage{multirow}
\usepackage{enumerate}
\usepackage{amsthm}
\usepackage{amsfonts,mathrsfs}
\usepackage{geometry} 
\usepackage{eepic}
\usepackage{ifthen}
\usepackage[vcentermath]{youngtab}
\usepackage[unicode=true,pdfusetitle, bookmarks=true,bookmarksnumbered=false,bookmarksopen=false, breaklinks=false,pdfborder={0 0 0},backref=false,colorlinks=false] {hyperref}
\hypersetup{
colorlinks,linkcolor=myurlcolor,citecolor=myurlcolor,urlcolor=myurlcolor}
\definecolor{myurlcolor}{rgb}{0,0,0.7}

\usepackage{color}

\newcommand{\red}{\textcolor{red}}

\newcommand{\blue}{\textcolor{blue}}

\newcommand{\proj}[1]{| #1\rangle\!\langle #1 |}

\usepackage{hyperref}
\hypersetup{pdfpagemode=UseNone}

\newcommand{\tinyspace}{\mspace{1mu}}

\newcommand{\op}[1]{\operatorname{#1}}

\newcommand{\abs}[1]{\left\lvert\tinyspace #1 \tinyspace\right\rvert}

\newcommand{\norm}[1]{\left\lVert\tinyspace #1 \tinyspace\right\rVert}

\renewcommand{\det}{\operatorname{det}}
\renewcommand{\t}{{\scriptscriptstyle\mathsf{T}}}

\newcommand{\setft}[1]{\mathrm{#1}}
\newcommand{\lin}[1]{\setft{L}\left(#1\right)}
\newcommand{\density}[1]{\setft{D}\left(#1\right)}
\newcommand{\unitary}[1]{\setft{U}\left(#1\right)}

\renewcommand{\vec}{\op{vec}}

\newcommand{\spn}{\op{span}}
\newcommand{\sign}{\op{sign}}

\def\GL{\mathsf{GL}}
\def\SL{\mathsf{SL}}

\def\SU{\mathsf{SU}}
\def\gl{\mathfrak{gl}}
\def\u{\mathfrak{u}}
\def\g{\mathfrak{g}}

\def\liet{\mathfrak{t}}

\def\Ad{\mathrm{Ad}}

\def \dif {\mathrm{d}}
\def \diag {\mathrm{diag}}

\def\complex{\mathbb{C}}
\def\real{\mathbb{R}}
\def\natural{\mathbb{N}}
\def\integer{\mathbb{Z}}
\def\I{\mathbb{1}}

\newenvironment{mylist}[1]{\begin{list}{}{
    \setlength{\leftmargin}{#1}
    \setlength{\rightmargin}{0mm}
    \setlength{\labelsep}{2mm}
    \setlength{\labelwidth}{8mm}
    \setlength{\itemsep}{0mm}}}
    {\end{list}}

\def\ot{\otimes}

\newcommand{\inner}[2]{\langle #1 , #2\rangle}
\newcommand{\iinner}[2]{\langle #1 | #2\rangle}
\newcommand{\out}[2]{| #1\rangle\langle #2 |}
\newcommand{\Inner}[2]{\left\langle #1 , #2\right\rangle}
\newcommand{\Innerm}[3]{\left\langle #1 \left| #2 \right| #3 \right\rangle}
\newcommand{\defeq}{\stackrel{\smash{\textnormal{\tiny def}}}{=}}

\newcommand{\Hom}{\mathrm{Hom}}
\newcommand{\End}{\mathrm{End}}

\newcommand{\pa}[1]{(#1)}
\newcommand{\Pa}[1]{\left(#1\right)}

\newcommand{\Br}[1]{\left[#1\right]}
\newcommand{\set}[1]{\{#1\}}
\newcommand{\Set}[1]{\left\{#1\right\}}

\newcommand{\ket}[1]{|#1\rangle}
\newcommand{\Ket}[1]{\left|#1\right\rangle}

\def\Jamiolkowski{J}
\newcommand{\jam}[1]{\Jamiolkowski\pa{#1}}

\DeclareMathOperator{\trace}{Tr}
\newcommand{\ptr}[2]{\trace_{#1}\pa{#2}}
\newcommand{\Ptr}[2]{\trace_{#1}\Pa{#2}}
\newcommand{\tr}[1]{\ptr{}{#1}}
\newcommand{\Tr}[1]{\Ptr{}{#1}}

\newcommand{\Abs}[1]{\left|\tinyspace#1\tinyspace\right|}

\def\cA{\mathcal{A}}\def\cB{\mathcal{B}}\def\cC{\mathcal{C}}\def\cE{\mathcal{E}}
\def\cF{\mathcal{F}}\def\cH{\mathcal{H}}\def\cI{\mathcal{I}}
\def\cM{\mathcal{M}}\def\cN{\mathcal{N}}\def\cO{\mathcal{O}}
\def\cT{\mathcal{T}}
\def\cV{\mathcal{V}}\def\cX{\mathcal{X}}\def\cY{\mathcal{Y}}

\def\bP{\mathbf{P}}\def\bQ{\mathbf{Q}}

\def\bsA{\boldsymbol{A}}\def\bsB{\boldsymbol{B}}\def\bsC{\boldsymbol{C}}\def\bsD{\boldsymbol{D}}
\def\bsF{\boldsymbol{F}}\def\bsG{\boldsymbol{G}}\def\bsJ{\boldsymbol{J}}
\def\bsM{\boldsymbol{M}}\def\bsN{\boldsymbol{N}}
\def\bsP{\boldsymbol{P}}\def\bsT{\boldsymbol{T}}
\def\bsU{\boldsymbol{U}}\def\bsV{\boldsymbol{V}}\def\bsW{\boldsymbol{W}}\def\bsX{\boldsymbol{X}}\def\bsY{\boldsymbol{Y}}
\def\bsZ{\boldsymbol{Z}}

\def\bsd{\boldsymbol{d}}

\def\bsu{\boldsymbol{u}}\def\bsv{\boldsymbol{v}}

\def\rG{\mathrm{G}}
\def\rL{\mathrm{L}}
\def\rT{\mathrm{T}}

\def\sE{\mathscr{E}}

\def\T{\textsf{T}}
\def\U{\textsf{U}}

\newtheorem{thrm}{Theorem}[section]

\newtheorem{prop}[thrm]{Proposition}
\newtheorem{cor}[thrm]{Corollary}
\theoremstyle{definition}
\newtheorem{definition}[thrm]{Definition}
\newtheorem{remark}[thrm]{Remark}
\newtheorem{exam}[thrm]{Example}

\numberwithin{equation}{section}

\newcounter{questionnumber}

\begin{document}

\title{\bf\large Matrix Integrals over Unitary Groups: An Application of Schur-Weyl Duality}

\author{\blue{Lin Zhang}\footnote{E-mail: godyalin@163.com}\\
  {\it\small School of Science, Hangzhou Dianzi University, Hangzhou 310018,
  PR~China}}

\date{}
\maketitle
\begin{abstract}
The integral formulae pertaining to the unitary group $\U(d)$ have
been comprehensively reviewed, yielding fresh results and innovative
proofs. Central to the derivation of these formulae lies the
employment of Schur-Weyl duality, a classical and powerful theorem
from the representation theory of groups. This duality serves as a
bridge, establishing a profound connection between the
representation theory of finite groups (or permutation groups) and
that of classical Lie groups, specifically the unitary groups. From
the perspective of Schur-Weyl duality, it becomes evident that the
computation of matrix integrals over the unitary group is
intricately intertwined with the so-called \blue{\texttt Weingarten
function}. The explicit evaluation of this function is heavily
dependent on three crucial aspects: firstly, the dimensions of the
irreducible representations of the unitary groups; secondly, the
dimensions of the irreducible representations of permutation groups;
and thirdly, the irreducible characters of permutation groups. For
the first two aspects, we can rely on well-established formulae.
Specifically, the dimensions of irreducible representations of both
unitary and permutation groups can be determined using the
\blue{\texttt hook-length formula} attributed to Frame, Robinson,
and Thrall, as well as the \blue{\texttt hook-content formula}
proposed by Stanley. However, the third aspect poses a more
intricate challenge. Unfortunately, despite significant efforts,
there remains no unifying closed-form formula for the generic
irreducible characters of permutation groups, except for a few
special cases involving particular partitions. Given the
significance of these irreducible characters, it is crucial to have
a comprehensive understanding of them. Fortunately, all the
information pertaining to the irreducible characters belonging to a
given permutation group is encoded in a so-called \blue{\texttt
character table}. For the convenience of researchers and
practitioners, we have included these character tables, encompassing
groups of orders up to six. These tables provide a concise and
accessible summary of the irreducible characters, enabling a deeper
understanding and easier manipulation of the integral formulae over
the unitary group $\U(d)$.
\end{abstract}
\newpage
\tableofcontents
\newpage
\section{Introduction}

This review article focuses primarily on the derivation and
application of useful matrix integrals over the unitary group.
Central to our approach in computing these integrals over the
unitary group $\U(d)$ is the Schur-Weyl duality, a powerful tool
derived from the representation theory of groups. Before delving
into the specifics of Schur-Weyl duality, it is imperative to
briefly introduce its underlying principles and explore its
applications in quantum information theory.

In classical information theory, the method of types serves as a
fundamental tool for executing various tasks, including estimating
probability distributions, randomness concentration, and data
compression. Notably, it has been demonstrated that the Schur basis
can be leveraged to generalize the classical method of types,
thereby enabling us to perform quantum analogues of these tasks.
This generalization is particularly relevant when studying systems
exhibiting permutation symmetry, as the Schur basis offers a natural
framework for such investigations.

The Schur transformation, which is closely related to Schur-Weyl
duality, finds applications in a wide range of quantum information
tasks. For instance, it can be employed to estimate the spectrum of
an unknown mixed state, enabling us to gain insights into the
quantum properties of complex systems. Additionally, the Schur
transformation facilitates universal distortion-free entanglement
concentration using only local operations, a crucial step in quantum
communication and quantum computing. Furthermore, it can be used to
encode information into decoherence-free subsystems, enhancing the
resilience of quantum information against environmental noise.
Beyond these applications, the Schur transformation also finds use
in scenarios where communication is conducted without a shared
reference frame, enabling robust quantum communication even in the
absence of a common reference point. Moreover, it plays a pivotal
role in the universal compression of quantum data, providing
efficient means to reduce the storage requirements for quantum
information.

Explicit results related to Schur-Weyl duality are numerous and
diverse. For instance, Keyl and Werner pioneered the use of
Schur-Weyl duality to estimate the spectrum of an unknown mixed
state $\rho$ with $d$ levels from its $k$-fold product state
\cite{Keyl}. This approach provides a practical means to analyze the
quantum properties of complex systems without full knowledge of
their states. Moreover, Harrow contributed significantly to the
computational aspect of Schur and Clebsch-Gordan transforms,
designing efficient quantum circuits for their implementation
\cite{Harrow,Bacon}. These circuits facilitate the manipulation of
quantum states and enable the execution of complex quantum
algorithms. Christandl, in his series of works
\cite{Christ2006,Christ2012}, employed Schur-Weyl duality to
investigate the intricate structure of multipartite quantum states.
His group-theoretic approach not only offered novel insights into
the entanglement properties of quantum systems but also led to the
derivation of entropy inequalities for von Neumann entropy,
including the strong subadditivity property. Gour, on the other
hand, leveraged Schur-Weyl duality to classify multipartite
entanglement in the finite-dimensional setting \cite{Gour}. This
classification scheme provides a framework for understanding and
manipulating entanglement, a crucial resource in quantum information
processing. Furthermore, generalizations of Schur-Weyl duality have
found applications in quantum estimation tasks. For instance,
Mashhad and Spekkens explored such generalizations and demonstrated
their utility in quantum estimation problems
\cite{Mashhad,Spekkens}. These generalizations broaden the scope of
Schur-Weyl duality, making it a versatile tool for addressing a wide
range of quantum information challenges.

In summary, Schur-Weyl duality has emerged as a pivotal tool in
quantum information theory, not only enabling the derivation of
explicit results but also facilitating the analysis and manipulation
of quantum systems. This duality offers a powerful framework for
understanding and manipulating quantum states, with potential
applications spanning a broad range from spectrum estimation and
entanglement classification to quantum state estimation,
entanglement concentration, and quantum data compression.
Additionally, this review article delves into the intricate
connections between Schur-Weyl duality, matrix integrals over the
unitary group, and their diverse applications in quantum information
theory. By leveraging the Schur basis and Schur transformation, we
acquire valuable insights into quantum systems, opening up new
avenues for addressing complex quantum challenges.

\section{Schur-Weyl duality}

In this section, we give the details of Schur-Weyl duality. A generalization of this duality is obtained within the
framework of the infinite-dimensional $C^*$-algebras \cite{Neeb}. In
order to arrive at Schur-Weyl duality, we need the following
ancillary results, well-known facts in representation theory. We assume familiarity with knowledge of representation theory of a compact Lie group or finite group \cite{Sepanski,Wallach}.

Before proceeding, let us recall some notions concerning von Neumann
algebras. Let $\cH$ be a Hilbert space and $\lin{\cH}$ the all
bounded linear operator defines on $\cH$. Assume that $\cM$ is a
subset of $\lin{\cH}$, we denote its \emph{commutant} $\cM'$ by the
set of all bounded operators on $\cH$ commuting with every operator
in $\cM$, that is
$$
\cM' \defeq \Set{\bsM'\in\lin{\cH}:
[\bsM',\bsM]=\bsM'\bsM-\bsM\bsM'=0\text{\ for all\ }\bsM\in\cM}.
$$
One has
\begin{eqnarray*}
\cM &\subseteq& \cM'' = \cM^{(\mathrm{iv})} = \cM^{(\mathrm{vi})} =
\cdots\\
\cM' &=& \cM''' = \cM^{(\mathrm{v})} = \cM^{(\mathrm{vii})} = \cdots
\end{eqnarray*}

\begin{definition}
A $\ast$-algebra $\cM$ on $\cH$ is said to be a \emph{von Neumann
algebra} if $\cM=\cM''$. The \emph{center} $\cC$ of a von Neumann
algebra $\cM$ is defined by $\cC \defeq \cM\bigcap\cM'$. A von
Neumann algebra is called a \emph{factor} if $\cC=\complex\I$.
\end{definition}

It will be useful throughout this paper to make use of a simple
correspondence between the spaces $\Hom(\cX, \cY)$ and $\cY \ot
\cX$, for given complex Euclidean spaces $\cX$ and $\cY$. The
mapping
\begin{eqnarray}
\vec: \Hom(\cX,\cY) \longrightarrow \cY \ot \cX
\end{eqnarray}
can be defined to be the linear mapping that represents a change of
bases from the standard basis of $\Hom(\cX, \cY)$ to the standard
basis of $\cY \ot \cX$. Specifically,
\begin{eqnarray}
\vec(\out{i}{j}) = \ket{i}\ot\ket{j} = \ket{ij}
\end{eqnarray}
for all $i,j$ at which point the mapping is determined for every
operator $\bsA \in \Hom(\cX, \cY)$ by linearity.

Next we review the Schmidt decomposition for bipartite pure states.
\begin{prop}\label{bi-vector}
Let $\ket{\psi} = \sum_{i,j=1}^{d_A,d_B}\gamma_{ij}\ket{a_i b_j} \in
\cH_A \ot \cH_B$ be a vector in the tensor product of two Hilbert
spaces, where $\set{\ket{a_i}}_{i=1}^{d_A}$ is an orthonormal basis
for $\cH_A$ and $\set{\ket{b_i}}_{i=1}^{d_B}$ is an orthonormal  for
$\cH_B$, respectively. Then there exists an orthonormal basis
$\set{\ket{\alpha_i}}$ of $\cH_A$ and an orthonormal basis
$\set{\ket{\beta_j}}$ of $\cH_B$ such that
\begin{center}
\fcolorbox{purple}{lightgray}{
\parbox{16.3cm}{
\begin{eqnarray}
\ket{\psi} = \sum_{k=1}^R \lambda_k \ket{\alpha_k \beta_k}
\end{eqnarray}}
}
\end{center}
holds, with positive real coefficients $\lambda_k$. The $\lambda_k$
are uniquely determined as the square roots of the eigenvalues of
the matrix $\Gamma\Gamma^\dagger$, where $\Gamma=[\gamma_{ij}]$ is
the matrix formed by the coefficients of $\ket{\psi}$. The number $R
\leqslant \min \set{d_A,d_B}$ is called the Schmidt rank of
$\ket{\psi}$. If the $\lambda_k$ are pairwise different, then also
the $\ket{\alpha_k}$ and $\ket{\beta_k}$ are unique up to a phase.
\end{prop}

\begin{proof}
For the given vector $\ket{\psi} \in \cH_A \ot \cH_B$, there exist a
matrix $\Gamma_{\psi}:\cH_B \to \cH_A$ such that
$\vec(\Gamma_{\psi}) = \ket{\psi}$, where
$\Gamma_{\psi}=[\gamma_{ij}]$. Now by the SVD of a matrix, there
exist positive real $\set{\lambda_k}_{k=1}^R$(singular values of
$\Gamma_{\psi}$), and an orthonormal set
$\set{\ket{\alpha_k}}_{k=1}^R$ in $\cH_{A}$ and an orthonormal set
$\{|\beta_{k}\rangle\}_{k=1}^{R}$ in $\cH_{B}$ such that
$$
\Gamma_{\psi} = \sum_{k=1}^R \lambda_k
\ket{\alpha_k}(\ket{\beta_k})^\t,
$$
which implies that
$$\ket{\psi} = \vec(\Gamma_{\psi}) = \sum_{k=1}^{R} \lambda_k \ket{\alpha_k}\ket{\beta_k},$$
where $R$ is the rank of $\Gamma_{\psi}$. It is clearly that $R
\leqslant \min \set{d_A,d_B}$.
\end{proof}

When the $\vec$ mapping is generalized to multipartite spaces,
caution should be given to the bipartite case (multipartite
situation similarly). Specifically, for given complex Euclidean
spaces $\cX_{A/B}$ and $\cY_{A/B}$,
\begin{eqnarray}
\vec: \Hom(\cX_A \ot \cX_B,\cY_A \ot \cY_B) \longrightarrow \cY_A
\ot \cX_A \ot \cY_B \ot \cX_B
\end{eqnarray}
is defined to be the linear mapping that represents a change of
bases from the standard basis of $\Hom(\cX_A \ot \cX_B,\cY_A \ot
\cY_B)$ to the standard basis of $\cY_A \ot \cX_A \ot \cY_B \ot
\cX_B$. Concretely,
\begin{eqnarray}
\vec(\out{m}{n} \ot \out{\mu}{\nu}): = \ket{mn} \ot \ket{\mu\nu},
\end{eqnarray}
where $\set{\ket{n}}$ is an orthonormal basis for $\cX_A$ and
$\set{\ket{\nu}}$ is an orthonormal basis for $\cX_B$, while
$\set{\ket{m}}$ is an orthonormal basis for $\cY_A$ and
$\set{\ket{\mu}}$ is an orthonormal basis for $\cY_B$. Analogously,
the mapping is determined for every operator $\bsX \in \Hom(\cX_A
\ot \cX_B,\cY_A \ot \cY_B)$ by linearity. Note that if $\bsX = \bsA
\ot \bsB$, where $\bsA \in \Hom(\cX_A,\cY_A)$ and $\bsB \in
\Hom(\cX_B,\cY_B)$, then
\begin{eqnarray}
\vec(\bsA \ot \bsB) = \vec(\bsA)\ot \vec(\bsB).
\end{eqnarray}

\begin{thrm}\label{bi-operator}
Suppose that $\set{\bsA_i}$ and $\set{\bsB_j}$ are some orthonormal
base for $\End(\cH_A)$ and $\End(\cH_B)$, respectively. $\bsW$ is
any operator in $\End(\cH_A \ot \cH_B) \equiv \End(\cH_A) \ot
\End(\cH_B)$ and it can be written as $\bsW =
\sum_{i,j=1}^{d^2_A,d^2_B} \gamma_{ij} \bsA_i \ot \bsB_j$. Then
there exist an orthonormal basis $\set{\bsG^A_i}$ for $\End(\cH_A)$
and an orthonormal basis $\set{\bsG^B_j}$ for $\End(\cH_B)$,
respectively, such that
\begin{center}
\fcolorbox{purple}{lightgray}{
\parbox{16.3cm}{
\begin{eqnarray}
\bsW = \sum_{k=1}^R \lambda_k \bsG^A_k \ot \bsG^B_k
\end{eqnarray}}}
\end{center}
holds, with positive real coefficients $\lambda_k$. The $\lambda_k$
are uniquely determined as the square roots of the eigenvalues of
the matrix $\Gamma\Gamma^\dagger$, where $\Gamma = [\gamma_{ij}]$ is
the matrix formed by the coefficients of $\bsW$ and $\gamma_{ij} =
\Tr{(\bsA_i \ot \bsB_j)^\dagger \bsW}$. The number $R \leqslant \min
\set{d^2_A,d^2_B}$.
\end{thrm}

\begin{proof}
Since $\bsW$ is any operator in $\End(\cH_A \ot \cH_B) \equiv
\End(\cH_A) \ot \End(\cH_B)$, and it can be written as
$$
\bsW = \sum_{i,j=1}^{d^2_A, d^2_B} c_{ij} \bsA_i \ot \bsB_j,
$$
where $\set{\bsA_i}$ and $\set{\bsB_j}$ are some orthonormal base
for $\End(\cH_A)$ and $\End(\cH_B)$, respectively. By the
vectorization of bipartite operators,
$$
\vec(\bsW) = \sum_{i,j=1}^{d^2_A,d^2_B} \gamma_{ij} \vec(\bsA_i) \ot
\vec(\bsB_j) \in \cH_A \ot \cH_A \ot \cH_B \ot \cH_B.
$$
Now write $\cH^{\ot2}_A = \cH_A \ot \cH_A$ and $\cH^{\ot2}_B = \cH_B
\ot \cH_B$, then $\cH_A \ot \cH_A \ot \cH_B \ot \cH_B = \cH^{\ot2}_A
\ot \cH^{\ot2}_B$ can be considered a bipartite space. So employing
the Schmidt decomposition of bipartite vectors gives rise to
$$
\vec(\bsW) = \sum_{k=1}^R \lambda_k \ket{\alpha_k} \ot \ket{\beta_k}
\in \cH^{\ot2}_A \ot \cH^{\ot2}_B,
$$
where $\ket{\alpha_k} \in \cH^{\ot2}_A$ and $\ket{\beta_k} \in
\cH^{\ot2}_B$ for which there exist an orthonormal basis
$\set{\bsG^A_i}$ for $\End(\cH_A)$ and an orthonormal basis
$\set{\bsG^B_j}$ for $\End(\cH_B)$, respectively, such that
$$
\vec(\bsG^A_k) = \ket{\alpha_k}~~\text{and}~~\vec(\bsG^B_k) =
\ket{\beta_k},
$$
which implies that
$$
\vec(\bsW) = \sum_{k=1}^R \lambda_k \vec(\bsG^A_k) \ot
\vec(\bsG^B_k) \Longleftrightarrow \bsW = \sum_{k=1}^R \lambda_k
\bsG^A_k \ot \bsG^B_k.
$$
Clearly, $R$ is the rank of $\bsW$ and
$R\leqslant\min\set{d^2_A,d^2_B}$.
\end{proof}

\begin{cor}\label{sct22}
With the above notations, if $\bsW$ is a Hermite operator, then
$\set{\bsG^{A/B}_k}_{k=1}^R$ can be required as two collections of
orthonormal Hermite matrices.
\end{cor}

\begin{prop}
Let $V$ and $W$ be finite dimensional complex vector spaces. If $\cM
\subseteq \End(V)$ and $\cN \subseteq \End(W)$ are von Neumann
algebras, then
\begin{center}
\fcolorbox{purple}{lightgray}{
\parbox{6cm}{
$$
(\cM\ot\cN)'=\cM'\ot\cN'.
$$
}}
\end{center}
\end{prop}

\begin{proof}
Apparently, $\cM' \ot \cN'\subseteq(\cM\ot\cN)'$. It suffices to
show that $(\cM\ot\cN)' \subseteq \cM' \ot \cN'$. For arbitrary
$\bsT \in (\cM\ot\cN)'$, by the Operator-Schmidt Decomposition,
\begin{eqnarray*}
\bsT = \sum_j \lambda_j \bsA_j\ot \bsB_j,
\end{eqnarray*}
where $\lambda_j\geqslant0$, and $\bsA_j$ and $\bsB_j$ are
orthonormal bases of $\End(\complex^m)$ and $\End(\complex^n)$,
respectively. Now for arbitrary $\bsM\in\cM$ and $\bsN\in\cN$,
$\bsM\ot\I_W, \I_V\ot \bsN\in\cM\ot\cN$, it follows that
\begin{eqnarray*}
[\bsT, \bsM\ot\I_W] = 0 = [\bsT,\I_V\ot \bsN].
\end{eqnarray*}
That is,
\begin{eqnarray*}
\sum_j \lambda_j [\bsA_j,\bsM]\ot \bsB_j = 0~~\text{and}~~ \sum_j
\lambda_j \bsA_j\ot [\bsB_j, \bsN] = 0.
\end{eqnarray*}
We drop those terms for which $\lambda_j$ are zero. Thus $\lambda_j$
is positive for all $j$ in the above two equations. Since
$\set{\bsA_j}$ and $\set{\bsB_j}$ are linearly independent,
respectively, it follows that
$$
[\bsA_j,\bsM] = 0~~\text{and}~~[\bsB_j, \bsN] = 0.
$$
This implies that $\bsA_j\in\cM'$ and $\bsB_j\in\cN'$. Therefore
$\bsT\in\cM'\ot\cN'$.
\end{proof}

\begin{prop}[The dual theorem]\label{double-commutants}
Let $V$ be a representation of a finite group $G$ with decomposition
$V \cong \bigoplus_{\alpha\in \widehat G}  V^{\oplus
n_\alpha}_\alpha\cong \bigoplus_{\alpha\in \widehat G} V_\alpha\ot
\complex^{n_\alpha}$, where $\widehat G$ be a complete set of
inequivalent irreps of $G$, and $V_\alpha$'s are irreps of $G$. Let
$\cA$ be the algebra generated by $V$ and $\cB=\cA'$ its commutant.
Then
\begin{center}
\fcolorbox{purple}{lightgray}{
\parbox{16.3cm}{
\begin{eqnarray}
&&\cA \cong \bigoplus_{\alpha\in \widehat G}\End(V_\alpha)
\ot\I_{\complex^{n_\alpha}},\\
&&\cB \cong \bigoplus_{\alpha\in \widehat G}\I_{V_\alpha}
\ot\End(\complex^{n_\alpha}).
\end{eqnarray}}}
\end{center}
Furthermore we have $\cB'=\cA$, where $\cB'$ is the commutant of
$\cB$. That is $\cA = \cA''$ and $\cB = \cB''$. Thus both $\cA$ and
$\cB$ are von Neumann algebras.
\end{prop}

\begin{proof}
It is easy to see that
$$
\frac{d_\alpha}{\abs{G}}\sum_{g\in
G}\overline{V_{\alpha,ij}(g)}V(g)\in \cA.
$$
That is because $\cA$ is the linear algebra generated by $V$. Since
$V$ can be decomposed into the direct sum of irreducible components
$V_\alpha$'s, by the orthogonality of the functions $V_{\alpha,ij}$,
we get
\begin{eqnarray*}
\frac{d_\alpha}{\abs{G}}\sum_{g\in G}\overline{V_{\alpha,ij}(g)}V(g)
&\cong& \frac{d_\alpha}{\abs{G}}\sum_{g\in
G}\overline{V_{\alpha,ij}(g)}\Pa{\bigoplus_{\beta\in \widehat
G}V_{\beta}(g)\ot \I_{\complex^{n_\beta}}}\\
&=& \bigoplus_{\beta\in \widehat
G}\Pa{\frac{d_\alpha}{\abs{G}}\sum_{g\in
G}\overline{V_{\alpha,ij}(g)}V_{\beta}(g)}\ot
\I_{\complex^{n_\beta}}\\
&=& \bigoplus_{\beta\in \widehat
G}\Pa{\frac{d_\alpha}{\abs{G}}\sum_{g\in
G}\overline{V_{\alpha,ij}(g)}\sum_{k,l}V_{\beta,kl}(g)E_{\beta,kl}}\ot
\I_{\complex^{n_\beta}}\\
&=& \bigoplus_{\beta\in \widehat
G}\sum_{k,l}\Pa{\frac{d_\alpha}{\abs{G}}\sum_{g\in
G}\overline{V_{\alpha,ij}(g)}V_{\beta,kl}(g)}E_{\beta,kl}\ot
\I_{\complex^{n_\beta}} \\
&=& E_{\alpha, ij}\ot\I_{\complex^{n_\alpha}},
\end{eqnarray*}
where
$$
\frac{d_\alpha}{\abs{G}}\sum_{g\in
G}\overline{V_{\alpha,ij}(g)}V_{\beta,kl}(g)=(V_{\alpha,ij},V_{\beta,kl})=\delta_{(\alpha,ij),(\beta,kl)},
$$
implying $E_{\alpha, ij}\ot\I_{\complex^{n_\alpha}}\in \cA$, hence
$\End(V_{\alpha})\ot\I_{\complex^{n_\alpha}}\subseteq \cA$ by the
linearity. Now
$$
\cA = \spn_\complex\set{V(g): g\in G} \cong \bigoplus_{\alpha\in
\widehat G}\spn_\complex\set{V_{\alpha}(g)\ot
\I_{\complex^{n_\alpha}}} = \bigoplus_{\alpha\in \widehat
G}\End(V_{\alpha})\ot\I_{\complex^{n_\alpha}}.
$$
Clearly $\bigoplus_{\alpha\in \widehat
G}\I_{V_{\alpha}}\ot\End(\complex^{n_\alpha})\subseteq \cA' = \cB$.
To see that every element in $\cB$ is of this form, i.e.
$\cB\subseteq\bigoplus_{\alpha\in \widehat
G}\I_{V_{\alpha}}\ot\End(\complex^{n_\alpha})$. Consider a
projection $P_\alpha$ onto $V_\alpha\ot\complex^{n_\alpha}$. The
projectors $P_\alpha$ form a resolution of the identity
($\I_{\cA}=\oplus_\alpha P_\alpha$) and $P_\alpha\in \cA$. Since
$\cA'=\cB$, it follows that any $\bsB\in\cB$ must commute with
$P_\alpha$: $P_\alpha \bsB = \bsB P_\alpha$. This leads to
$$
\bsB = \Pa{\bigoplus_\alpha P_\alpha}\bsB = \bigoplus_\alpha
P_\alpha \bsB P_\alpha = \bigoplus_\alpha \bsB_\alpha.
$$
Moreover, $\bsB_\alpha:=P_\alpha\bsB P_\alpha\in
(\End(V_{\alpha})\ot\I_{\complex^{n_\alpha}})' =
\I_{V_{\alpha}}\ot\End(\complex^{n_\alpha})$, thus it must be of the
form $\bsB_\alpha = \I_{V_\alpha}\ot \hat\bsB_{\alpha}$ for some
$\hat\bsB_\alpha\in \End(\complex^{n_\alpha})$. This means that
$$
\bsB = \bigoplus_\alpha \bsB_\alpha = \bigoplus_\alpha
\I_{V_\alpha}\ot \hat\bsB_{\alpha}\in \bigoplus_{\alpha\in \widehat
G}\I_{V_{\alpha}}\ot\End(\complex^{n_\alpha}).
$$
That is, $\cB\subset\bigoplus_{\alpha\in \widehat
G}\I_{V_{\alpha}}\ot\End(\complex^{n_\alpha})$. Therefore
$\cB'=\cA$.
\end{proof}

\begin{remark}
For any reducible representation $V$, there is a basis under which
the action of $V(g)$ can be expressed as
\begin{eqnarray}\label{eq-decom-of-rep}
V(g)\cong \bigoplus_{\alpha\in \widehat G}
\bigoplus^{n_\alpha}_{j=1} V_\alpha(g) = \bigoplus_{\alpha\in
\widehat G} V_\alpha(g) \ot\I_{n_\alpha},
\end{eqnarray}
where $\alpha\in\widehat G$ labels an irrep $V_\alpha$ and
$n_\alpha$ is the multiplicity of the irrep $V_\alpha$ in the
representation $V$. Here we use $\cong$ to indicate that there
exists a \blue{unitary change of basis} relating the left-hand size
to the right-hand side. Under this change of basis we obtain a
similar decomposition of the representation space $V$ (known as the
\emph{isotypic decomposition}):
\begin{eqnarray}\label{eq-iso-decom}
V \cong \bigoplus_{\alpha\in\widehat G} V_\alpha\ot
\Hom_G(V_\alpha,V).
\end{eqnarray}
Since $G$ acts trivially on $\Hom_G(V_\alpha,V)$,
Eq.~\eqref{eq-decom-of-rep} remains the same.

The value of Eq.~\eqref{eq-iso-decom} is that the unitary mapping
from the right-hand side (RHS) to the left-hand side (LHS) has a
simple explicit expression: it corresponds to the canonical map
$\varphi: \cX\ot \Hom(\cX\ot\cY)\to \cY$ given by $\varphi(x\ot f) =
f(x)$. Of course, this doesn't tell us how to describe
$\Hom_G(V_\alpha,V)$, or how to specify an orthonormal basis for the
space, but we will later find this form of the decomposition useful.
\end{remark}

\begin{remark}
In the above proof, we see that
\begin{center}
\fcolorbox{purple}{lightgray}{
\parbox{6cm}{
\begin{eqnarray*}
\frac{d_\alpha}{\abs{G}}\sum_{g\in G}\overline{V_{\alpha,ij}(g)}V(g)
\cong E_{\alpha,ij}\ot\I_{\complex^{n_\alpha}},
\end{eqnarray*}}}
\end{center}
which implies that
\begin{center}
\fcolorbox{purple}{lightgray}{
\parbox{6cm}{
\begin{eqnarray*}
\frac{d_\alpha}{\abs{G}}\sum_{g\in G}\overline{\chi_{\alpha}(g)}V(g)
\cong \I_{V_\alpha}\ot\I_{\complex^{n_\alpha}},
\end{eqnarray*}}}
\end{center}
where $\chi_\alpha=\sum^{\dim(V_\alpha)}_{i=1} V_{\alpha,ii}$ is the
\blue{irreducible character} corresponding to the irrep $V_\alpha$,
and $\I_{V_\alpha} = \sum^{\dim(V_\alpha)}_{i=1}E_{\alpha,ii}$. This
amounts to say that
\begin{center}
\fcolorbox{purple}{lightgray}{
\parbox{16cm}{
\begin{eqnarray*}
P_\alpha:=\frac{d_\alpha}{\abs{G}}\sum_{g\in
G}\overline{\chi_{\alpha}(g)}V(g): V\longrightarrow V_\alpha\ot
\Hom_G(V_\alpha,V)
\end{eqnarray*}}}
\end{center}
is a projection onto some full isotropic component related to the
irrep $V_\alpha$. Based on the observation, we get that
$\Set{P_\alpha}_{\alpha\in\widehat G}$ is the orthogonal resolution
of identity on $V$: $\sum_{\alpha\in\widehat G}P_\alpha=\I_V$.
\end{remark}

Consider a system of $k$ qudits, each with a standard local
computational basis $\set{\ket{i},i=1,\ldots,d}$. The Schur-Weyl
duality relates transforms on the system performed by local
$d$-dimensional unitary operations to those performed by permutation
of the qudits. Recall that the symmetric group $S_k$ is the group of
all permutations of $k$ objects. This group is naturally represented
in our system by
\begin{center}
\fcolorbox{purple}{lightgray}{
\parbox{16.3cm}{
\begin{eqnarray}\label{eq:P}
\bP(\pi)\ket{i_1\cdots i_k} := \ket{i_{\pi^{-1}(1)}\cdots
i_{\pi^{-1}(k)}},
\end{eqnarray}}}
\end{center}
where $\pi\in S_k$ is a permutation and $\ket{i_1\cdots i_k}$ is
shorthand for $\ket{i_1}\ot\cdots\ot\ket{i_k}$. Let $\U(d)$ denote
the group of $d\times d$ unitary operators. This group is naturally
represented in our system by
\begin{center}
\fcolorbox{purple}{lightgray}{
\parbox{16.3cm}{
\begin{eqnarray}\label{eq:Q}
\bQ(\bsU)\ket{i_1\cdots i_k} := \bsU\ket{i_1}\ot\cdots\ot
\bsU\ket{i_k},
\end{eqnarray}}}
\end{center}
where $\bsU\in\U(d)$. In fact, $\bQ(\bsU)=\bsU^{\ot k}$, which is
called the \emph{collective action} of $\bsU\in\U(d)$. Thus we have
the following famous result:
\begin{center}
\fcolorbox{purple}{lightgray}{
\parbox{16.3cm}{
\begin{thrm}[Schur]\label{th-schur}
Let $\cA = \spn_\complex\Set{\bP(\pi): \pi\in S_k}$ and $\cB =
\spn_\complex\Set{\bQ(\bsU): \bsU\in\mathsf{U}(d)}$. Then:
\begin{eqnarray}
\cA' = \cB\quad\text{and}\quad \cA = \cB'.
\end{eqnarray}
\end{thrm}}}
\end{center}

\begin{remark}
When treated as matrix algebras, such pairs $(\cA,\cB)$ are known as
\emph{dual reductive pairs} since the collective action of the
unitary group on the tensor space and the permutation action of
tensor factors are mutual commutants.
\end{remark}

\begin{proof}[First proof]
The proof is separated into two steps:
\begin{enumerate}[(i)]
\item $\cA'= \spn_\complex\set{\bsA^{\ot k}: \bsA\in \End(\complex^d)}$.
\item $\spn_\complex\set{\bsA^{\ot k}: \bsA\in \End(\complex^d)}=\cB$.
\end{enumerate}
In order to show that the item (i) holds, note that
$\End((\complex^d)^{\ot k}) = \End(\complex^d)^{\ot k}$, Firstly we
show that
\begin{eqnarray}
\cA'= \End(\complex^d)^{\ot
k}\bigcap\bP(S_k)'=\End((\complex^d)^{\ot k})\bigcap\bP(S_k)' =
\spn_\complex\set{\bQ(\bsA): \bsA\in \End(\complex^d)}.
\end{eqnarray}
We need only show that LHS is contained in RHS (that is,
$\cA'\subset\cB$) since the reverse inclusion $\cB\subset\cA'$ is
trivial.

For arbitrary $\Gamma\in \cA'=\End(\complex^d)^{\ot
k}\bigcap\bP(S_k)'$, we have $\Gamma = \cT_{S_k}(\Gamma)$ and
$\Gamma\in \End(\complex^d)^{\ot k}$, where
$\cT_{S_k}=\frac1{k!}\sum_{\pi\in S_k}\Ad_{\bP(\pi)}$. It suffices
to show that
$$
\Gamma := \cT_{S_k}(\bsA_1\ot\cdots\ot \bsA_k)\in
\spn_{\complex}\set{\bQ(\bsA): \bsA\in \End(\complex^d)},
$$
where $\bsA_j\in \End(\complex^d)$. In what follows, we show that
each such $\Gamma$ can be written in terms of tensor products
$\bsA^{\ot k}$. Since
\begin{eqnarray*}
\frac{\dif}{\dif t}(\bsM+t\bsN)^{\ot k}
=\sum^{k-1}_{j=0}(\bsM+t\bsN)^{\ot j}\bsN(\bsM+t\bsN)^{\ot (k-j-1)},
\end{eqnarray*}
it follows that
\begin{eqnarray*}
\left.\frac{\dif}{\dif t}\right|_{t=0}(\bsM+t\bsN)^{\ot k}
=\sum^{k-1}_{j=0}\bsM^{\ot j}\bsN\bsM^{\ot (k-j-1)}.
\end{eqnarray*}
Consider the following partial derivative
\begin{eqnarray}\label{eq:partial}
\left.\frac{\partial^{k-1}}{\partial t_2\ldots\partial
t_k}\right|_{t_2=\cdots=t_k=0}\Pa{\bsA_1+\sum^k_{j=2}t_j\bsA_j}^{\ot
k},
\end{eqnarray}
which can be realized by subsequently applying
\begin{eqnarray*}
\left.\frac{\partial}{\partial t_j}\right|_{t_j=0}\Pa{\bsA_1 +
t_j\bsA_j}^{\ot k} = \lim_{t_j\to 0} \frac{(\bsA + t_j\bsA_j)^{\ot
k} - \bsA^{\ot k}}{t_j},
\end{eqnarray*}
iteratively going from $j=k$ all the way to $j=2$. The
\eqref{eq:partial} takes the form of a limit of sums of tensor
powers. Thus we obtain that
\begin{eqnarray*}
\left.\frac{\partial^{k-1}}{\partial t_2\ldots\partial
t_k}\right|_{t_2=\cdots=t_k=0}\Pa{\bsA_1+\sum^k_{j=2}t_j\bsA_j}^{\ot
k}=\sum_{\tau\in S_k}\bsA_{\tau(1)}\ot\bsA_{\tau(2)}\ot\cdots\ot
\bsA_{\tau(k)}.
\end{eqnarray*}
Since $\spn_\complex\set{\bQ(\bsA): \bsA\in \End(\complex^d)}$ is a
finite dimensional vector space, this limit is contained in
$\spn_\complex\set{\bQ(\bsA): \bsA\in \End(\complex^d)}$. On the
other hand, a direct calculation shows that
\begin{center}
\fcolorbox{purple}{lightgray}{
\parbox{16cm}{
\begin{eqnarray}
k!\cdot\Gamma = \left.\frac{\partial^{k-1}}{\partial
t_2\ldots\partial
t_k}\right|_{t_2=\cdots=t_k=0}\Pa{\bsA_1+\sum^k_{j=2}t_j\bsA_j}^{\ot
k}
\end{eqnarray}}}
\end{center}
and hence all operators $\Gamma$ are contained in
$\spn_\complex\set{\bQ(\bsA): \bsA\in \End(\complex^d)}$.

Now we turn to prove that the item (ii) holds. Firstly we show that
\begin{eqnarray*}
\spn_\complex\set{\bsU^{\ot k}: \bsU\in\U(d)} =
\spn_\complex\set{\bsT^{\ot k}: \bsT\in \GL(d,\complex)}.
\end{eqnarray*}
For any $\bsT\in \GL(d,\complex)$, there exists $\bsM\in
\End(\complex^d)$ such that
$$
\bsT = e^{\bsM}.
$$
(The elementary proof of this fact is shifted to the following
remark.) Then
\begin{eqnarray}\label{eq:exp}
\bsT^{\ot n} = (e^{\bsM})^{\ot n} = \exp\Pa{\sum^n_{j=1}\I^{\ot
j-1}\ot \bsM\ot\I^{n-j}} = \exp(\bQ_*(\bsM)),
\end{eqnarray}
where
$$
\bQ_*(\bsM) := \left.\frac{\dif}{\dif t}\right|_{t=0}
\bQ(e^{t\bsM}).
$$
Clearly $\bQ(e^{t\bsM}) = e^{t\bQ_*(\bsM)}$ for any real
$t\in\real$. In fact, $\bQ$ is a Lie group representation of $\U(d)$
or $\GL(d,\complex)$. $\bQ_*$ is a Lie algebra representation
induced by $\bQ$. If we can show that $\bQ_*(\bsM)\in
\spn_{\complex}\set{\bsU^{\ot n}: \bsU\in\U(d)}$, then by
\eqref{eq:exp}, it follows that $\bsT^{\ot n}\in
\spn_\complex\set{\bsU^{\ot n}: \bsU\in\U(d)}$.

Next, we show that $\bQ_*(\bsM)\in \spn_\complex\set{\bsU^{\ot n}:
\bsU\in\U(d)}$. For any skew-Hermitian operator $\bsX$, $e^{t\bsX}$
is a unitary, thus $\bQ(e^{t\bsX})\in \spn_\complex\set{\bsU^{\ot
n}: \bsU\in\U(d)}$. By the connection of $\bQ$ and $\bQ_*$, we have
$\bQ(e^{t\bsX}) = e^{t\bQ_*(\bsX)}$, implying that $\bQ_*(\bsX)\in
\spn_\complex\set{\bsU^{\ot n}: \bsU\in\U(d)}$, where
$\bsX\in\u(d)$, a Lie algebra of $\U(d)$. Let $\bsM =
\bsX+\sqrt{-1}\bsY$ for $\bsX,\bsY\in\u(d)$. Thus by the
complex-linearity of $\bQ_*$, it follows that
$$
\bQ_*(\bsM) = \bQ_*(\bsX) + \sqrt{-1}\bQ_*(\bsY).
$$
Since $\spn_\complex\set{\bsU^{\ot n}: \bsU\in\U(d)}$ is a
complex-linear space, it follows that
\begin{eqnarray*}
\bQ_*(\bsX) + \sqrt{-1}\bQ_*(\bsY) \in \spn_\complex\set{\bsU^{\ot
n}: \bsU\in\U(d)}
\end{eqnarray*}
whenever $\bQ_*(\bsX),\bQ_*(\bsY)\in\spn_\complex\set{\bsU^{\ot n}:
\bsU\in\U(d)}$. Therefore $\bQ_*(\bsM)\in
\spn_\complex\set{\bsU^{\ot n}: \bsU\in\U(d)}$.

Up to now, we established the fact that
\begin{eqnarray*}
\spn_\complex\set{\bsU^{\ot n}: \bsU\in\U(d)} =
\spn_\complex\set{\bsT^{\ot n}: \bsT\in \GL(d,\complex)}.
\end{eqnarray*}
Secondly, we show that
\begin{eqnarray}\label{eq:exp-inv}
\spn_\complex\set{\bsT^{\ot n}: \bsT\in \GL(d,\complex)} =
\spn_\complex\set{\bsA^{\ot n}: \bsA\in \End(\complex^d)}.
\end{eqnarray}
We use the fact that $\GL(d,\complex)$ is dense in
$\End(\complex^d)$. Indeed, for any $\bsA\in \End(\complex^d)$, by
the Singular Value Decomposition, we have
$$
\bsA = \bsU\bsD\bsV^\dagger,
$$
where $\bsU,\bsV\in \U(d)$ and $\bsD$ is a diagonal matrix whose
diagonal entries are nonnegative. Define
$$
\bsT_\varepsilon =
\bsU\Pa{\bsD+\frac{\varepsilon}{1+\varepsilon}\I}\bsV^\dagger
$$
for very small positive real $\varepsilon$. Apparently
$\bsT_\varepsilon\in \GL(d,\complex)$ and
$\norm{\bsA-\bsT_\varepsilon}<\varepsilon$. This indicates that
$\GL(d,\complex)$ is dense in $\End(\complex^d)$ in the norm
topology.

For any fixed $\bsA\in\End(\complex^d)$, we take
$\bsT\in\GL(d,\complex)$ such that $\norm{\bsA-\bsT}$ is very small.
Since
$$
\norm{\bQ(\bsA) - \bQ(\bsT)}\leqslant n
\Delta^{n-1}\norm{\bsA-\bsT},
$$
where $\Delta:= \max\set{\norm{\bsA},\norm{\bsT}}$, it follows, from
the fact that $\spn_\complex\set{\bsT^{\ot n}:
\bsT\in\GL(d,\complex)}$ is closed (in the finite-dimensional
setting), that \eqref{eq:exp-inv} is true. Therefore the proof is
complete.
\end{proof}

\begin{remark}
In this Remark, we will show that, for every $\bsT\in
\GL(d,\complex)$, there exists $\bsM\in\End(\complex^d)$ such that
$\bsT=e^{\bsM}$. This result is a famous one in Lie theory. A
general method for its proof is rather involved. To avoid usage of
advanced tools in Lie theory. We give here an elementarily proof of
it. We just use the matrix technique.

Indeed, it is easy to show that if $\bsT$ is a diagonalizable
matrix, then the conclusion is true. For a general case, we separate
the proof into two steps:

\emph{Case 1}. There is a sequence of diagonalizable matrices
$\bsT_k$ satisfying that
\begin{enumerate}[(i)]
\item $\lim_{k\to\infty} \norm{\bsT_k -\bsT}=0$,
\item If $\bsT_k = e^{\bsM_k}$, then there is a constant $c>0$ such that $\norm{\bsM_k}\leqslant
c$ holds for every $k$.
\end{enumerate}
Now we show that the \emph{existence} of $\bsT_k$. Consider the
Jordan canonical decomposition of $\bsT$ for $\bsT =
\bsP\bsJ\bsP^{-1}$. Let $t_j$ be the diagonal entries of $\bsJ$.
Note that $\bsT$ is an invertible matrix, so $t_j\neq0$ for every
$1\leqslant j\leqslant d$. Let
$$
\bsT_k := \bsP(\bsJ + \Lambda_k)\bsP^{-1},
$$
where $\Lambda_k : =
\mathrm{diag}(\lambda^k_1,\lambda^k_2,\ldots,\lambda^k_d)$. Then
$\bsT_k$ meets the conditions (i) and (ii) in \emph{Case 1} if
\begin{enumerate}[(a)]
\item $\lim_k \lambda^k_j = 0$ for $j=1,\ldots,d$;
\item $t_j +\lambda^k_j$ are all different when $j$ runs from 1 to $d$
for every given $k$. Thus $\bsT_k$ has $d$ different eigenvalues
$t_j+\lambda^k_j$, and of course $\bsT_k$ is diagonalizable;
\item there is a constant $c$ such that
$\abs{\ln(t_j+\lambda^k_j)}\leqslant c$ for every $k$ and $j$. Note
that if (b) is true, then $\norm{\bsM_k} = \max_j
\abs{\ln(t_j+\lambda^k_j)}$.
\end{enumerate}
The construction of $\lambda^k_j$ satisfying (a)--(c) is described
as follows: For any given $k$, let $\lambda^k_1 = \tfrac{t_1}{k}$,
and $\lambda^k_j$ be one of
$\tfrac{t_j}{k},\tfrac{t_j}{k+1},\ldots,\tfrac{t_j}{k+j}$ such that
$t_i + \lambda^k_i\neq t_j + \lambda^k_j$ whenever $i<j$. Apparently
(a) and (b) are satisfied. To check (c), we have
$$
\abs{\ln(t_j+\lambda^k_j)} =
\abs{\ln(t_j)+\ln(1+\lambda^k_j/t_j)}\leqslant \abs{\ln(t_j)} +
\abs{\ln(1+\lambda^k_j/t_j)},
$$
taking $c=\max_j \abs{\ln(t_j)}+\ln2$ is enough. That is
$\norm{\bsM_k} = \max_j \abs{\ln(t_j+\lambda^k_j)}\leqslant c$ for
all $k$.

\emph{Case 2}. When \emph{Case 1} holds, since the exponential
function is a smooth and continuous function, so the image of the
\emph{compact} set, $\exp\Pa{B(0,c)}$ must be closed, where $B(0,c)$
is the closed ball with radius $c$, thus the limit $\bsT$ of
$e^{\bsM_k}$ is also in $\exp\Pa{B(0,c)}$. This means that there
exists $\bsM\in B(0,c)$ such that $\bsT=e^{\bsM}$. The proof is
finished.
\end{remark}

We remark here that the above first proof of Schur-Weyl duality makes reference to PhD thesis of Christandl \cite{Christ2006}. The following second proof is taken from the book of Goodman and Wallach \cite{Wallach}.

\begin{proof}[Second proof]
Let $\set{\ket{1},\ldots,\ket{d}}$ be the standard basis for
$\complex^d$. For an ordered $k$-tuple $I=(i_1,\ldots,i_k)$ with
$i_1,\ldots,i_k\in [d]$, where $[d]:=\set{1,\ldots,d}$, define
$\abs{I}=k$ and $\ket{I}:=\ket{i_1\cdots i_k}$. The tensors
$\set{\ket{I}:I\in [d]^k}$, with $I$ ranging over the all such
$k$-tuples, give a basis for $(\complex^d)^{\ot k}$. The group $S_k$
permutes this basis by the action $\bP(\pi)\ket{I} = \ket{\pi\cdot
I}$, where for $I=(i_1,\ldots,i_k)$ and $\pi\in S_k$, we define
$$
\pi\cdot(i_1,\ldots,i_k) :=
(i_{\pi^{-1}(1)},\ldots,i_{\pi^{-1}(k)}).
$$
Note that $\pi$ changes the \emph{positions}(1 to $k$) of the
indices, not their values (1 to $d$), and we have $(\sigma\pi)\cdot
I = \sigma\cdot(\pi\cdot I)$ for $\sigma,\pi\in S_k$.

Suppose $\bsB\in\End((\complex^d)^{\ot k})$ has matrix $[b_{I,J}]$
relative to the basis $\set{\ket{I}: I\in[d]^k}$:
$\Innerm{I}{\bsB}{J} = b_{I,J}$ and
$$
\bsB\ket{J} = \sum_{I\in[d]^k} b_{I,J}\ket{I}.
$$
We have
\begin{eqnarray*}
\bsB\bP(\pi)\ket{J} = \bsB\ket{\pi\cdot J} = \sum_I b_{I,\pi\cdot
J}\ket{I}
\end{eqnarray*}
for $\pi\in S_k$, whereas
\begin{eqnarray*}
\bP(\pi)\bsB\ket{J} =\bP(\pi)\Pa{\sum_I\proj{I}}\bsB\ket{J} = \sum_I
b_{I,J}\ket{\pi\cdot I} = \sum_I b_{\pi^{-1}\cdot I, J}\ket{I}.
\end{eqnarray*}
Thus $\bsB\in\cA'$, i.e., $\bsB$ commutes with $\bP(\pi)$ for each
$\pi\in S_k$, if and only if $b_{I,\pi\cdot J} = b_{\pi^{-1}\cdot I,
J}$ for all multi-indices $I,J$ and all $\pi\in S_k$. Replacing $I$
by $\pi\cdot I$, we can write this condition as
\begin{eqnarray}\label{eq:orbit}
b_{\pi\cdot I,\pi\cdot J} = b_{I, J}\quad( \forall
I,J\in[d]^k;\pi\in S_k).
\end{eqnarray}
Consider the \blue{non-degenerate bilinear form}
$\Inner{\bsX}{\bsY}: =\Tr{\bsX\bsY}$ on $\End((\complex^d)^{\ot
k})$.
\begin{itemize}
\item We claim that the restriction of this form to $\cA'$ is
non-degenerate.
\end{itemize}
Indeed, we have a projection from $\End((\complex^d)^{\ot k})$ onto
$\cA'$, i.e.,
$$
\#:\End((\complex^d)^{\ot k})\longrightarrow \cA',
$$
defined by $\#:\bsX\mapsto \bsX^{\#}$ which is given by averaging
over $S_k$:
\begin{eqnarray*}
\bsX^{\#} = \frac1{k!}\sum_{\pi\in S_k} \bP(\pi)\bsX
\bP(\pi)^{-1}\in\cA'\quad(\forall \bsX\in \End((\complex^d)^{\ot
k})).
\end{eqnarray*}
If $\bsB\in\cA'$, i.e., $\bsB\bP(\pi)=\bP(\pi)\bsB$ for each $\pi\in
S_k$, then
$$
\Inner{\bsX^{\#}}{\bsB} =  \frac1{k!}\sum_{\pi\in
S_k}\Tr{\bP(\pi)\bsX \bP(\pi)^{-1}\bsB} = \Inner{\bsX}{\bsB}.
$$
Thus $\Inner{\cA'}{\bsB}=0$ means that
$$
\Inner{\bsX^{\#}}{\bsB}=0\quad (\forall \bsX\in
\End((\complex^d)^{\ot k})).
$$
Because $\Inner{\bsX}{\bsB}=\Inner{\bsX^{\#}}{\bsB}$ for all
$\bsX\in\End((\complex^d)^{\ot k})$ when $\bsB\in\cA'$, we see that
$\Inner{\bsX}{\bsB}=0$ for all $\bsX\in\End((\complex^d)^{\ot k})$,
and so $\bsB=0$. \blue{Hence the trace form on $\cA'$ is
non-degenerate.}

To show that $\cA'=\cB$, note firstly that it is trivially that
$\cB\subset\cA'$, we decompose $\cA'=\cB\oplus \cC$. It suffices to
show that $\cC=\set{0}$. That is, \blue{it thus suffices to show
that if $\bsB\in\cA'\bigcap\cC=\cC$, i.e., $\bsB$ is orthogonal to
$\cB$, then $\bsB=0$}. Now if $g=[g_{ij}]\in\GL(d,\complex)$, then
$\bQ(g)$ has matrix $g_{I,J}=g_{i_1j_1}\cdots g_{i_kj_k}$ relative
to the basis $\set{\ket{I}:I\in[d]^k}$. Thus from the assumption
$\Inner{\bsB}{\cB}=0$ (here $\Inner{\cdot}{\cdot}$ is the bilinear
form defined previously, not Hilbert-Schmidt inner product), we see
that
\begin{eqnarray*}
\Inner{\bsB}{\bQ(g)}=\sum_{I,J}b_{I,J}g_{I,J}=\sum_{I,J}b_{I,J}g_{j_1i_1}\cdots
g_{j_ki_k}=0
\end{eqnarray*}
for all $g\in\GL(d,\complex)$, where $[b_{I,J}]$ is the matrix of
$\bsB$. In what follows, we prove that $\bsB=0$. Define a linear
polynomial function $p_{\bsB}$ on $M_d(\complex)$ by
$$
p_{\bsB}(\bsX) =\Inner{\bsB}{\bQ(\bsX)}=
\sum_{I,J}b_{I,J}x_{j_1i_1}\cdots x_{j_ki_k}
$$
for $\bsX=[x_{ij}]\in M_d(\complex)$. Clearly $p_{\bsB}$ is vanished
over $\GL(d,\complex)$, a dense subset of $M_d(\complex)$; and
$p_{\bsB}$ is a continuous function on $M_d(\complex)$, therefore
$p_{\bsB}$ is vanished over the whole $M_d(\complex)$, i.e.,
$p_{\bsB}\equiv 0$, so for all $\bsX=[x_{ij}]\in M_d(\complex)$, we
have
\begin{eqnarray}\label{eq:global}
\Inner{\bsB}{\bQ(\bsX)}=\sum_{I,J}b_{I,J}x_{j_1i_1}\cdots
x_{j_ki_k}=0.
\end{eqnarray}
In what follows, we show that $b_{I,J}=0$ for all $I,J$. Namely,
$\bsB=0$. We begin by grouping the terms in the above equation
according distinct monomials in the matrix entries $\set{x_{ij}}$.
Introduce the notation $x_{I,J}=x_{i_1j_1}\cdots x_{i_kj_k}$, and
view these monomials as polynomial functions on $M_d(\complex)$. Let
$\Theta$ be the set of all ordered pairs $(I,J)$ of multi-indices
with $\abs{I}=\abs{J}=k$, i.e.,
$$
\Theta\defeq\Set{(I,J):\abs{I}=\abs{J}=k}.
$$
The group $S_k$ acts on $\Theta$ by
$$
\pi\cdot(I,J) = (\pi\cdot I,\pi\cdot J).
$$
From Eq.~\eqref{eq:orbit}, we see that $\bsB$ commutes with $S_k$ if
and only if the function $(I,J)\mapsto b_{I,J}$ is constant on the
orbits of $S_k$ in $\Theta$.

The action of $S_k$ on $\Theta$ defines an equivalence relation on
$\Theta$, where $(I,J)\sim(I',J')$ if $(I',J')=(\pi\cdot I,\pi\cdot
J)$ for some $\pi\in S_k$. This gives a decomposition of $\Theta$
into disjoint equivalence classes. \blue{Choose a set $\Gamma$ of
representatives for the equivalence classes}. Then every monomial
$x_{I,J}$ with $\abs{I}=\abs{J}=k$ can be written as $x_\gamma$ for
some $\gamma\in \Gamma$. Indeed, since the variables $x_{ij}$
mutually commute, we have
$$
x_\gamma = x_{\pi\cdot\gamma}\quad(\forall \pi\in S_k;
\gamma\in\Gamma).
$$
Suppose $x_{I,J}=x_{I',J'}$. Then there must be an integer $p$ such
that $x_{i'_1j'_1}=x_{i_p j_p}$. Call $p=1'$. Similarly, there must
be an integer $q\neq p$ such that $x_{i'_2j'_2}=x_{i_q j_q}$. Call
$q=2'$. Continuing this way, we obtain a permutation
$$
\pi: (1,2,\ldots,k)\to (1',2',\ldots,k')
$$
such that $I=\pi\cdot I'$ and $J=\pi\cdot J'$. This proves that
$\gamma$ is uniquely determined by $x_\gamma$. For
$\gamma\in\Gamma$, let $n_\gamma=\abs{S_k\cdot\gamma}$ be the
cardinality of the corresponding orbit.

Assume that the coefficients $b_{I,J}$ satisfy Eqs.~\eqref{eq:orbit}
and \eqref{eq:global}. Since $b_{I,J}=b_\gamma$ for all $(I,J)\in
S_k\cdot \gamma$, it follows from Eq.~\eqref{eq:global} that
$$
\sum_{\gamma \in\Gamma} n_\gamma b_\gamma x_\gamma = 0.
$$
Since the set of monomials $\set{x_\gamma: \gamma\in\Gamma}$ is
linearly independent, this implies that $b_{I,J}=0$ for all
$(I,J)\in\Theta$. This proves that $\bsB=0$. Hence $\cB=\cA'$.
\end{proof}

The following result concerns with a wonderful decomposition of the
representations on $k$-fold tensor space $(\complex^d)^{\ot k}$ of
$\U(d)$ and $S_k$, respectively, using their corresponding irreps
accordingly. The proof is taken from \cite{Christ2006}.
\begin{center}
\fcolorbox{purple}{lightgray}{
\parbox{16.3cm}{
\begin{thrm}[Schur-Weyl duality]\label{th:S-W-Duality}
There exist a basis, known as Schur basis, in which representation
$\Pa{\bQ\bP,(\complex^d)^{\ot k}}$ of $\mathsf{U}(d)\times S_k$
decomposes into irreducible representations $\bQ_\lambda$ and
$\bP_\lambda$ of $\mathsf{U}(d)$ and $S_k$, respectively:
\begin{enumerate}[(i)]
\item $(\complex^d)^{\ot k}\cong \bigoplus_{\lambda\vdash_dk} \bQ_\lambda\ot
\bP_\lambda$;
\item $\bP(\pi)\cong \bigoplus_{\lambda\vdash_dk} \I_{\bQ_\lambda}\ot
\bP_\lambda(\pi)$;
\item $\bQ(\bsU)\cong \bigoplus_{\lambda\vdash_dk} \bQ_\lambda(\bsU)\ot
\I_{\bP_\lambda}$,
\end{enumerate}
where the notation $\lambda\vdash_dk$ means the partition $\lambda$
of $k$ with no more than $d$ parts. Since $\bQ$ and $\bP$ commute,
we can define representation $\Pa{\bQ\bP,(\complex^d)^{\ot k}}$ of
$\mathsf{U}(d)\times S_k$ as
\begin{eqnarray}
\bQ\bP(\bsU,\pi) = \bQ(\bsU)\bP(\pi) = \bP(\pi)\bQ(\bsU)\quad
\forall \Pa{(\bsU,\pi)\in \mathsf{U}(d)\times S_k}.
\end{eqnarray}
Then:
\begin{eqnarray}\label{schur-duality}
\bQ\bP(\bsU,\pi) = \bsU^{\ot k}\bP(\pi) = \bP(\pi) \bsU^{\ot k}\cong
\bigoplus_{\lambda\vdash_dk} \bQ_\lambda(\bsU)\ot \bP_\lambda(\pi).
\end{eqnarray}
\end{thrm}}}
\end{center}

In order to prove the above theorem, we first observe that algebras
generated by $\bP$ and $\bQ$ centralize each other. Then we can
apply \emph{double commutant theorem} to get expression
Eq.~\eqref{schur-duality} only with unspecified range of $\lambda$.
In order to specify the range, we find a correspondence between
irreducible representations of $S_k$ and $\U(d)$ and partitions
$\lambda\vdash_dk$.

We call the unitary transformation performing the basis change from
standard basis to Schur basis, \emph{Schur transform} and denote by
$\bsU_{\mathrm{sch}}$. It has been shown that Schur transform can be
implemented efficiently on a quantum computer.

\begin{proof}
The application of the Duality Theorem~\ref{double-commutants} to
$G=S_k$ (and to its dual partner $\U(d)$, Theorem.~\ref{th-schur})
shows the above three equations, where $\bP_\lambda$ are irreducible
representations of $S_k$. The representation of $\U(d)$ that is
paired with $\bP_\lambda$ is denoted by $\bQ_\lambda$.

In the following, we show that $\bQ_\lambda$'s are irreducible. A
brief but elegant argument is follows: $\bQ_\lambda$ is irreducible
if and only if its extension to $\GL(d,\complex)$ is irreducible.
That is $\bQ_\lambda(\U(d))$ is irreducible if and only if
$\bQ_\lambda(\GL(d,\complex))$ is irreducible. So it suffices to
show that $\bQ_\lambda$ is indecomposable under $\GL(d,\complex)$.
By Schur's Lemma this is equivalent to showing that
$\End_{\GL(d,\complex)}(\bQ_\lambda)\cong \complex$. That is, the
maps in $\End(\bQ_\lambda)$ that commute with the action of
$\GL(d,\complex)$ are proportional to the identity.

In what follows, we show that
$\End_{\GL(d,\complex)}(\bQ_\lambda)\cong \complex$. From Schur's
Lemma, we have
$$
\End_{S_k}\Pa{(\complex^d)^{\ot k}} \cong\bigoplus_{\lambda}
\End(\bQ_\lambda)\ot\I_{\bP_\lambda}\cong \bigoplus_{\lambda}
\End(\bQ_\lambda).
$$
Thus
$$
\End_{\GL(d,\complex)\times S_k}\Pa{(\complex^d)^{\ot k}} \cong
\bigoplus_{\lambda} \End_{\GL(d,\complex)}(\bQ_\lambda).
$$
By the dual theorem, $\GL(d,\complex)$ and $S_k$ are double
commutants,
$$
\End_{S_k}\Pa{(\complex^d)^{\ot k}} = \spn_\complex\Set{\bsT^{\ot
k}: \bsT\in\GL(d,\complex)},
$$
and thus $\End_{\GL(d,\complex)\times S_k}\Pa{(\complex^d)^{\ot k}}$
is clearly contained in the center of
$\End_{S_k}\Pa{(\complex^d)^{\ot k}}$. Therefore
$\End_{\GL(d,\complex)}(\bQ_\lambda)$ is contained in the center of
$\End(\bQ_\lambda)\cong \complex$. Finally
$$
\End_{\GL(d,\complex)}(\bQ_\lambda)\cong \complex.
$$
For the proof of $\lambda\vdash_dk$, it is rather involved since we
need the notion of highest weight classification of a compact Lie
group. Recall that \blue{Highest Weight Classification} Theorem for
a connected compact Lie group $K$ with maximal torus $T$:
\begin{itemize}
\item There is a one-to-one correspondence between irreps and dominant
analytically integral weights \cite{Sepanski} given by mapping
$V_\lambda\leftrightarrow \lambda$, where $\lambda$ is dominant
analytically integral weights of $T$.
\end{itemize}
When $K=\U(d)$, all dominant analytically integral weights of its
maximal torus $T$ just corresponds to the set of partitions
$(\bQ_\lambda\leftrightarrow)\lambda\vdash_dk$. Besides, from the
side of symmetric group $S_k$, due to the fact that the number of
irreps of a finite group is equal to the number of conjugacy classes
of that group and the fact that the conjugacy classes of $S_k$
correspond to the partitions of $k$, we see that the partitions
$\lambda\vdash_dk$ parameterizes all irreps $\bP_\lambda$. In
summary, in all decompositions, $\lambda\vdash_dk$. We are done.
\end{proof}

\begin{remark}
By the Duality Theorem~\ref{double-commutants} and
Theorem~\ref{th:S-W-Duality}, it follows that $\bQ(\bsX)\in\cB$ for
$\bsX\in\End(\complex^d)$. Furthermore the decomposition of
$\bQ(\bsX)$ is of the form:
\begin{eqnarray}
\bQ(\bsX) \cong \bigoplus_{\lambda\vdash (k,d)}\bQ_\lambda(\bsX)\ot
\I_{\bP_\lambda}.
\end{eqnarray}
Therefore
\begin{eqnarray}
\bsX^{\ot k}\bP(\pi) = \bP(\pi)\bsX^{\ot k} \cong
\bigoplus_{\lambda\vdash (k,d)}\bQ_\lambda(\bsX)\ot
\bP_\lambda(\pi).
\end{eqnarray}
\end{remark}

The dimensions of pairing irreps for $\U(d)$ and $S_k$,
respectively, in Schur-Weyl duality can be computed by the so-called
\emph{hook length formulae}. The hook of box $(i, j)$ in a Young
diagram determined by a partition $\lambda$ is given by the box
itself, the boxes to its right and below. The hook length is the
number of boxes in a hook. Specifically, we have the following
result without its proof:
\begin{thrm}[Hook-length formula and hook-content formula]
The dimensions of pairing irreps for $\mathsf{U}(d)$ and $S_k$ in
Schur-Weyl duality, respectively, can be given as follows: For each
partition $\lambda\vdash_dk$, the dimensions of irreps $\bP_\lambda$
and $\bQ_\lambda$ are given by, respectively,
\begin{center}
\fcolorbox{purple}{lightgray}{
\parbox{16.3cm}{
\begin{eqnarray}
\dim(\bP_\lambda) &=& \frac{k!}{\prod_{(i,j)\in\lambda}h(i,j)},\\
\dim(\bQ_\lambda) &=& \prod_{(i,j)\in\lambda}\frac{d+j-i}{h(i,j)} =
\prod_{1\leqslant i<j\leqslant d}\frac{\lambda_i - \lambda_j +
j-i}{j-i}.
\end{eqnarray}}}
\end{center}
\end{thrm}
The formula for $\dim(\bP_\lambda)$ is called the \blue{hook-length
formula} and that for $\dim(\bQ_\lambda)$ is called the
\blue{hook-content formula}.

We will see the concrete example which is the most simple one:
\begin{exam}
Suppose that $k = 2$ and $d$ is greater than one. Then the
Schur-Weyl duality is the statement that the space of two-tensors
decomposes into symmetric and antisymmetric parts, each of which is
also an irreducible module for $\GL(d,\complex)$:
\begin{eqnarray*}
\complex^d\ot\complex^d = \vee^2\complex^d\oplus \wedge^2\complex^d,
\end{eqnarray*}
where
\begin{eqnarray*}
\vee^2\complex^d&=&\spn_\complex\Set{\bsu\ot\bsv+\bsv\ot\bsu:\bsu,\bsv\in
\complex^d},\\
\wedge^2\complex^d&=&\spn_\complex\Set{\bsu\ot\bsv-\bsv\ot\bsu:\bsu,\bsv\in
\complex^d}.
\end{eqnarray*}
Moreover $\dim(\vee^2\complex^d)=\frac{d(d+1)}2$ and
$\dim(\wedge^2\complex^d)=\frac{d(d-1)}2$. The symmetric group $S_2$
consists of two elements and has two irreducible representations,
the trivial representation and the sign representation. The trivial
representation of $S_2$ gives rise to the symmetric tensors, which
are invariant (i.e. do not change) under the permutation of the
factors, and the sign representation corresponds to the
skew-symmetric tensors, which flip the sign.
\end{exam}

\section{Matrix integrals over unitary groups}

In this section, we will give the proofs on some integrals over
unitary matrix group. We will use the uniform bi-invariant
Haar-measure $\mu$ over unitary matrix group $\U(d)$. We also use
the vec-operator correspondence.  The $\vec$ mapping is defined as
follows:
\begin{eqnarray*}
\vec(\out{i}{j}) = \ket{ij}.
\end{eqnarray*}
Thus $\vec(\I_d) = \sum^d_{j=1}\ket{jj}$. Clearly
\begin{eqnarray}
\vec(\bsA\bsX\bsB)=\bsA\ot \bsB^\t\vec(\bsX).
\end{eqnarray}
In what follows, we will employ Schur-Weyl duality to give the
computations about the integrals of the following forms:
\begin{eqnarray}
\int_{\U(d)} \bsU^{\ot k} \bsA (\bsU^{\ot k})^\dagger
\dif\mu(\bsU)~~\text{or}~~\int_{\U(d)} \bsU^{\ot k} \ot (\bsU^{\ot
k})^\dagger \dif\mu(\bsU).
\end{eqnarray}
We demonstrate the integral formulae for the special cases where
$k=1,2$ with detailed proofs since they have extremely important
applications in quantum information theory. Analogously, we also
obtain the explicit computations about the integrals of the
following forms:
\begin{eqnarray}
\int_{\U(d)} \bsU^k \bsA (\bsU^k)^\dagger
\dif\mu(\bsU)~~\text{or}~~\int_{\U(d)} \bsU^k \ot (\bsU^k)^\dagger
\dif\mu(\bsU).
\end{eqnarray}

\subsection{The case where $k=1$}

\begin{center}
\fcolorbox{purple}{lightgray}{
\parbox{16.3cm}{
\begin{prop}[Completely depolarizing channel]\label{prop:u-integral}
It holds that
\begin{eqnarray}
\int_{\mathsf{U}(d)}\bsU\bsA\bsU^\dagger \dif\mu(\bsU) =
\frac{\Tr{\bsA}}{d}\I_d,
\end{eqnarray}
where $\bsA\in M_d(\complex)$.
\end{prop}}}
\end{center}

\begin{proof}
For any $\bsV\in \U(d)$, we have
\begin{eqnarray*}
\bsV\Pa{\int_{\U(d)}\bsU\bsA\bsU^\dagger \dif\mu(\bsU)}\bsV^\dagger
&=&
\int_{\U(d)}(\bsV\bsU)\bsA(\bsV\bsU)^\dagger \dif\mu(\bsU)\\
&=& \int_{\U(d)}(\bsV\bsU)\bsA(\bsV\bsU)^\dagger \dif\mu(\bsV\bsU)\\
&=& \int_{\U(d)}\bsW\bsA\bsW^\dagger \dif\mu(\bsW) =
\int_{\U(d)}\bsU\bsA\bsU^\dagger \dif\mu(\bsU),
\end{eqnarray*}
implying that $\int_{\U(d)}\bsU\bsA\bsU^\dagger \dif\mu(\bsU)$
commutes with $\U(d)$. Thus $\int_{\U(d)}\bsU\bsA\bsU^\dagger
\dif\mu(\bsU) = \lambda_{\bsA}\I_d$. By taking trace over both
sides, we get $\lambda_{\bsA} = \frac{\Tr{\bsA}}d$. Therefore the
desired conclusion is obtained.
\end{proof}

The application of Proposition~\ref{prop:u-integral} can be found in
\cite{Frank}.

\begin{cor}\label{cor:randomized}
It holds that
\begin{eqnarray}
\int_{\mathsf{U}(d_A)}
(\bsU_A\ot\I_B)\bsX_{AB}(\bsU_A\ot\I_B)^\dagger \dif\mu(\bsU_A) =
\frac{\I_A}{d_A}\ot\Ptr{A}{\bsX_{AB}}.
\end{eqnarray}
\end{cor}

\begin{proof}
We chose an orthonormal base $\set{\ket{\mu}: \mu=1,\ldots,d_B}$ for
the second Hilbert space $B$. Then $\bsX_{AB} =
\sum_{\mu,\nu=1}^{d_B}\bsX^A_{\mu\nu}\ot\out{\mu}{\nu}$ such that
\begin{eqnarray*}
&&\int_{\U(d_A)} (\bsU_A\ot\I_B)\bsX_{AB}(\bsU_A\ot\I_B)^\dagger
\dif\mu(\bsU_A) = \sum_{\mu,\nu=1}^{d_B} \Pa{\int_{\U(d_A)}
\bsU_A\bsX^A_{\mu\nu}\bsU^\dagger_A \dif\mu(\bsU_A)}\ot\out{\mu}{\nu}\\
&&=\sum_{\mu,\nu=1}^{d_B}
\Pa{\Tr{\bsX^A_{\mu\nu}}\frac{\I_A}{d_A}}\ot\out{\mu}{\nu} =
\frac{\I_A}{d_A}\ot
\Pa{\sum_{\mu,\nu=1}^{d_B}\Tr{\bsX^A_{\mu\nu}}\out{\mu}{\nu}} =
\frac{\I_A}{d_A}\ot\Ptr{A}{\bsX_{AB}}.
\end{eqnarray*}
This completes the proof.
\end{proof}

\begin{cor}\label{cor:randomized-systems}
It holds that
\begin{eqnarray}
\int_{\mathsf{U}(d_A)}\int_{\mathsf{U}(d_B)} (\bsU_A\ot
\bsU_B)\bsX_{AB}(\bsU_A\ot \bsU_B)^\dagger
\dif\mu(\bsU_A)\dif\mu(\bsU_B) =
\Ptr{AB}{\bsX_{AB}}\frac{\I_A}{d_A}\ot\frac{\I_B}{d_B}
\end{eqnarray}
\end{cor}

\begin{cor}\label{cor:UU}
It holds that
\begin{eqnarray}
\int_{\mathsf{U}(d)}\bsU\ot\overline{\bsU} \dif\mu(\bsU) =
\frac1d\out{\vec(\I_d)}{\vec(\I_d)}.
\end{eqnarray}
\end{cor}

\begin{proof}
Since
\begin{eqnarray*}
\vec\Pa{\int_{\U(d)}\bsU\bsA\bsU^\dagger \dif\mu(\bsU)}&=&
\Pa{\int_{\U(d)}\bsU\ot\overline{\bsU} \dif\mu(\bsU)}\ket{\vec(\bsA)},\\
\vec\Pa{\frac{\Tr{\bsA}}{d}\I_d}&=&
\frac1d\ket{\vec(\I_d)}\inner{\vec(\I_d)}{\vec(\bsA)}.
\end{eqnarray*}
Using Proposition~\ref{prop:u-integral}, it follows that
\begin{eqnarray*}
\int_{\U(d)}\bsU\ot\overline{\bsU} \dif\mu(\bsU) =
\frac1d\out{\vec(\I_d)}{\vec(\I_d)},
\end{eqnarray*}
implying the result.
\end{proof}

\begin{center}
\fcolorbox{purple}{lightgray}{
\parbox{16.3cm}{
\begin{cor}\label{cor:uu-swap}
It holds that
\begin{eqnarray}
\int_{\mathsf{U}(d)} \bsU\ot \bsU^\dagger \dif\mu(\bsU) =
\frac{\bsF}{d},
\end{eqnarray}
where $\bsF$ is the swap operator defined as
$\bsF=\sum^d_{i,j=1}\out{ij}{ji}$.
\end{cor}}}
\end{center}

\begin{proof}[The first proof]
By taking partial transposes relative to second subsystems over both
sides in Corollary~\ref{cor:UU}, we get the desired identity.
\end{proof}
\begin{proof}[The second proof]
Let $\bsM=\int_{\U(d)} \bsU\ot \bsU^\dagger \dif\mu(\bsU)$. Since
Haar-measure $\mu$ is uniform over the unitary group $\U(d)$, it
follows that $\mu(\bsU)=\mu(\bsV)$ for any $\bsU,\bsV\in\U(d)$. In
particular, $\mu(\bsU)=\mu(\bsU^\dagger)$. Thus $\bsM^\dagger=\bsM$.

From the elementary fact that $\Tr{(\bsA\ot \bsB) \bsF}
=\Tr{\bsA\bsB}$, we have $\Tr{\bsM\bsF}=d$. Since Haar-measure is
left-regular, it follows that
\begin{eqnarray*}
(\bsV\ot\I)\bsM(\I\ot \bsV^\dagger) &=& \int_{\U(d)}\bsV\bsU\ot
\bsU^\dagger
\bsV^\dagger \dif\mu(\bsU) \\
&=& \int_{\U(d)}\bsV\bsU\ot (\bsV\bsU)^\dagger \dif\mu(\bsV\bsU) =
\bsM.
\end{eqnarray*}
That is $(\bsV\ot\I)\bsM(\I\ot \bsV^\dagger) = \bsM$ for all
$\bsV\in\U(d)$. By taking traces over both sides, we have
\begin{eqnarray*}
\Tr{\bsM} = \Tr{(\bsV\ot\I)\bsM(\I\ot \bsV^\dagger)} =
\Tr{\bsM(\bsV\ot \bsV^\dagger)}.
\end{eqnarray*}
By taking integrals over both sides, we have
\begin{eqnarray*}
\int_{\U(d)}\Tr{\bsM}\dif\mu(\bsV) = \int_{\U(d)}\Tr{\bsM(\bsV\ot
\bsV^\dagger)}\dif\mu(\bsV),
\end{eqnarray*}
which means that $\Tr{\bsM}=\Tr{\bsM^2}$. By Cauchy-Schwartz
inequality, we get
$$
d^2 = \Br{\Tr{\bsM\bsF}}^2\leqslant \Tr{\bsM^2}\Tr{\bsF^2} =
d^2\Tr{\bsM},
$$
implies that $\Tr{\bsM}\geqslant1$. In what follows, we show that
$\Tr{\bsM}=1$. By the definition of $M$, we have
\begin{eqnarray*}
&&\Tr{\bsM} = \int_{\U(d)} \abs{\Tr{\bsU}}^2\dif\mu(\bsU) \\
&&= \int_{\U(d)}
\inner{\vec(\I_d)}{\vec(\bsU)}\inner{\vec(\bsU)}{\vec(\I_d)}\dif\mu(\bsU)\\
&&=
\Innerm{\vec(\I_d)}{\int_{\U(d)}\out{\vec(\bsU)}{\vec(\bsU)}\dif\mu(\bsU)}{\vec(\I_d)}.
\end{eqnarray*}
Define a unital quantum channel $\Gamma$ as follows:
$$
\Gamma=\int_{\U(d)}\mathrm{Ad}_{\bsU}\dif\mu(\bsU).
$$
Thus by Proposition~\ref{prop:u-integral}, we have
$\Gamma(\bsX)=\Tr{\bsX}\frac{\I_d}{d}$. By Choi-Jiamio{\l}kowksi
isomorphism, it follows that
$$
\jam{\Gamma} = (\Gamma\ot\I)(\out{\vec(\I_d)}{\vec(\I_d)}) =
\int_{\U(d)}\out{\vec(\bsU)}{\vec(\bsU)}\dif\mu(\bsU).
$$
For the completely depolarizing channel $\Gamma(\bsX) =
\Tr{\bsX}\frac{\I_d}{d}$, we already know that $\jam{\Gamma} =
\frac1d\I_d\ot\I_d$. Therefore
\begin{center}
\fcolorbox{purple}{lightgray}{
\parbox{16.3cm}{
\begin{eqnarray}
\int_{\U(d)}\out{\vec(\bsU)}{\vec(\bsU)}\dif\mu(\bsU) =
\frac1d\I_d\ot\I_d.
\end{eqnarray}}}
\end{center}
Finally $\Tr{\bsM}=\frac1d\inner{\vec(\I_d)}{\vec(\I_d)} = 1$. This
indicates that Cauchy-Schwartz inequality is saturated, and moreover
the saturation happens if and only if $\bsM\propto \bsF$. Let
$\bsM=\lambda \bsF$. By taking traces over both sides, we have
$\lambda=\frac1d$. The desired conclusion is obtained.
\end{proof}

\begin{proof}[The third proof]
We derive directly the integral formula from the Schur Orthogonality Relations of a compact Lie group. See the Section~\ref{sect:concluding-remarks}.
\end{proof}

\begin{cor}\label{cor:integral-of-visibility}
It holds that
\begin{eqnarray}
\int_{\mathsf{U}(d)} \abs{\Tr{\bsA\bsU}}^2 \dif\mu(\bsU) =
\frac1d\tr{\bsA^\dagger \bsA},
\end{eqnarray}
where $\bsA\in M_d(\complex)$.
\end{cor}

\begin{proof}
In fact,
$$
\abs{\Tr{\bsA\bsU}}^2 = \Tr{\bsA\bsU}\overline{\Tr{\bsA\bsU}} =
\Tr{(\bsA\ot \bsA^\dagger)(\bsU\ot \bsU^\dagger)}.
$$
It follows that
\begin{eqnarray*}
\int_{\U(d)} \abs{\Tr{\bsA\bsU}}^2 \dif\mu(\bsU) &=& \Tr{(\bsA\ot
\bsA^\dagger)\int_{\U(d)} \bsU\ot \bsU^\dagger \dif\mu(\bsU)}\\
&=& \frac1d\Tr{(\bsA\ot \bsA^\dagger)\bsF} =
\frac1d\Tr{\bsA\bsA^\dagger},
\end{eqnarray*}
implying the result.
\end{proof}

\begin{cor}\label{cor:local-integral}
It holds that
\begin{eqnarray}
\int_{\mathsf{U}(d_1)}\int_{\mathsf{U}(d_2)}
\abs{\Tr{\bsA(\bsU\ot\bsV)}}^2\dif\mu(\bsU)\dif\mu(\bsV) =
\frac1{d_1d_2}\Tr{\bsA^\dagger \bsA},
\end{eqnarray}
where $\bsA\in M_{d_1d_2}(\complex)$.
\end{cor}

\begin{proof}
By the SVD of a matrix, we have
\begin{eqnarray*}
\bsA = \sum_j s_j \out{\Phi_j}{\Psi_j},
\end{eqnarray*}
where $s_j := s_j(\bsA)$ is the singular values of the matrix $\bsA$
and $\ket{\Phi_j},\ket{\Psi_j}\in \complex^{d_1}\ot \complex^{d_2}$.
From the properties of the vec mapping for a matrix, we see that
there exist $d_2\times d_1$ matrices $\bsX_j$ and $\bsY_j$,
respectively, such that
\begin{eqnarray*}
\ket{\Phi_j} = \vec(\bsX_j),~~\ket{\Psi} = \vec(\bsY_j).
\end{eqnarray*}
This indicates that
\begin{eqnarray*}
&&\abs{\Tr{\bsA(\bsU\ot \bsV)}}^2 = \abs{\sum_j
s_j\Innerm{\Psi_j}{\bsU\ot
\bsV}{\Phi_j}}^2\\
&&= \sum_{i,j} s_is_j\Innerm{\Psi_i}{\bsU\ot
\bsV}{\Phi_i}\overline{\Innerm{\Psi_j}{\bsU\ot \bsV}{\Phi_j}}\\
&&= \sum_{i,j} s_is_j\Innerm{\Psi_i}{\bsU\ot \bsV}{\Phi_i}
\Innerm{\Phi_j}{\bsU^\dagger\ot \bsV^\dagger}{\Psi_j},
\end{eqnarray*}
which implies that
\begin{eqnarray*}
&&\abs{\Tr{\bsA(\bsU\ot \bsV)}}^2 = \sum_{i,j} s_is_j
\Inner{\bsY_i}{\bsU\bsX_i
\bsV^\t}\Inner{\bsX_j}{\bsU^\dagger \bsY_j (\bsV^\dagger)^\t}\\
&&= \sum_{i,j} s_is_j \Inner{\bsY_i\ot \bsX_j}{(\bsU\ot
\bsU^\dagger) (\bsX_i\ot \bsY_j)(\bsV^\t\ot (\bsV^\t)^\dagger)}.
\end{eqnarray*}
Thus
\begin{eqnarray*}
&&\int_{\U(d_1)}\int_{\U(d_2)} \abs{\Tr{\bsA(\bsU\ot
\bsV)}}^2\dif\mu(\bsU)\dif\mu(\bsV)\\
&&=\sum_{i,j} s_is_j \Inner{\bsY_i\ot
\bsX_j}{\Pa{\int_{\U(d_1)}\bsU\ot \bsU^\dagger \dif\mu(\bsU)}
(\bsX_i\ot \bsY_j)\Pa{\int_{\U(d_2)} \bsV^\t\ot
(\bsV^\t)^\dagger \dif\mu(\bsV)}}\\
&&= \frac1{d_1d_2} \sum_{i,j} s_is_j \Inner{\bsY_i\ot
\bsX_j}{\bsF_{11}
(\bsX_i\ot \bsY_j)\bsF_{22}}\\
&&=\frac1{d_1d_2} \sum_{i,j} s_is_j \Tr{(\bsY_i\ot \bsX_j)^\dagger
\bsF_{11} (\bsX_i\ot \bsY_j)\bsF_{22}},
\end{eqnarray*}
where $\bsF_{11}$ is the swap operator on
$\complex^{d_1}\ot\complex^{d_1}$, $\bsF_{22}$ is the swap operator
on $\complex^{d_2}\ot\complex^{d_2}$.

Taking orthonormal base $\ket{\mu}$ and $\ket{m}$ of
$\complex^{d_1}$ and $\complex^{d_2}$, respectively, gives rise to
$$
\bsF_{11} = \sum^{d_1}_{\mu,\nu=1} \out{\mu\nu}{\nu\mu},~~\bsF_{22}
= \sum^{d_2}_{m,n=1}\out{mn}{nm}.
$$
By substituting both operators into the above expression, it follows
that
\begin{eqnarray*}
\int_{\U(d_1)}\int_{\U(d_2)} \abs{\Tr{\bsA(\bsU\ot
\bsV)}}^2\dif\mu(\bsU)\dif\mu(\bsV) &=& \frac1{d_1d_2} \sum_{i,j}
s_is_j \Tr{\bsX_i
\bsX^\dagger_j}\Tr{\bsY^\dagger_i \bsY_j}\\
&=& \frac1{d_1d_2}\Tr{\bsA^\dagger \bsA}.
\end{eqnarray*}
The proof is complete.
\end{proof}
Note that Corollaries~\ref{cor:uu-swap},
\ref{cor:integral-of-visibility}, and \ref{cor:local-integral} are
used in the recent paper \cite{Lin} to establish an interesting
relationship between quantum correlation and interference
visibility. In what follows, we obtain a general result:
\begin{prop}
It holds that
\begin{eqnarray}
\int_{\mathsf{U}(d_1)\times\cdots\times\mathsf{U}(d_n)}
\abs{\Tr{\bsA(\bsU_1\ot \cdots\ot \bsU_n)}}^2\dif\mu(\bsU_1)\cdots
\dif\mu(\bsU_n) = \frac1d\Tr{\bsA^\dagger \bsA},
\end{eqnarray}
where $\bsA\in M_d(\complex)$ for $d=\prod^n_{j=1}d_j$.
\end{prop}
This result will be useful in the investigation of multipartite
quantum correlation. The detail of its proof is as follows.

\begin{proof}
Firstly we note from Corollary~\ref{cor:randomized} that
\begin{eqnarray*}
\int_{\U(d_1)} (\bsU_1\ot\I_{2\ldots n})\bsX_{12\ldots
n}(\bsU_1\ot\I_{2\ldots n})^\dagger \dif\mu(\bsU_1) =
\frac{\I_1}{d_1}\ot\Ptr{1}{\bsX_{12\ldots n}}.
\end{eqnarray*}
Furthermore, we have
\begin{eqnarray*}
\int_{\U(d_1)}\int_{\U(d_2)} (\bsU_1\ot \bsU_2\ot\I_{3\ldots
n})\bsX_{12\ldots n}(\bsU_1\ot \bsU_2\ot\I_{3\ldots n})^\dagger
\dif\mu(\bsU_1)\dif\mu(\bsU_2) =
\frac{\I_1}{d_1}\ot\frac{\I_2}{d_2}\ot\Ptr{12}{\bsX_{12\ldots n}}.
\end{eqnarray*}
By induction, we have
\begin{eqnarray*}
&&\int_{\U(d_1)}\int_{\U(d_2)}\cdots \int_{\U(d_n)}(\bsU_1\ot
\bsU_2\ot\cdots \ot \bsU_n)\bsX_{12\ldots n}(\bsU_1\ot
\bsU_2\ot\cdots \ot \bsU_n)^\dagger
\dif\mu(\bsU_1)\dif\mu(\bsU_2)\cdots
\dif\mu(\bsU_n) \nonumber\\
&&= \Ptr{12\ldots n}{\bsX_{12\ldots
n}}\frac{\I_1}{d_1}\ot\frac{\I_2}{d_2}\ot\cdots\ot\frac{\I_n}{d_n}.
\end{eqnarray*}
This implies that
\begin{eqnarray*}
&&\int_{\U(d_1)}\int_{\U(d_2)}\cdots \int_{\U(d_n)}\out{\bsU_1\ot
\bsU_2\ot\cdots \ot \bsU_n}{\bsU_1\ot
\bsU_2\ot\cdots \ot \bsU_n} \dif\mu(\bsU_1)\dif\mu(\bsU_2)\cdots \dif\mu(\bsU_n)\\
&&=\frac1d \I_{12\ldots n}\ot\I_{12\ldots n},
\end{eqnarray*}
where $d=\prod^n_{j=1}d_j$. Now
$$
\abs{\Tr{\bsA(\bsU_1\ot \bsU_2\ot\cdots\ot \bsU_n)}}^2 =
\Inner{\bsA^\dagger}{\bsU_1\ot \bsU_2\ot\cdots\ot
\bsU_n}\Inner{\bsU_1\ot \bsU_2\ot\cdots\ot \bsU_n}{\bsA^\dagger}
$$
implying
\begin{eqnarray*}
&&\int_{\U(d_1)}\cdots\int_{\U(d_n)}\abs{\Tr{\bsA(\bsU_1\ot
\bsU_2\ot\cdots\ot \bsU_n)}}^2 \\
&&= \Innerm{\bsA^\dagger}{\int_{\U(d_1)}\cdots
\int_{\U(d_n)}\out{\bsU_1\ot \bsU_2\ot\cdots \ot \bsU_n}{\bsU_1\ot
\bsU_2\ot\cdots \ot \bsU_n} \dif\mu(\bsU_1)\cdots \dif\mu(\bsU_n)}{\bsA^\dagger}\\
&&= \frac1d\inner{\bsA^\dagger}{\bsA^\dagger} =
\frac1d\Tr{\bsA^\dagger \bsA}.
\end{eqnarray*}
We are done.
\end{proof}

\subsection{The case where $k=2$}

\begin{center}
\fcolorbox{purple}{lightgray}{
\parbox{16.3cm}{
\begin{prop}\label{prop:uu-integral}
It holds that
\begin{eqnarray}
&&\int_{\mathsf{U}(d)}(\bsU\ot \bsU)\bsA(\bsU\ot \bsU)^\dagger
\dif\mu(\bsU) \notag\\
&&= \Pa{\frac{\Tr{\bsA}}{d^2-1} -
\frac{\Tr{\bsA\bsF}}{d(d^2-1)}}\I_{d^2} -
\Pa{\frac{\Tr{\bsA}}{d(d^2-1)}- \frac{\Tr{\bsA\bsF}}{d^2-1}}\bsF,
\end{eqnarray}
where $\bsA\in M_{d^2}(\complex)$ and the swap operator $\bsF$ is
defined by $\bsF\ket{ij}=\ket{ji}$ for all $i,j=1,\ldots,d$.
\end{prop}}}
\end{center}

\begin{proof}
Analogously, we have $\int_{\U(d)}(\bsU\ot \bsU)\bsA(\bsU\ot
\bsU)^\dagger \dif\mu(\bsU)$ commutes with $\set{\bsV\ot \bsV:
\bsV\in\U(d)}$. Denote $\bsP_\wedge := \tfrac12(\I_{d^2}-\bsF)$ and
$\bsP_\vee := \tfrac12(\I_{d^2}+\bsF)$. It is easy to see that
$\Tr{\bsP_\wedge} = \tfrac12(d^2-d)$ and $\Tr{\bsP_\vee} =
\tfrac12(d^2+d)$. Since $\bsF=\sum_{i,j}\out{ij}{ji}$, it follows
that $\bsF^\dagger = \bsF$ and $\bsF^2=\I_{d^2}$. Thus both
$\bsP_\wedge$ and $\bsP_\vee$ are projectors and $\bsP_\wedge +
\bsP_\vee=\I_{d^2}$.

Because $\complex^d\ot\complex^d = \wedge^2\complex^d\oplus
\vee^2\complex^d$, we have $\bsP_\wedge(\complex^d\ot\complex^d)
\bsP_\wedge= \wedge^2\complex^d$ and
$\bsP_\vee(\complex^d\ot\complex^d) \bsP_\vee= \vee^2\complex^d$.
Besides, for any $\bsV\in \U(d)$,
$$
\bsV\ot \bsV \defeq \Br{\begin{array}{cc}
               \bsP_\wedge(\bsV\ot \bsV)\bsP_\wedge & 0 \\
               0 & \bsP_\vee(\bsV\ot \bsV)\bsP_\vee
             \end{array}
}
$$
Now write
$$
\int_{\U(d)}(\bsU\ot \bsU)\bsA(\bsU\ot \bsU)^\dagger \dif\mu(\bsU) =
\Br{\begin{array}{cc}
               \bsM_{11} & \bsM_{12} \\
               \bsM_{21} & \bsM_{22}
             \end{array}}
$$
is a block matrix, where $\bsM_{11}\in \End(\wedge^2
\complex^d),\bsM_{22}\in \End(\vee^2 \complex^d)$ and
$$
\bsM_{12}\in\Hom_{\U(d)}(\vee^2 \complex^d,\wedge^2
\complex^d),\bsM_{21}\in\Hom_{\unitary{d}}(\wedge^2
\complex^d,\vee^2 \complex^d).
$$
Thus
\begin{eqnarray*}
&&\Br{\begin{array}{cc}
               \bsM_{11} & \bsM_{12} \\
               \bsM_{21} & \bsM_{22}
             \end{array}}\Br{\begin{array}{cc}
               \bsP_\wedge(\bsV\ot \bsV)\bsP_\wedge & 0 \\
               0 & \bsP_\vee(\bsV\ot \bsV)\bsP_\vee
             \end{array}} \\
             &&=\Br{\begin{array}{cc}
               \bsP_\wedge(\bsV\ot \bsV)\bsP_\wedge & 0 \\
               0 & \bsP_\vee(\bsV\ot \bsV)\bsP_\vee
             \end{array}} \Br{\begin{array}{cc}
               \bsM_{11} & \bsM_{12} \\
               \bsM_{21} & \bsM_{22}
             \end{array}}.
\end{eqnarray*}
We get that, for all $\bsV\in\U(d)$,
\begin{eqnarray*}
\begin{cases}
\bsM_{11}(\wedge^2\bsV) &= (\wedge^2\bsV) \bsM_{11},\\
\bsM_{22}(\vee^2\bsV) &= (\vee^2\bsV) \bsM_{22},\\
\bsM_{12} (\vee^2\bsV) &= (\wedge^2\bsV) \bsM_{12},\\
\bsM_{21} (\wedge^2\bsV) &= (\vee^2\bsV) \bsM_{21}.
\end{cases}
\end{eqnarray*}
Therefore we obtained that
$$
\bsM_{11} = \lambda(\bsA) \bsP_\wedge,~~\bsM_{22} = \mu(\bsA)
\bsP_\vee,~~ \bsM_{12} = 0,~~ \bsM_{21} = 0.
$$
That is
\begin{eqnarray}\label{eq-u-design}
\int_{\U(d)}(\bsU\ot \bsU)\bsA(\bsU\ot \bsU)^\dagger \dif\mu(\bsU) =
\Br{\begin{array}{cc}
                                                            \lambda(\bsA) \bsP_\wedge & 0 \\
                                                            0 &
                                                            \mu(\bsA)
                                                            \bsP_\vee
                                                          \end{array}
} = \lambda(\bsA) \bsP_\wedge + \mu(\bsA) \bsP_\vee.
\end{eqnarray}
If $\bsA=\I_{d^2}$ in Eq.~\eqref{eq-u-design}, then $\I_{d^2} =
\lambda(\I_{d^2}) \bsP_\wedge + \mu(\I_{d^2}) \bsP_\vee$. Thus
$\lambda(\I_{d^2})= \mu(\I_{d^2}) = 1$ since $\I_{d^2} = \bsP_\wedge
+ \bsP_\vee$ and $\bsP_\wedge \bot \bsP_\vee$.

If $\bsA = \bsP_\wedge$ in Eq.~\eqref{eq-u-design}, then
$\bsP_\wedge = \lambda(\bsP_\wedge) \bsP_\wedge + \mu(\bsP_\wedge)
\bsP_\vee$ since $\bsU\ot \bsU$ commutes with $\bsP_\wedge$. Thus
$\lambda(\bsP_\wedge) = 1$ and $\mu(\bsP_\wedge)= 0$. Note that
$\lambda(\bsA),\mu(\bsA)$ are two linear functional. Thus we have:
$\lambda(\bsF) = -1$ and $\mu(\bsF) = 1$. This indicates that
$$
\int_{\U(d)}(\bsU\ot \bsU)\bsF(\bsU\ot \bsU)^\dagger \dif\mu(\bsU) =
\lambda(\bsF) \bsP_\wedge + \mu(\bsF) \bsP_\vee = \bsP_\vee -
\bsP_\wedge = \bsF.
$$
More simpler approach to this identity can be described as follows:
Since $\bsF(\bsM\ot \bsN)\bsF = \bsN\ot \bsM$, it follows that
$\bsF(\bsM\ot \bsN) = (\bsN\ot \bsM)\bsF$. Thus
\begin{eqnarray*}
\int_{\U(d)}(\bsU\ot \bsU)\bsF(\bsU\ot \bsU)^\dagger \dif\mu(\bsU)
&=&
\int_{\U(d)}\bsF(\bsU\ot \bsU)(\bsU\ot \bsU)^\dagger \dif\mu(\bsU)\\
&=& \bsF\int_{\U(d)}\dif\mu(\bsU) = \bsF = \bsP_\vee - \bsP_\wedge.
\end{eqnarray*}
Apparently
\begin{eqnarray*}
\int_{\U(d)}(\bsU\ot \bsU)^\dagger \bsF(\bsU\ot \bsU) \dif\mu(\bsU)
= \bsF = \bsP_\vee - \bsP_\wedge.
\end{eqnarray*}
By taking trace over both sides above, we get
$$
\Tr{\bsA} = \lambda(\bsA)\Tr{\bsP_\wedge} + \mu(\bsA)\Tr{\bsP_\vee}.
$$
Now by multiplying $\bsF$ on both sides in Eq.~\eqref{eq-u-design}
and then taking trace again, we get
\begin{eqnarray*}
&&\int_{\U(d)}\Tr{(\bsU\ot \bsU)^\dagger \bsF (\bsU\ot \bsU)\bsA}\dif\mu(\bsU) \\
&&=
\lambda(\bsA) \Tr{\bsP_\wedge \bsF} + \mu(\bsA) \Tr{\bsP_\vee \bsF}\\
&&= \mu(\bsA) \Tr{\bsP_\vee} -  \lambda(\bsA) \Tr{\bsP_\wedge},
\end{eqnarray*}
where we used the fact that $\bsP_\wedge \bsF = -\bsP_\wedge$ and
$\bsP_\vee \bsF = \bsP_\vee$. Thus we have
\begin{eqnarray*}
\begin{cases}
\frac{d(d-1)}2\lambda(\bsA) + \frac{d(d+1)}2\mu(\bsA) &= \Tr{\bsA},\\
\frac{d(d+1)}2\mu(\bsA) - \frac{d(d-1)}2\lambda(\bsA) &=
\Tr{\bsA\bsF}.
\end{cases}
\end{eqnarray*}
Solving this group of two binary equations gives rise to
\begin{eqnarray*}
\begin{cases}
\lambda(\bsA) =& \frac{\Tr{\bsA} - \Tr{\bsA\bsF}}{d(d-1)},\\
\mu(\bsA) =& \frac{\Tr{\bsA} + \Tr{\bsA\bsF}}{d(d+1)}.
\end{cases}
\end{eqnarray*}
Finally we obtained the desired conclusion as follows:
\begin{eqnarray*}
\int_{\U(d)}(\bsU\ot \bsU)\bsA(\bsU\ot \bsU)^\dagger \dif\mu(\bsU) =
\frac{\Tr{\bsA} - \Tr{\bsA\bsF}}{d(d-1)} \bsP_\wedge +
\frac{\Tr{\bsA} + \Tr{\bsA\bsF}}{d(d+1)} \bsP_\vee.
\end{eqnarray*}
We are done.
\end{proof}

The applications of Proposition~\ref{prop:uu-integral} in quantum
information theory can be found in \cite{Roga,Dupuis}.
\begin{center}
\fcolorbox{purple}{lightgray}{
\parbox{16.3cm}{
\begin{cor}\label{cor:super-operator}
It holds that
\begin{eqnarray}
&&\int_{\mathsf{U}(d)} \bsU^\dagger \bsA\bsU\bsX\bsU^\dagger
\bsB\bsU \dif\mu(\bsU)\notag\\
&&= \frac{d\Tr{\bsA\bsB}-\Tr{\bsA}\Tr{\bsB}}{d(d^2-1)}\Tr{\bsX}\I_d
+ \frac{d\Tr{\bsA}\Tr{\bsB}-\Tr{\bsA\bsB}}{d(d^2-1)}\bsX.
\end{eqnarray}
\end{cor}}}
\end{center}

\begin{proof}
It suffices to compute the integral $\int_{\U(d)} (\bsU^\dagger
\bsA\bsU) \ot (\bsU^\dagger \bsB\bsU) \dif\mu(\bsU)$ since
\begin{eqnarray}\label{eq:map}
\vec\Pa{\int_{\U(d)} (\bsU^\dagger \bsA\bsU)\bsX(\bsU^\dagger
\bsB\bsU) \dif\mu(\bsU)} = \int_{\U(d)} (\bsU^\dagger \bsA\bsU) \ot
(\bsU^\dagger \bsB\bsU)^\t \dif\mu(\bsU) \vec(\bsX).
\end{eqnarray}
Once we get the formula for $\int_{\U(d)} (\bsU^\dagger \bsA\bsU)
\ot (\bsU^\dagger \bsB\bsU) \dif\mu(\bsU)$, taking partial transpose
relative to the second factor in the tensor product, we get the
formula for $\int_{\U(d)} (\bsU^\dagger \bsA\bsU) \ot (\bsU^\dagger
\bsB\bsU)^\t \dif\mu(\bsU)$.

Now by Proposition~\ref{prop:uu-integral}, we have
\begin{eqnarray*}
&&\int_{\U(d)} (\bsU^\dagger \bsA\bsU) \ot (\bsU^\dagger \bsB\bsU)
\dif\mu(\bsU)=
\int_{\U(d)} (\bsU \ot \bsU)^\dagger(\bsA\ot \bsB)(\bsU\ot \bsU) \dif\mu(\bsU)\\
&&=\Pa{\frac{\Tr{\bsA}\Tr{\bsB}}{d^2-1} -
\frac{\Tr{\bsA\bsB}}{d(d^2-1)}}\I_{d^2} -
\Pa{\frac{\Tr{\bsA}\Tr{\bsB}}{d(d^2-1)}- \frac{\Tr{\bsA\bsB}}{d^2-1}}\bsF\\
&&= \frac{d\Tr{\bsA}\Tr{\bsB} -\Tr{\bsA\bsB}}{d(d^2-1)}\I_{d^2} +
\frac{d\Tr{\bsA\bsB}-\Tr{\bsA}\Tr{\bsB}}{d(d^2-1)}\bsF,
\end{eqnarray*}
implying that
\begin{eqnarray*}
&&\int_{\U(d)} (\bsU^\dagger \bsA\bsU) \ot (\bsU^\dagger \bsB\bsU)^\t \dif\mu(\bsU)\\
&&= \frac{d\Tr{\bsA}\Tr{\bsB} -\Tr{\bsA\bsB}}{d(d^2-1)}\I_{d^2} +
\frac{d\Tr{\bsA\bsB}-\Tr{\bsA}\Tr{\bsB}}{d(d^2-1)}\out{\vec(\I_d)}{\vec(\I_d)}.
\end{eqnarray*}
Substituting this identity into \eqref{eq:map} gives the desired
result.
\end{proof}
Recall that a super-operator $\Phi$ is \emph{unitarily invariant} if
$\Ad_{\bsU^\dagger}\circ\Phi\circ\Ad_{\bsU}= \Phi$ for all
$\bsU\in\U(d)$. We also note that an super-operator $\Phi$ on
$\End(\cH_d)$ can be represented as
\begin{eqnarray}
\Phi(\bsX) = \sum_j \bsA_j \bsX \bsB^\dagger_j.
\end{eqnarray}
Now we may give the specific form of any unitarily invariant
super-operator in the following corollary.
\begin{cor}\label{cor:u-invariant-super-operator}
Let $\Phi$ be a unitarily invariant super-operator on $\End(\cH_d)$.
Then
\begin{eqnarray}
\Phi(\bsX) =
\frac{d\Tr{\Phi(\I_d)}-\Tr{\Phi}}{d(d^2-1)}\Tr{\bsX}\I_d +
\frac{d\Tr{\Phi}-\Tr{\Phi(\I_d)}}{d(d^2-1)}\bsX,
\end{eqnarray}
where $\Tr{\Phi}$ is the trace of super-operator $\Phi$, defined by
$\Tr{\Phi} :=
\sum_{\mu,\nu}\Innerm{\mu}{\Phi(\out{\mu}{\nu})}{\nu}$.
\end{cor}

\begin{proof}
Apparently $\Ad_{\bsU^\dagger}\circ\Phi\circ\Ad_{\bsU}= \Phi$ for
all $\bsU\in\U(d)$. This implies that, for the uniform Haar measure
$\dif\mu(\bsU)$ over the unitary group,
\begin{eqnarray*}
\Phi(\bsX) &=& \int_{\U(d)}\Phi(\bsX)\dif\mu(\bsU) =
\int_{\U(d)}\bsU^\dagger\Phi(\bsU\bsX\bsU^\dagger)\bsU\dif\mu(\bsU)\\
&=& \sum_j \int_{\U(d)}\bsU^\dagger \bsA_j\bsU\bsX\bsU^\dagger
\bsB^\dagger_j\bsU\dif\mu(\bsU).
\end{eqnarray*}
By Corollary~\ref{cor:super-operator},
\begin{eqnarray*}
&&\int_{\U(d)}\bsU^\dagger \bsA_j\bsU\bsX\bsU^\dagger \bsB^\dagger_j\bsU\dif\mu(\bsU) \\
&&=
\frac{d\Tr{\bsA_j\bsB^\dagger_j}-\Tr{\bsA_j}\Tr{\bsB^\dagger_j}}{d(d^2-1)}\Tr{\bsX}\I_d
+
\frac{d\Tr{\bsA_j}\Tr{\bsB^\dagger_j}-\Tr{\bsA_j\bsB^\dagger_j}}{d(d^2-1)}\bsX.
\end{eqnarray*}
Thus
\begin{eqnarray*}
\Phi(\bsX)&=&
\frac{d\sum_j\Tr{\bsA_j\bsB^\dagger_j}-\sum_j\Tr{\bsA_j}\Tr{\bsB^\dagger_j}}{d(d^2-1)}\Tr{\bsX}\I_d
\\
&&+
\frac{d\sum_j\Tr{\bsA_j}\Tr{\bsB^\dagger_j}-\sum_j\Tr{\bsA_j\bsB^\dagger_j}}{d(d^2-1)}\bsX\\
&=&\frac{d\Tr{\Phi(\I_d)}-\Tr{\Phi}}{d(d^2-1)}\Tr{\bsX}\I_d +
\frac{d\Tr{\Phi}-\Tr{\Phi(\I_d)}}{d(d^2-1)}\bsX,
\end{eqnarray*}
where we have used the fact that
\begin{eqnarray*}
\Tr{\Phi} &=& \sum_{\mu,\nu}
\Inner{\out{\mu}{\nu}}{\Phi(\out{\mu}{\nu})} = \sum_{\mu,\nu}
\Innerm{\mu}{\Phi(\out{\mu}{\nu})}{\nu}\\
&=& \sum_j\sum_{\mu,\nu}
\Innerm{\mu}{\bsA_j\out{\mu}{\nu}\bsB^\dagger_j}{\nu} = \sum_j
\Pa{\sum_\mu \Innerm{\mu}{\bsA_j}{\mu}}\Pa{\sum_\nu
\Innerm{\nu}{\bsB^\dagger_j}{\nu}}\\
&=& \sum_j\Tr{\bsA_j}\Tr{\bsB^\dagger_j}.
\end{eqnarray*}
There is a caution that the trace of super-operator $\Phi$ is
different from the trace of operator $\Phi(\I_d)$.
\end{proof}
We can simplify this expression if we assume more structure on the
super-operator. A trace-preserving unitarily invariant quantum
operation $\Lambda$ is a \blue{depolarizing channel}: for
$\rho\in\density{\cH_d}$,
\begin{eqnarray}
\Lambda(\rho) = p\rho + (1-p)\frac{\I_d}{d},~~~\Pa{p=\frac{\Tr{\Phi}
- 1}{d^2-1}}.
\end{eqnarray}
Indeed, this easily follows from the facts that $\Tr{\Phi(\I_d)}=d$
and $\Tr{\rho}=1$.

Let $\Phi$ be a super-operator on $\End(\cH_d)$. Define the
\emph{twirled} super-operator
\begin{eqnarray}
\Phi_\T = \int_{\U(d)} \Ad_{\bsU^\dagger}\circ\Phi\circ\Ad_{\bsU}
\dif\mu(\bsU).
\end{eqnarray}
Clearly twirled super-operator $\Phi_\T$ is unitarily invariant.

\begin{remark}
From the proof of Corollary~\ref{cor:u-invariant-super-operator}, we
see that for any super-operator $\Phi\in\End(\cH_d)$,
\begin{eqnarray}
\int_{\U(d)}\bsU^\dagger\Phi(\bsU\bsX\bsU^\dagger)\bsU\dif\mu(\bsU)
= \frac{d\Tr{\Phi(\I_d)}-\Tr{\Phi}}{d(d^2-1)}\Tr{\bsX}\I_d +
\frac{d\Tr{\Phi}-\Tr{\Phi(\I_d)}}{d(d^2-1)}\bsX.
\end{eqnarray}
Now let $d=d_Ad_B$ and $\cH_d=\cH_A\ot\cH_B$ with $\dim(\cH_A)=d_A$
and $\dim(\cH_B)=d_B$. Assume that $\bsX = \rho_{AB}$, a density
matrix on $\cH_A\ot\cH_B$. Fixing an orthonormal basis
$\set{\ket{\psi_{B,j}}:j=1,\ldots,d_B}$ for $\cH_B$. Suppose that
$\Phi(\bsX) = \ptr{B}{\bsX}\ot\I_B$. Then it can be rewritten as:
$$
\Phi(\bsX) =
\sum^{d_B}_{i,j=1}(\I_A\ot\out{\psi_{B,i}}{\psi_{B,j}})\bsX(\I_A\ot\out{\psi_{B,j}}{\psi_{B,i}}).
$$
Clearly $\Phi(\I_A\ot\I_B) = d_B\I_A\ot\I_B$, implying that
$$
\Tr{\Phi(\I_A\ot\I_B)} = d_Ad^2_B~~\text{and}~~\Tr{\Phi} =
\sum^{d_B}_{i,j=1}(d_A\delta_{ij})^2= d^2_Ad_B.
$$
From the above discussion, we see that
\begin{eqnarray}
\int_{\U(d)}\bsU^\dagger\Phi(\bsU\rho_{AB}\bsU^\dagger)\bsU\dif\mu(\bsU)
= \frac{dd_B-d_A}{d^2-1}\I_A\ot\I_B +
\frac{dd_A-d_B}{d^2-1}\rho_{AB}.
\end{eqnarray}
Denote $\rho'_{AB} = \bsU\rho_{AB}\bsU^\dagger$ and
$\rho'_{A}=\ptr{B}{\rho'_{AB}}$. Then
\begin{eqnarray*}
\Tr{(\rho'_A)^2} &=& \Tr{(\rho'_A\ot\I_B)\rho'_{AB}} =
\Tr{\Phi(\rho'_{AB})\rho'_{AB}}\\
&=&\Tr{\bsU^\dagger\Phi(\bsU\rho_{AB}\bsU^\dagger)\bsU\rho_{AB}}.
\end{eqnarray*}
Therefore
\begin{eqnarray*}
\left\langle\Tr{(\rho'_A)^2}\right\rangle :=
\int\Tr{(\rho'_A)^2}\dif\mu(\bsU) = \Tr{\int
\bsU^\dagger\Phi(\bsU\rho_{AB}\bsU^\dagger)\bsU\dif\mu(\bsU)\rho_{AB}}.
\end{eqnarray*}
That is,
\begin{center}
\fcolorbox{purple}{lightgray}{
\parbox{16.3cm}{
\begin{eqnarray}
\left\langle\Tr{(\rho'_A)^2}\right\rangle = \frac{dd_B-d_A}{d^2-1}+
\frac{dd_A-d_B}{d^2-1}\Tr{\rho^2_{AB}}.
\end{eqnarray}}}
\end{center}
In particular, if $\rho_{AB}$ is a bipartite pure state, then
$\Tr{\rho^2_{AB}}=1$, and
\begin{eqnarray}
\left\langle\Tr{(\rho'_A)^2}\right\rangle=\frac{d_A+d_B}{d_Ad_B+1}.
\end{eqnarray}
\end{remark}

\begin{cor}
Let $\bsX,\bsY\in\End(\complex^d)$. Then the uniform average of
$\Innerm{\psi}{\bsX}{\psi}\Innerm{\psi}{\bsY}{\psi}$ over pure state
vectors $\ket{\psi}\in\complex^d$ is given by
\begin{eqnarray}
\int_{\complex^d} \Innerm{\psi}{\bsX}{\psi}\Innerm{\psi}{\bsY}{\psi}
\dif\mu(\psi) = \frac{\Tr{\bsX\bsY} + \Tr{\bsX}\Tr{\bsY}}{d(d+1)},
\end{eqnarray}
where
\begin{eqnarray}
\dif\mu(\psi) \defeq
\frac{\Gamma(d)}{2\pi^d}\delta(1-\norm{\psi})[\dif\psi]
\end{eqnarray}
for Lebesgue volume element $[\dif\psi]=\prod^n_{k=1}\dif x_k\dif
y_k$ via $\psi_k=x_k+\mathrm{i}y_k$ with $x_k,y_k\in\real$. Note
that here $\delta$ is the Dirac delta function.
\end{cor}

\begin{proof}
The original integral can be reduced to the computing of the
following integral:
\begin{eqnarray*}
&&\int_{\U(d)} (\bsU\ot \bsU) (\bsX\ot \bsY)(\bsU\ot \bsU)^\dagger \dif\mu(\bsU)\\
&&=\frac{\Tr{\bsX\ot \bsY} - \Tr{(\bsX\ot \bsY)\bsF}}{d(d-1)}
P_\wedge + \frac{\Tr{\bsX\ot \bsY} + \Tr{(\bsX\ot
\bsY)\bsF}}{d(d+1)} P_\vee.
\end{eqnarray*}
Since $P_\vee\ket{\psi_0\psi_0} =\ket{\psi_0\psi_0}$ and
$P_\wedge\ket{\psi_0\psi_0} =0$, it follows that
\begin{eqnarray*}
&&\int
\Innerm{\psi}{\bsX}{\psi}\Innerm{\psi}{\bsY}{\psi} \dif\mu(\psi) \\
&&= \Innerm{\psi_0\psi_0}{\frac{\Tr{\bsX}\Tr{\bsY} -
\Tr{\bsX\bsY}}{d(d-1)} \bsP_\wedge + \frac{\Tr{\bsX}\Tr{\bsY} +
\Tr{\bsX\bsY}}{d(d+1)} \bsP_\vee}{\psi_0\psi_0} \\
&&= \frac{\Tr{\bsX\bsY} + \Tr{\bsX}\Tr{\bsY}}{d(d+1)}.
\end{eqnarray*}
We are done.
\end{proof}
As a direct consequence of the above Corollary, it follows that for
any super-operator $\Phi$ on $\End(\cH_d)$,
\begin{eqnarray}
\int_{\complex^d} \Innerm{\psi}{\Phi(\out{\psi}{\psi})}{\psi}
\dif\mu(\psi) = \frac{\Tr{\Phi(\I_d)} + \Tr{\Phi}}{d(d+1)}.
\end{eqnarray}

\begin{cor}\label{cor:UUUU}
It holds that
\begin{eqnarray}
&&\int_{\mathsf{U}(d)}\bsU\ot\overline{\bsU}\ot \bsU\ot\overline{\bsU} \dif\mu(\bsU)\nonumber \\
&&= \frac1{d^2-1} \Pa{\out{\vec(\I_{d^2})}{\vec(\I_{d^2})} +
\out{\vec(\bsF)}{\vec(\bsF)}}\nonumber \\
&&~~~~- \frac1{d(d^2-1)}\Pa{\out{\vec(\I_{d^2})}{\vec(\bsF)} +
\out{\vec(\bsF)}{\vec(\I_{d^2})}}.
\end{eqnarray}
\end{cor}

\begin{proof}
Apparently, this result can be derived from
Proposition~\ref{prop:uu-integral}.
\end{proof}

\begin{center}
\fcolorbox{purple}{lightgray}{
\parbox{16.3cm}{
\begin{cor}\label{cor:uuuu}
It holds that
\begin{eqnarray}
\int_{\mathsf{U}(d)} \bsU\ot \bsU^\dagger\ot \bsU\ot \bsU^\dagger
\dif\mu(\bsU) = \frac{\bP_{(12)(34)} + \bP_{(14)(23)}}{d^2-1} -
\frac{\bP_{(1432)} + \bP_{(1234)}}{d(d^2-1)}.
\end{eqnarray}
Moreover
\begin{eqnarray}
\int_{\mathsf{U}(d)} \bsU\ot \bsU\ot \bsU^\dagger\ot \bsU^\dagger
\dif\mu(\bsU) = \frac{\bP_{(13)(24)} + \bP_{(14)(23)}}{d^2-1} -
\frac{\bP_{(1423)} + \bP_{(1324)}}{d(d^2-1)}.
\end{eqnarray}
\end{cor}}}
\end{center}

\begin{proof}
By taking partial transposes relative to the 2nd and 4th
subsystems at the same time over both sides in
Corollary~\ref{cor:UUUU}, it suffices to show that
\begin{eqnarray*}
(\out{\vec(\I_{d^2})}{\vec(\I_{d^2})})^{\t_{2,4}} &=&
\bP_{(12)(34)},\\
(\out{\vec(\bsF)}{\vec(\bsF)})^{\t_{2,4}} &=&
\bP_{(14)(23)},\\
(\out{\vec(\I_{d^2})}{\vec(\bsF)})^{\t_{2,4}} &=&
\bP_{(1432)},\\
(\out{\vec(\bsF)}{\vec(\I_{d^2})})^{\t_{2,4}} &=& \bP_{(1234)}.
\end{eqnarray*}
Note that
$$
\vec(\I_{d^2}) = \sum^{d}_{i,j=1}\ket{iijj},~~\vec(\bsF) =
\sum^{d}_{i,j=1}\ket{ijji}.
$$
It follows that
\begin{eqnarray*}
\out{\vec(\I_{d^2})}{\vec(\I_{d^2})} &=&
\sum^{d}_{i,j,k,l=1}\out{iijj}{kkll},\\
\out{\vec(\bsF)}{\vec(\bsF)} &=& \sum^{d}_{i,j,k,l=1}\out{ijji}{kllk},\\
\out{\vec(\I_{d^2})}{\vec(\bsF)} &=&
\sum^{d}_{i,j,k,l=1}\out{iijj}{kllk},\\
\out{\vec(\bsF)}{\vec(\I_{d^2})} &=&
\sum^{d}_{i,j,k,l=1}\out{ijji}{kkll}.
\end{eqnarray*}
Therefore we have
\begin{eqnarray*}
(\out{\vec(\I_{d^2})}{\vec(\I_{d^2})})^{\t_{2,4}} &=&
\sum^{d}_{i,j,k,l=1}\out{ikjl}{kilj} = \bP_{(12)(34)},\\
(\out{\vec(\bsF)}{\vec(\bsF)})^{\t_{2,4}} &=& \sum^{d}_{i,j,k,l=1}\out{iljk}{kjli} = \bP_{(14)(23)},\\
(\out{\vec(\I_{d^2})}{\vec(\bsF)})^{\t_{2,4}} &=&
\sum^{d}_{i,j,k,l=1}\out{iljk}{kilj} = \bP_{(1432)},\\
(\out{\vec(\bsF)}{\vec(\I_{d^2})})^{\t_{2,4}} &=&
\sum^{d}_{i,j,k,l=1}\out{ikjl}{kjli} = \bP_{(1234)}.
\end{eqnarray*}
The proof is complete.
\end{proof}

\begin{cor}\label{cor:reverseorder}
It holds that
\begin{eqnarray}
\int_{\mathsf{U}(d)} \bsU\ot \Big(\bsU^\dagger\Big)^\t\ot \bsU^\dagger\ot\bsU^\t
\dif\mu(\bsU) = \bP_{(24)}\Br{\frac{\bP^{\t_{2,4}}_{(13)(24)} + \bP^{\t_{2,4}}_{(14)(23)}}{d^2-1} -
\frac{\bP^{\t_{2,4}}_{(1423)} + \bP^{\t_{2,4}}_{(1324)}}{d(d^2-1)}}\bP_{(24)}.
\end{eqnarray}
\end{cor}

\begin{proof}
Note that
\begin{eqnarray}\label{eq:uu*uu*}
\int_{\mathsf{U}(d)} \bsU\ot \bsU^\t\ot \bsU^\dagger\ot
(\bsU^\dagger)^\t \dif\mu(\bsU) = \frac{\bP^{\t_{2,4}}_{(13)(24)} +
\bP^{\t_{2,4}}_{(14)(23)}}{d^2-1} - \frac{\bP^{\t_{2,4}}_{(1423)} +
\bP^{\t_{2,4}}_{(1324)}}{d(d^2-1)}.
\end{eqnarray}
By swaping the second and fourth factors using $\bP_{(24)}$, we get the desired result.
\end{proof}

\begin{center}
\fcolorbox{purple}{lightgray}{
\parbox{16.3cm}{
\begin{cor}\label{cor:uu*u*u}
It holds that
\begin{eqnarray}
\int_{\mathsf{U}(d)} (\bsU\ot\bsU^\dagger)\bsM (\bsU^\dagger\ot\bsU)\dif\mu(\bsU)= \frac{\bsF\bsM\bsF + \Tr{\bsM}\I_{d^2}}{d^2-1} - \frac{\Ptr{1}{\bsM}\ot\I_d+\I_d\ot\Ptr{2}{\bsM}}{d(d^2-1)},
\end{eqnarray}
where $\bsF=\sum^d_{i,j=1}\out{ij}{ji}$ is the swap operator.
\end{cor}}}
\end{center}

\begin{proof}
In fact,
\begin{eqnarray*}
&&\vec\Pa{\int_{\mathsf{U}(d)} (\bsU\ot\bsU^\dagger)\bsM (\bsU^\dagger\ot\bsU)\dif\mu(\bsU)}\\
&&=\Pa{\int_{\mathsf{U}(d)} \bsU\ot \Big(\bsU^\dagger\Big)^\t\ot \bsU^\dagger\ot\bsU^\t
\dif\mu(\bsU)\dif\mu(\bsU)}\vec(\bsM).
\end{eqnarray*}
Since
\begin{eqnarray*}
\bP^{\t_{2,4}}_{(13)(24)} &=& \sum^{d}_{i,j,k,l=1}\out{ilkj}{kjil} = \bP_{(13)(24)},\\
\bP^{\t_{2,4}}_{(14)(23)} &=& \sum^{d}_{i,j,k,l=1}\out{illi}{kjjk}=\proj{\vec(\bsF)},\\
\bP^{\t_{2,4}}_{(1423)} &=&  \sum^{d}_{i,j,k,l=1}\out{illj}{kjik},\\
\bP^{\t_{2,4}}_{(1324)} &=& \sum^{d}_{i,j,k,l=1}\out{ilki}{kjjl},
\end{eqnarray*}
it follows that
\begin{eqnarray*}
\bP_{(24)}\bP^{\t_{2,4}}_{(13)(24)}\bP_{(24)} &=& \sum^{d}_{i,j,k,l=1}\out{ijkl}{klij} = \bP_{(13)(24)},\\
\bP_{(24)}\bP^{\t_{2,4}}_{(14)(23)}\bP_{(24)} &=& \sum^{d}_{i,j,k,l=1}\out{iill}{kkjj}=\proj{\vec(\I_d)}^{\ot2},\\
\bP_{(24)}\bP^{\t_{2,4}}_{(1423)}\bP_{(24)} &=&  \sum^{d}_{i,j,k,l=1}\out{ijll}{kkij},\\
\bP_{(24)}\bP^{\t_{2,4}}_{(1324)}\bP_{(24)} &=& \sum^{d}_{i,j,k,l=1}\out{iikl}{kljj}.
\end{eqnarray*}
Note that, via $\vec(\out{m\mu}{n\nu})=\vec(\out{m}{n}\ot\out{\mu}{\nu})=\vec(\out{m}{n})\ot\vec(\out{\mu}{\nu})=\ket{mn\mu\nu}$,
\begin{eqnarray*}
\bP_{(24)}\bP^{\t_{2,4}}_{(13)(24)}\bP_{(24)}\vec(\bsM) &=& \bP_{(13)(24)}\vec(\bsM)= \sum_{m,n,\mu,\nu}\Innerm{m\mu}{\bsM}{n\nu}\bP_{(13)(24)}\vec(\out{m\mu}{n\nu}) \\
&=& \sum_{m,n,\mu,\nu}\Innerm{m\mu}{\bsM}{n\nu}\bP_{(13)(24)}\ket{m n\mu\nu}\\
&=&  \sum_{m,n,\mu,\nu}\Innerm{m\mu}{\bsM}{n\nu} \ket{\mu\nu mn}\\
&=& \sum_{m,n,\mu,\nu}\Innerm{m\mu}{\bsM}{n\nu} \vec(\out{\mu m}{\nu n})\\
&=& \vec(\bsF\bsM\bsF),
\end{eqnarray*}
and
\begin{eqnarray*}
&&\bP_{(24)}\bP^{\t_{2,4}}_{(14)(23)}\bP_{(24)}\vec(\bsM) = \vec(\I_{d^2})\vec(\I_{d^2})^\dagger\vec(\bsM)\\
&&=\Inner{\vec(\I_{d^2})}{\vec(\bsM)}\vec(\I_{d^2})=\Inner{\I_{d^2}}{\bsM}\vec(\I_{d^2})=\vec(\Tr{\bsM}\I_{d^2}),
\end{eqnarray*}
and
\begin{eqnarray*}
&&\bP_{(24)}\bP^{\t_{2,4}}_{(1423)}\bP_{(24)}\vec(\bsM) = \sum_{i,j,k,l}\out{ijll}{kkij}\vec(\bsM)=\sum_{i,j,k,l}\Innerm{ki}{\bsM}{kj}\ket{ijll}\\
&&=\sum_{i,j,k,l}\Innerm{ki}{\bsM}{kj}\vec(\out{il}{jl}) = \sum_{i,j}\Innerm{i}{\Ptr{1}{\bsM}}{j}\vec(\out{i}{j}\ot\I_d)\\
&&=\vec\Pa{\Pa{\sum_{i,j}\Innerm{i}{\Ptr{1}{\bsM}}{j}\out{i}{j}}\ot\I_d} = \vec\Pa{\Ptr{1}{\bsM}\ot\I_d},
\end{eqnarray*}
and similarly, we have
\begin{eqnarray*}
\bP_{(24)}\bP^{\t_{2,4}}_{(1423)}\bP_{(24)}\vec(\bsM) = \vec\Pa{\I_d\ot\Ptr{2}{\bsM}}.
\end{eqnarray*}
The desired conclusion is obtained immediately.
\end{proof}

If a bipartite state $\rho\in\density{\complex^d\ot\complex^d}$ is such that
$\rho = (\bsU\ot\bsU^\dagger)\rho(\bsU^\dagger\ot\bsU)$ for all $\bsU\in\U(d)$,
then $\rho_1=\bsU\rho_1\bsU^\dagger, \rho_2=\bsU^\dagger\rho_2\bsU$, and
\begin{eqnarray}
\rho = \frac{\bsF\rho\bsF + \I_{d^2}}{d^2-1} - \frac{\rho_2\ot\I_d+\I_d\ot\rho_1}{d(d^2-1)},
\end{eqnarray}
impling that $\rho_1=\rho_2=\frac{\I_d}d$ and
\begin{eqnarray}
\bsF\rho\bsF = \frac{\rho + \I_{d^2}}{d^2-1} - \frac{\I_d\ot\rho_2+\rho_1\ot\I_d}{d(d^2-1)}.
\end{eqnarray}
Note that $\bsF\rho\bsF=\rho$. We obtain that $\rho$ must be completely mixed state.

As a quantum channel, the above integral defining a unital bipartite channel
\begin{eqnarray}
\Phi(\rho_{12}) &=&  \int_{\mathsf{U}(d)} (\bsU\ot\bsU^\dagger)\rho_{12} (\bsU^\dagger\ot\bsU)\dif\mu(\bsU)\notag\\
&=&  \frac{\bsF\rho_{12}\bsF + \I_{d^2}}{d^2-1} - \frac{\rho_2\ot\I_d+\I_d\ot\rho_1}{d(d^2-1)}\notag\\
&=&\bsF\Pa{\frac{\rho_{12} + \I_{d^2}}{d^2-1} - \frac{\I_d\ot\rho_2+\rho_1\ot\I_d}{d(d^2-1)}}\bsF\notag\\
&=&\frac1{d^2-1}\bsF\Pa{\rho_{12} + \I_{d^2} - \frac{\I_d}d\ot\rho_2-\rho_1\ot\frac{\I_d}d}\bsF.
\end{eqnarray}

\begin{center}
\fcolorbox{purple}{lightgray}{
\parbox{16.3cm}{
\begin{cor}\label{cor:uututu}
It holds that
\begin{eqnarray}
&&\int_{\mathsf{U}(d)} (\bsU\ot\bsU^\t)\bsM (\bsU\ot\bsU^\t)^\dagger\dif\mu(\bsU)\notag\\
&&= \frac{\bsF\bsM^\t\bsF + \Tr{\bsM}\I_{d^2}}{d^2-1} - \frac{(\Ptr{1}{\bsM})^\t\ot\I_d+\I_d\ot(\Ptr{2}{\bsM})^\t}{d(d^2-1)}.
\end{eqnarray}
\end{cor}}}
\end{center}

\begin{proof}
By swaping the 3rd and 4th factors in the integral formula in Corollary~\ref{cor:reverseorder}, and we get that:
\begin{eqnarray}
\int_{\mathsf{U}(d)} \bsU\ot \Big(\bsU^\dagger\Big)^\t\ot\bsU^\t\ot \bsU^\dagger
\dif\mu(\bsU) = \bP_{(234)}\Br{\frac{\bP^{\t_{2,4}}_{(13)(24)} + \bP^{\t_{2,4}}_{(14)(23)}}{d^2-1} -
\frac{\bP^{\t_{2,4}}_{(1423)} + \bP^{\t_{2,4}}_{(1324)}}{d(d^2-1)}}\bP_{(243)}.
\end{eqnarray}
The proof is obtained immediately.
\end{proof}

\begin{cor}
It holds that
\begin{eqnarray}
\int_{\mathsf{U}(d)} \abs{\Tr{\bsA\bsU}}^4\dif\mu(\bsU) =
\frac2{d^2-1}\Br{\Tr{\bsA^\dagger \bsA}}^2 -
\frac2{d(d^2-1)}\Tr{(\bsA^\dagger \bsA)^2},
\end{eqnarray}
where $\bsA\in M_d(\complex)$.
\end{cor}

\begin{proof}
Note that
$$
\abs{\Tr{\bsA\bsU}}^4 = \Tr{\Br{\bsA^{\ot 2}\ot (\bsA^{\ot
2})^\dagger} \Br{\bsU^{\ot 2}\ot(\bsU^{\ot 2})^\dagger}}.
$$
By Corollary~\ref{cor:uuuu}, we obtain the final result.
\end{proof}

\begin{cor}
It holds that
\begin{eqnarray}
\int_{\mathsf{U}(d)} \abs{\Tr{\bsA\bsU\bsB\bsU}}^2
\dif\mu(\bsU)={\scriptsize
\frac{\abs{\Tr{\bsA^\dagger\bsB}}^2+\Tr{\bsA\bsA^\dagger}\Tr{\bsB\bsB^\dagger}}{d^2-1}-\frac{\Tr{\bsA^\dagger\bsA\bsB^\dagger\bsB}+\Tr{\bsA\bsA^\dagger\bsB\bsB^\dagger}}{d(d^2-1)}.}
\end{eqnarray}
where $\bsA,\bsB\in M_d(\complex)$.
\end{cor}

\begin{proof}
In fact,
\begin{eqnarray}
\int_{\mathsf{U}(d)} \abs{\Tr{\bsA\bsU\bsB\bsU}}^2 \dif\mu(\bsU) =
\int_{\mathsf{U}(d)}\dif\mu(\bsU)\Tr{(\bsA\ot\bsA^\dagger)(\bsU\ot\bsU^\dagger)(\bsB\ot\bsB^\dagger)(\bsU\ot\bsU^\dagger)}.
\end{eqnarray}
It suffices to calculate the following integral:
\begin{eqnarray}
\int_{\mathsf{U}(d)}\dif\mu(\bsU)(\bsU\ot\bsU^\dagger)\bsX(\bsU\ot\bsU^\dagger).
\end{eqnarray}
Note that
\begin{eqnarray}
\vec\Pa{\int_{\mathsf{U}(d)}\dif\mu(\bsU)(\bsU\ot\bsU^\dagger)\bsX(\bsU\ot\bsU^\dagger)}=\int_{\mathsf{U}(d)}\dif\mu(\bsU)\bsU\ot\bsU^\t\ot\bsU^\dagger\ot(\bsU^\dagger)^\t\vec(\bsX).
\end{eqnarray}
Using Eq.~\eqref{eq:uu*uu*}, we get that
\begin{eqnarray*}
\bP_{(13)(24)}\vec(\bsX) &=& \vec(\bsF\bsX\bsF),\\
\proj{\vec(\bsF)}\vec(\bsX) &=&  \Tr{\bsF\bsX}\vec(\bsF),\\
\sum_{i,j,k,l}\out{illj}{kjik}\vec(\bsX) &=&
\vec((\Ptr{1}{\bsX\bsF}\ot\I_d)\bsF),\\
\sum_{i,j,k,l}\out{ilki}{kjjl}\vec(\bsX) &=&
\vec((\I_d\ot\Ptr{2}{\bsX\bsF})\bsF).
\end{eqnarray*}
Based on this observation, we obtain that
\begin{eqnarray*}
&&\vec\Pa{\int_{\mathsf{U}(d)}\dif\mu(\bsU)(\bsU\ot\bsU^\dagger)\bsX(\bsU\ot\bsU^\dagger)}\\
&&=\frac{\vec(\bsF\bsX\bsF)+\Tr{\bsF\bsX}\vec(\bsF)}{d^2-1} -
\frac{\vec((\Ptr{1}{\bsX\bsF}\ot\I_d)\bsF)+\vec((\I_d\ot\Ptr{2}{\bsX\bsF})\bsF)}{d(d^2-1)},
\end{eqnarray*}
implying that
\begin{eqnarray*}
&&\int_{\mathsf{U}(d)}\dif\mu(\bsU)(\bsU\ot\bsU^\dagger)\bsX(\bsU\ot\bsU^\dagger)\\
&&=\frac{\bsF\bsX\bsF+\Tr{\bsF\bsX}\bsF}{d^2-1} -
\frac{(\Ptr{1}{\bsX\bsF}\ot\I_d)\bsF+(\I_d\ot\Ptr{2}{\bsX\bsF})\bsF}{d(d^2-1)}.
\end{eqnarray*}
Let $\bsX=\bsB\ot\bsB^\dagger$ be in the above. We get that
\begin{eqnarray*}
&&\int_{\mathsf{U}(d)}\dif\mu(\bsU)(\bsU\ot\bsU^\dagger)(\bsB\ot\bsB^\dagger)(\bsU\ot\bsU^\dagger)\\
&&=\frac{\bsB^\dagger\ot\bsB+\Tr{\bsB\bsB^\dagger}\bsF}{d^2-1} -
\frac{(\bsB^\dagger\bsB\ot\I_d)\bsF+(\I_d\ot\bsB\bsB^\dagger)\bsF}{d(d^2-1)}.
\end{eqnarray*}
Finally, we obtain that
\begin{eqnarray*}
\int_{\mathsf{U}(d)} \abs{\Tr{\bsA\bsU\bsB\bsU}}^2
\dif\mu(\bsU)={\scriptsize
\frac{\abs{\Tr{\bsA^\dagger\bsB}}^2+\Tr{\bsA\bsA^\dagger}\Tr{\bsB\bsB^\dagger}}{d^2-1}-\frac{\Tr{\bsA^\dagger\bsA\bsB^\dagger\bsB}+\Tr{\bsA\bsA^\dagger\bsB\bsB^\dagger}}{d(d^2-1)}.}
\end{eqnarray*}
This completes the proof.
\end{proof}

\subsection{The general case}

The partial materials in this subsection are written based on the
results in \cite{Benoit,Collins}. We recall that for an algebra
inclusion $\cM\subset\cN$, a \emph{conditional expectation} is a
$\cM$-bimodule map $\sE: \cN\to\cM$ such that $\sE(\I_\cN) =
\I_\cM$. For $\bsA\in\End((\complex^d)^{\ot k})$, we define
\begin{center}
\fcolorbox{purple}{lightgray}{
\parbox{16.3cm}{
\begin{eqnarray}
\sE^{(k)}_{d}(\bsA) = \int_{\U(d)} \bsU^{\ot k}\bsA \Pa{\bsU^{\ot
k}}^\dagger \dif\mu(\bsU).
\end{eqnarray}}}
\end{center}
Clearly $\sE^{(k)}_{d}: \End((\complex^d)^{\ot k})\to
\bP(\complex[S_k])$ is a conditional expectation. Moreover
$\sE^{(k)}_{d}$ is an orthogonal projection onto
$\bP(\complex[S_k])$. It is compatible with the trace in the sense
that $\trace\circ\sE^{(k)}_{d} = \trace$.

For $\bsA\in\End((\complex^d)^{\ot k})$, we set
\begin{eqnarray}
\Delta_{d,k}(\bsA) &\defeq& \sum_{\pi\in S_k}
\Inner{\bP(\pi)}{\bsA}\bP(\pi)=
\sum_{\pi\in S_k} \Tr{\bsA\bP(\pi^{-1})}\bP(\pi) \notag\\
&=& \sum_{\pi\in S_k} \Tr{\bsA\bP(\pi)}\bP(\pi^{-1}) \in
\bP(\complex[S_k]).
\end{eqnarray}
In particular, $\Delta_{d,k}(\I)=\sum_{\pi\in S_k}
\Tr{\bP(\pi^{-1})}\bP(\pi)=\sum_{\pi\in S_k}
\chi_{\bP}(\pi^{-1})\bP(\pi)$, thus
\begin{center}
\fcolorbox{purple}{lightgray}{
\parbox{16.3cm}{
\begin{eqnarray}
\Delta_{d,k}(\I) = \sum_{\pi\in S_k} \chi_{\bP}(\pi)\bP(\pi),
\end{eqnarray}}}
\end{center}
where $\chi_{\bP}(\pi)$ can be analytically calculated in the
following proposition.

\begin{prop}
For the action of $S_k$ on the tensor space $(\complex^d)^{\ot k}$,
defined in Eq.~\eqref{eq:P}, its character is given by
\begin{center}
\fcolorbox{purple}{lightgray}{
\parbox{16.3cm}{
\begin{eqnarray}
\chi_{\bP}(\pi) = \Tr{\bP(\pi)} = d^{c(\pi)} \quad(\forall\pi\in
S_k),
\end{eqnarray}}}
\end{center}
where $c(\pi)$ is the number of cycles into which it decomposes.
\end{prop}
If we denote by $\lambda(\pi)$ the partition of $k$ given by the
lengths of these cycles in $\pi$, then we see immediately that
$c(\pi)$ is just the height (i.e., the number of boxes in the first
column) of partition $\lambda(\pi)\vdash k$.

\begin{proof}
(i) If $\Set{\ket{j}:j=1,\ldots,d}$ is a given basis of
$\complex^d$, then the induced basis of $(\complex^d)^{\ot k}$ can
be described as
$$
\ket{i_1i_2\cdots i_k}\equiv\ket{i_1}\ot\cdots\ket{i_k},
$$
where each index $i_1,\ldots,i_k$ is in $[d]:=\set{1,2,\ldots,d}$.
For the permutation representation $\bP(\pi)$, the trace
$\Tr{\bP(\pi)}$ of an element $\pi\in S_k$ must be equal the number
of the elements of the basis fixed by $\pi$. If $\pi=(1,2,...,k)$ is
one cycle (that is, $c(\pi)=1$), then one basis vector
$\ket{i_1i_2\cdots i_k}$ is fixed by $\pi$, i.e.,
$\bP(\pi)\ket{i_1i_2\cdots i_k}=\ket{i_1i_2\cdots i_k}$, or
$\ket{i_ki_1\cdots i_{k-1}}=\ket{i_1i_2\cdots i_k}$ if and only if
\begin{eqnarray*}
i_1=i_2=\cdots=i_k\in[d].
\end{eqnarray*}
This means that the set $\set{\bsv\in (\complex^d)^{\ot
k}:\bP(\pi)\bsv=\bsv}$ is of $d$ dimension. Thus
$\chi_{\bP}(\pi)=d=d^{c(\pi)}$.

(ii) If $\pi$ is a product of $\ell$ cycles of lengths
$m_1,m_2,\ldots,m_\ell$ (that is, $c(\pi)=\ell$) which can be of the
following form (up to conjugacy):
$$
\pi=(1,2,\ldots,m_1)(m_1+1,m_1+2,\ldots,m_1+m_2)\cdots(k-m_\ell+1,\ldots,k).
$$
We see that $\ket{i_1i_2\cdots i_k}$ is fixed by $\pi$ if and only
if it is of the form
$$
\ket{i_1}^{\ot m_1}\ket{i_2}^{\ot m_2}\cdots\ket{i_\ell}^{\ot
m_\ell},
$$
where each index $i_1,\ldots,i_\ell$ is in $[d]$. Apparently the
number
$\Abs{\Set{(i_1,\ldots,i_\ell)\in[d]^\ell}}=d^\ell=d^{c(\pi)}$.

In summary, $\chi_{\bP}(\pi) = \Tr{\bP(\pi)} = d^{c(\pi)}$ for
$\pi\in S_k$.
\end{proof}

\begin{prop}\label{prop:Deltadk}
The linear mapping $\Delta_{d,k}$ embraces the following properties:
\begin{enumerate}[(i)]
\item $\Delta_{d,k}$ is a $\bP(\complex[S_k])$-$\bP(\complex[S_k])$ bimodule
morphism in the sense that
$$
\Delta_{d,k}(\bsA\bP(\sigma)) =
\Delta_{d,k}(\bsA)\bP(\sigma),~~\Delta_{d,k}(\bP(\sigma)\bsA) =
\bP(\sigma)\Delta_{d,k}(\bsA).
$$
\item $\Delta_{d,k}(\I)$ coincides with the character of $\bP$, hence it is
equal to
\begin{center}
\fcolorbox{purple}{lightgray}{
\parbox{15.3cm}{
\begin{eqnarray}
\Delta_{d,k}(\I) = k! \sum_{\lambda\vdash_d k}
\frac{s_{\lambda}(1^{\times d})}{f^\lambda}\bsC_\lambda.
\end{eqnarray}}}
\end{center}
and is an invertible element of $\complex[S_k]$; its inverse will be
called \blue{Weingarten function} and is equal to
\begin{center}
\fcolorbox{purple}{lightgray}{
\parbox{15.3cm}{
\begin{eqnarray}
\mathrm{Wg}_{d,k} = \frac1{(k!)^2}\sum_{\lambda\vdash_d k}
\frac{(f^\lambda)^2}{s_{\lambda}(1^{\times d})}\chi_\lambda.
\end{eqnarray}}}
\end{center}
\item The relation between $\Delta_{d,k}(\bsA)$ and $\sE^{(k)}_d(\bsA)$ is explicitly
given by
\begin{center}
\fcolorbox{purple}{lightgray}{
\parbox{15.3cm}{
\begin{eqnarray}
\Delta_{d,k}(\bsA) = \sE^{(k)}_d(\bsA)\Delta_{d,k}(\I).
\end{eqnarray}}}
\end{center}
This leads to the formula:
\begin{center}
\fcolorbox{purple}{lightgray}{
\parbox{15.3cm}{
\begin{eqnarray}\label{eq:general-integral}
&&\sE^{(k)}_d(\bsA) = \Delta_{d,k}(\bsA)\Delta_{d,k}(\I)^{-1}\notag\\
&&= \Pa{\sum_{\pi\in
S_k}\Tr{\bsA\bP(\pi)}\bP(\pi^{-1})}\Pa{\sum_{\pi\in
S_k}\mathrm{Wg}_{d,k}(\pi)\bP(\pi^{-1})},
\end{eqnarray}}}
\end{center}
where
\begin{eqnarray}
\Delta_{d,k}(\I)^{-1}=\frac1{k!}\sum_{\lambda\vdash_d k}
\frac{f^\lambda}{s_{\lambda}(1^{\times
d})}\bsC_\lambda=\Pa{\sum_{\pi\in
S_k}\mathrm{Wg}_{d,k}(\pi)\bP(\pi^{-1})}.
\end{eqnarray}
\item The range of $\Delta_{d,k}$ is equal to $\bP(\complex[S_k])$.
\item The following holds true in $\bP(\complex[S_k])$:
$$
\Delta_{d,k}(\bsA\sE^{(k)}_d(\bsB)) =
\Delta_{d,k}(\bsA)\Delta_{d,k}(\bsB)\Delta_{d,k}(\I)^{-1}.
$$
\end{enumerate}
\end{prop}

\begin{proof}
(i). Clearly we have:
\begin{eqnarray*}
\Delta_{d,k}(\bsA\bP(\sigma)) &=& \sum_{\pi\in S_k}
\Tr{[\bsA\bP(\sigma)]\bP(\pi^{-1})}\bP(\pi)\\
&=& \sum_{\pi\in S_k}
\Tr{\bsA\bP(\sigma\pi^{-1})}\bP((\sigma\pi^{-1})^{-1})\bP(\sigma)\\
&=&\Delta_{d,k}(\bsA)\bP(\sigma).
\end{eqnarray*}
Similarly, we also have: $\Delta_{d,k}(\bP(\sigma)\bsA) =
\bP(\sigma)\Delta_{d,k}(\bsA)$. Furthermore we get
\begin{center}
\fcolorbox{purple}{lightgray}{
\parbox{16.3cm}{
\begin{eqnarray}
\Delta_{d,k}(\bP(\sigma_l) \bsA\bP(\sigma_r)) =
\bP(\sigma_l)\Delta_{d,k}(\bsA)\bP(\sigma_r),
\end{eqnarray}}}
\end{center}
where $\sigma_l,\sigma_r\in S_k$. Therefore $\Delta_{d,k}$ is
bimodule
morphism.\\
(ii). Let $\bsA=\I$ in the definition of $\Delta_{d,k}$. We get that
\begin{eqnarray}
\Delta_{d,k}(\I) = \sum_{\pi\in S_k} \Tr{\bP(\pi^{-1})}\bP(\pi) =
\sum_{\pi\in S_k} \chi(\pi^{-1})\bP(\pi).
\end{eqnarray}
By Schur-Weyl duality, we have
$$
(\complex^d)^{\ot k} \cong \bigoplus_{\lambda\vdash_d k}
\bQ_\lambda\ot\bP_\lambda
$$
and
\begin{eqnarray}
\chi_{\bP} = \sum_{\lambda\vdash_d k} d_\lambda \chi_\lambda,
\end{eqnarray}
where $d_\lambda$ is the multiplicities of $\bP_\lambda$, i.e.
$d_\lambda = \dim(\bQ_\lambda)=s_{\lambda}(1^{\times d})$. Hence
$$
\chi_{\bP}(\pi^{-1}) = \sum_{\lambda\vdash_d k} d_\lambda
\chi_\lambda(\pi^{-1})= \sum_{\lambda\vdash_d k}
s_{\lambda}(1^{\times d}) \chi_\lambda(\pi^{-1}),
$$
which is substituted into the rhs of expression of
$\Delta_{d,k}(\I)$ above, gives rise to
\begin{eqnarray*}
\Delta_{d,k}(\I) &=& \sum_{\pi\in S_k}\Pa{\sum_{\lambda\vdash_dk}
s_{\lambda}(1^{\times d}) \chi_\lambda(\pi^{-1})}\bP(\pi)\\
&=& \sum_{\lambda\vdash_d k} s_{\lambda}(1^{\times
d})\Pa{\sum_{\pi\in S_k} \chi_\lambda(\pi^{-1})\bP(\pi)}.
\end{eqnarray*}
Since the minimal central projection $\bsC_\lambda$ in
$\bP(\complex[S_k])$ must be of the following form:
\begin{center}
\fcolorbox{purple}{lightgray}{
\parbox{16.3cm}{
\begin{eqnarray}\label{eq:cproj}
\bsC_\lambda := \frac{f^\lambda}{k!}\sum_{\pi\in
S_k}\chi_\lambda(\pi^{-1})\bP(\pi)= \frac{f^\lambda}{k!}\sum_{\pi\in
S_k}\chi_\lambda(\pi)\bP(\pi),
\end{eqnarray}}}
\end{center}
where $f^\lambda:=\dim(\bP_\lambda)$, it follows that
\begin{eqnarray*}
\sum_{\pi\in S_k}\chi_\lambda(\pi^{-1})\bP(\pi) =
\frac{k!}{f^\lambda}\bsC_\lambda
\end{eqnarray*}
Thus
\begin{eqnarray*}
\Delta_{d,k}(\I) = k!\sum_{\lambda\vdash_d
k}\frac{s_{\lambda}(1^{\times d})}{f^\lambda}\bsC_\lambda.
\end{eqnarray*}
Moreover $\Delta_{d,k}(\I)$ is invertible and
\begin{center}
\fcolorbox{purple}{lightgray}{
\parbox{7cm}{
$$
\Delta_{d,k}(\I)^{-1} = \frac1{k!}\sum_{\lambda\vdash_d
k}\frac{f^\lambda}{s_{\lambda}(1^{\times d})}\bsC_\lambda.
$$}}
\end{center}
We denote by $\mathrm{Wg}_{d,k}$ the function corresponding to
$\Delta_{d,k}(\I)^{-1}$, i.e.
$$
\mathrm{Wg}_{d,k} = \frac1{(k!)^2}\sum_{\lambda\vdash_d k}
\frac{(f^\lambda)^2}{s_{\lambda}(1^{\times d})}\chi_\lambda.
$$
(iii). Since $\bQ(\bsU)$ commutes with $\bP(\pi)$, it follows that
\begin{eqnarray*}
\Delta_{d,k}(\sE^{(k)}_d(\bsA)) &=& \sum_{\pi\in S_k}
\Tr{\sE^{(k)}_d(\bsA)\bP(\pi^{-1})}\bP(\pi)\\
&=& \sum_{\pi\in S_k}
\Tr{\int_{\U(d)}\bQ(\bsU)\bsA\bQ(\bsU)^\dagger \dif\mu(\bsU)\bP(\pi^{-1})}\bP(\pi)\\
&=& \sum_{\pi\in S_k}
\Tr{\bsA\int_{\U(d)}\bQ(\bsU)^\dagger\bP(\pi^{-1})\bQ(\bsU)\dif\mu(\bsU)}\bP(\pi)\\
&=&\sum_{\pi\in S_k} \Tr{\bsA\bP(\pi^{-1})}\bP(\pi) =
\Delta_{d,k}(\bsA),
\end{eqnarray*}
implying
\begin{eqnarray*}
\Delta_{d,k}(\bsA) = \Delta_{d,k}(\sE^{(k)}_d(\bsA)\I) =
\sE^{(k)}_d(\bsA)\Delta_{d,k}(\I)
\end{eqnarray*}
by the fact that $\Delta_{d,k}$ is bimodule morphism and
$\sE^{(k)}_d(\bsA)\in\bP(\complex[S_k])$. We can get more that
\begin{eqnarray*}
\Delta_{d,k}(\bsA) = \Delta_{d,k}(\sE^{(k)}_d(\bsA)) =
\sE^{(k)}_d(\bsA)\Delta_{d,k}(\I) =
\Delta_{d,k}(\I)\sE^{(k)}_d(\bsA).
\end{eqnarray*}
This indicates that
\begin{eqnarray*}
\sE^{(k)}_d(\bsA) &=& \Delta_{d,k}(\bsA)\Delta_{d,k}(\I)^{-1} = \Delta_{d,k}(\I)^{-1}\Delta_{d,k}(\bsA)\\
&=& \frac1{k!}\Pa{\sum_{\pi\in S_k}
\Tr{\bsA\bP(\pi^{-1})}\bP(\pi)}\Pa{\sum_{\lambda\vdash
k}\frac{f^\lambda}{s_{\lambda}(1^{\times d})}\bsC_\lambda}\\
&=&\Pa{\sum_{\pi\in
S_k}\Tr{\bsA\bP(\pi)}\bP(\pi^{-1})}\Pa{\sum_{\pi\in
S_k}\text{Wg}_{d,k}(\pi)\bP(\pi^{-1})}.
\end{eqnarray*}
(iv). It is trivially from (ii) and (iii).\\
(v). It is easily seen that
\begin{eqnarray*}
\Delta_{d,k}(\bsA\sE^{(k)}_d(\bsB)) =
\Delta_{d,k}(\bsA)\sE^{(k)}_d(\bsB) =
\Delta_{d,k}(\bsA)\Delta^{(k)}_d(\bsB)\Delta_{d,k}(\I)^{-1}.
\end{eqnarray*}
We are done.
\end{proof}

\begin{remark}
From Eq.~\eqref{eq:general-integral} in
Proposition~\ref{prop:Deltadk}, we see that
\begin{center}
\fcolorbox{purple}{lightgray}{
\parbox{16.3cm}{
\begin{eqnarray}
\sE^{(k)}_d(\bsA) =  \Pa{\sum_{\pi\in
S_k}\Tr{\bsA\bP(\pi)}\bP(\pi^{-1})}\Pa{\sum_{\pi\in
S_k}\mathrm{Wg}_{d,k}(\pi)\bP(\pi^{-1})},
\end{eqnarray}}}
\end{center}
where the evaluations of Weingarten function
$\mathrm{Wg}_{d,k}(\pi)$ for each $\pi\in S_k$ is the key point.
Since the Weingarten function $\mathrm{Wg}_{d,k}$ is a class
function whose value is kept invariant on the conjugacy classes of
$S_k$. We see from the definition of Weingarten function:
$$
\mathrm{Wg}_{d,k} = \frac1{(k!)^2}\sum_{\lambda\vdash_d k}
\frac{(f^\lambda)^2}{s_{\lambda}(1^{\times d})}\chi_\lambda,
$$
that the calculation of $\sE^{(k)}_d(\bsA)$ is reduced to the
calculations of irreducible characters
$\chi_\lambda(\lambda\vdash_dk)$ of the permutation group $S_k$. For
convenience, we list partial tables of irreducible characters of the
permutation groups for some lower orders in
Appendix~\ref{app:irrech}.
\end{remark}

\begin{center}
\fcolorbox{purple}{lightgray}{
\parbox{16.3cm}{
\begin{cor}
Let $k$ be a positive integer and $\mathbf{i}=(i_1,\ldots,i_k)$,
$\mathbf{i}'=(i'_1,\ldots,i'_k)$, $\mathbf{j}=(j_1,\ldots,j_k)$,
$\mathbf{j}'=(j'_1,\ldots,j'_k)$ be $k$-tuples of positive integers.
Then
\begin{eqnarray}
&&\int_{\mathsf{U}(d)} U_{i_1 j_1}\cdots U_{i_k j_k}
\overline{U_{i'_1
j'_1}}\cdots\overline{U_{i'_k j'_k}}\dif\mu(\bsU) \notag\\
&&= \sum_{\sigma,\tau\in S_k}\mathrm{Wg}_{d,k}(\sigma\tau^{-1})
\iinner{i_1}{i'_{\sigma(1)}}\cdots
\iinner{i_k}{i'_{\sigma(k)}}\iinner{j_1}{j'_{\tau(1)}}\cdots\iinner{j_k}{j'_{\tau(k)}}\\
&&=\sum_{\sigma,\tau\in S_k}\mathrm{Wg}_{d,k}(\sigma\tau^{-1})
\iinner{\mathbf{i}}{\mathbf{i}'_\sigma}\iinner{\mathbf{j}}{\mathbf{j}'_\tau},
\end{eqnarray}
where $\ket{\mathbf{i}}=\ket{i_1,\ldots, i_k}$ and
$\ket{\mathbf{i}'_\sigma}=\ket{i'_{\sigma(1)},\ldots,
i'_{\sigma(k)}}$.
\end{cor}}}
\end{center}

\begin{proof}
Note that
$$
\Delta_{d,k}(\bsA\sE^{(k)}_d(\bsB)) =
\Delta_{d,k}(\bsA)\Delta_{d,k}(\bsB)\Delta_{d,k}(\I)^{-1}.
$$
In order to show our result, it is enough to take appropriate
$\bsA=\out{\mathbf{i}'}{\mathbf{i}}$ and
$\bsB=\out{\mathbf{j}}{\mathbf{j}'}$, where
$\ket{\mathbf{i}}=\ket{i_1\cdots i_k}$, etc. Now that
\begin{eqnarray}
\int_{\U(d)} U_{i_1 j_1}\cdots U_{i_k j_k} \overline{U_{i'_1
j'_1}}\cdots\overline{U_{i'_k j'_k}}\dif\mu(\bsU) =
\Tr{\bsA\sE_k(\bsB)}.
\end{eqnarray}
By the definition of $\Delta_{d,k}$, we get
\begin{eqnarray*}
\Delta_{d,k}(\bsA\sE^{(k)}_d(\bsB)) &=& \sum_{\pi\in S_k}
\Tr{\bsA\sE^{(k)}_d(\bsB)\bP(\pi^{-1})}\bP(\pi) \\
&=& \Tr{\bsA\sE^{(k)}_d(\bsB)}\I + \sum_{\pi\in
S_k\backslash\set{\mathrm{id}}}
\Tr{\bsA\sE^{(k)}_d(\bsB)\bP(\pi^{-1})}\bP(\pi)
\end{eqnarray*}
and
\begin{eqnarray*}
\Delta_{d,k}(\bsA) &=& \sum_{\sigma\in S_k}
\Tr{\bsA\bP(\sigma^{-1})}\bP(\sigma)\\ &=& \sum_{\sigma\in
S_k}\Innerm{\mathbf{i}}{\bP(\sigma^{-1})}{\mathbf{i}'}\bP(\sigma)\\
&=& \sum_{\sigma\in S_k}\iinner{i_1}{i'_{\sigma(1)}}\cdots
\iinner{i_k}{i'_{\sigma(k)}}\bP(\sigma),
\end{eqnarray*}
where $\bP(\sigma)\ket{i_1\cdots i_k} =
\ket{i_{\sigma^{-1}(1)}\cdots i_{\sigma^{-1}(k)}}$ or
$\bP(\sigma)\ket{i_{\sigma(1)}\cdots i_{\sigma(k)}} = \ket{i_1\cdots
i_k}$. That is,
$$
\bP(\sigma) = \sum_{i_1,\ldots, i_k\in[d]}\out{i_1\cdots
i_k}{i_{\sigma(1)}\cdots i_{\sigma(k)}} =
\sum_{\mathbf{i}}\out{\mathbf{i}}{\mathbf{i}_\sigma}.
$$
Note also that $\bP(\sigma)^\dagger  = \bP(\sigma^{-1})$. Therefore
$$
\bP(\sigma^{-1}) = \bP(\sigma)^\dagger = \sum_{i_1,\ldots,
i_k\in[d]}\out{i_{\sigma(1)}\cdots i_{\sigma(k)}}{i_1\cdots
i_k}=\sum_{\mathbf{i}}\out{\mathbf{i}_\sigma}{\mathbf{i}}.
$$
Similarly
\begin{eqnarray*}
\Delta_{d,k}(\bsB) &=& \sum_{\tau\in S_k}
\Innerm{\mathbf{j}'}{\bP(\tau^{-1})}{\mathbf{j}}\bP(\tau) =
\sum_{\sigma\in S_k}\iinner{j'_1}{j_{\tau(1)}}\cdots
\iinner{j'_k}{j_{\tau(k)}}\bP(\tau)\\
&=&\sum_{\tau\in S_k}\iinner{j_{\tau(1)}}{j'_1}\cdots
\iinner{j_{\tau(k)}}{j'_k}\bP(\tau) = \sum_{\tau\in
S_k}\iinner{j_1}{j'_{\tau^{-1}(1)}}\cdots
\iinner{j_k}{j'_{\tau^{-1}(1)}}\bP(\tau)\\
&=&\sum_{\tau\in S_k}\iinner{j_1}{j'_{\tau(1)}}\cdots
\iinner{j_k}{j'_{\tau(1)}}\bP(\tau^{-1}).
\end{eqnarray*}
Note that
\begin{center}
\fcolorbox{purple}{lightgray}{
\parbox{12cm}{
\begin{eqnarray*}
\Delta_{d,k}(\I)^{-1} &=& \Pa{k!\sum_{\lambda\vdash_d
k}\frac{s_{\lambda}(1^{\times d})}{f^\lambda}\bsC_\lambda}^{-1} =
\frac1{k!}\sum_{\lambda\vdash_d k}
\frac{f^\lambda}{s_{\lambda}(1^{\times d})}\bsC_\lambda\\
&=& \sum_{\pi\in S_k}\Pa{\frac1{(k!)^2}\sum_{\lambda\vdash_d k}
\frac{(f^\lambda)^2}{s_{\lambda}(1^{\times d})}\chi_\lambda(\pi^{-1})}\bP(\pi)\\
&=& \sum_{\pi\in S_k}\mathrm{Wg}_{d,k}(\pi^{-1})\bP(\pi),
\end{eqnarray*}}}
\end{center}
where
$$
\bsC_\lambda := \sum_{\pi\in
S_k}\frac{f^\lambda}{k!}\chi_\lambda(\pi^{-1})\bP(\pi)
$$
is the minimal central projection and
$$
\mathrm{Wg}_{d,k} := \frac1{(k!)^2}\sum_{\lambda\vdash k}
\frac{(f^\lambda)^2}{s_{\lambda}(1^{\times d})}\chi_\lambda
$$
is the \emph{Weingarten function}.

Up to now, we can get
\begin{eqnarray*}
&&\Delta_{d,k}(\bsA)\Delta_{d,k}(\bsB)\Delta_{d,k}(\I)^{-1} \\
&&= \sum_{\sigma,\tau,\pi\in S_k} \iinner{i_1}{i'_{\sigma(1)}}\cdots
\iinner{i_k}{i'_{\sigma(k)}} \iinner{j_1}{j'_{\tau(1)}}\cdots
\iinner{j_k}{j'_{\tau(1)}}\mathrm{Wg}_{d,k}(\pi^{-1})
\bP(\sigma\tau^{-1}\pi)\\
&&= \sum_{\sigma,\tau\in S_k}\iinner{i_1}{i'_{\sigma(1)}}\cdots
\iinner{i_k}{i'_{\sigma(k)}} \iinner{j_1}{j'_{\tau(1)}}\cdots
\iinner{j_k}{j'_{\tau(1)}}\mathrm{Wg}_{d,k}(\sigma\tau^{-1})\I \\
&&~~~+\sum_{\sigma,\tau,\pi\in S_k: \sigma\tau^{-1}\pi\neq e}
\iinner{i_1}{i'_{\sigma(1)}}\cdots \iinner{i_k}{i'_{\sigma(k)}}
\iinner{j_1}{j'_{\tau(1)}}\cdots
\iinner{j_k}{j'_{\tau(1)}}\mathrm{Wg}_{d,k}(\pi^{-1})
\bP(\sigma\tau^{-1}\pi).
\end{eqnarray*}
Comparing both sides, we get
\begin{eqnarray*}
&&\int_{\U(d)} U_{i_1 j_1}\cdots U_{i_k j_k} \overline{U_{i'_1
j'_1}}\cdots\overline{U_{i'_k j'_k}}\dif\mu(\bsU) = \Tr{\bsA\sE_k(\bsB)} \\
&&= \sum_{\sigma,\tau\in S_k}\mathrm{Wg}_{d,k}(\sigma\tau^{-1})
\iinner{i_1}{i'_{\sigma(1)}}\cdots
\iinner{i_k}{i'_{\sigma(k)}}\iinner{j_1}{j'_{\tau(1)}}\cdots\iinner{j_k}{j'_{\tau(k)}}.
\end{eqnarray*}
This completes the proof.
\end{proof}

\begin{cor}
If $k\neq l$, then
$$
\int_{\mathsf{U}(d)}U_{i_1 j_1}\cdots U_{i_k j_k} \overline{U_{i'_1
j'_1}}\cdots\overline{U_{i'_l j'_l}}\dif\mu(\bsU) = 0.
$$
\end{cor}

\begin{proof}
For every $z\in \U(1)$, the map $\U(d)\ni \bsU\mapsto z\bsU\in\U(d)$
is measure-preserving, therefore
\begin{eqnarray}
&&\int_{\U(d)}U_{i_1 j_1}\cdots U_{i_k j_k} \overline{U_{i'_1
j'_1}}\cdots\overline{U_{i'_l j'_l}}\dif\mu(\bsU) \\
&&= \int_{\U(d)}zU_{i_1 j_1}\cdots zU_{i_k j_k} \overline{zU_{i'_1
j'_1}}\cdots\overline{zU_{i'_l j'_l}}\dif\mu(\bsU)\\
&&= z^{k-l}\int_{\U(d)}U_{i_1 j_1}\cdots U_{i_k j_k}
\overline{U_{i'_1 j'_1}}\cdots\overline{U_{i'_l j'_l}}\dif\mu(\bsU),
\end{eqnarray}
implying that
$$
(1-z^{k-l})\int_{\U(d)}U_{i_1 j_1}\cdots U_{i_k j_k}
\overline{U_{i'_1 j'_1}}\cdots\overline{U_{i'_l j'_l}}\dif\mu(\bsU)
= 0.
$$
By the arbitrariness of $z\in\U(1)$, there exists a $z_0\in\U(1)$
such that $z^{k-l}_0\neq1$ since $k\neq l$.
\end{proof}

\begin{remark}
What is $\int_{\U(d)} \bsU\dif\mu(\bsU)$? One approach to see this
is to form a $d\times d$ matrix $\bsM$ whose $(i,j)$-th entry is the
$\int_{\U(d)} U_{ij}\dif\mu(\bsU)$, for $1\leqslant i,j\leqslant d$.
Writing this in terms of matrix form, we have for any fixed $\bsV\in
\U(d)$,
$$
\bsM = \int_{\U(d)} \bsU\dif\mu(\bsU) = \int_{\U(d)}
\bsV\bsU\dif\mu(\bsU) = \bsV\int_{\U(d)} \bsU\dif\mu(\bsU) =
\bsV\bsM,
$$
where we used the fact that $\dif\mu(\bsU)$ is regular. But
$\bsV\bsM=\bsM$ for all unitary $\bsV$ can only hold if
$\bsM=\mathbf{0}$. That is,
\begin{center}
\fcolorbox{purple}{lightgray}{
\parbox{5cm}{
$$
\int_{\U(d)} \bsU\dif\mu(\bsU) = \mathbf{0}.
$$
}}
\end{center}
\end{remark}

\begin{center}
\fcolorbox{purple}{lightgray}{
\parbox{16.3cm}{
\begin{cor}
For two distinct nonnegative integers $k\geqslant1$ or
$l\geqslant1$, it holds that
\begin{eqnarray}
\int_{\mathsf{U}(d)} \bsU^{\ot k}\ot (\bsU^{\ot
l})^\dagger\dif\mu(\bsU) =\mathbf{0}.
\end{eqnarray}
In particular, for $l=0$,
\begin{eqnarray}
\int_{\mathsf{U}(d)} \Tr{\bsU}^k \dif\mu(\bsU)
=0=\int_{\mathsf{U}(d)} \Tr{\bsU^k} \dif\mu(\bsU).
\end{eqnarray}
\end{cor}
}}
\end{center}

\begin{cor}
For the integer $k\geqslant 1$, it holds that
\begin{eqnarray}
\int_{\mathsf{U}(d)} \det(\bsU)^k \dif\mu(\bsU) =0.
\end{eqnarray}
\end{cor}

\begin{cor}
It holds that
\begin{eqnarray}
\int_{\mathsf{U}(d)} \bsU^{\ot k}\ot \Pa{\bsU^{\ot k}}^\dagger
\dif\mu(\bsU) = \sum_{\sigma,\tau\in S_k}
\mathrm{Wg}_{d,k}(\sigma\tau^{-1})\bP_{\tau+k,\sigma^{-1}},
\end{eqnarray}
where, for any $\pi_1,\pi_2\in S_k$,
\begin{eqnarray}
\bP_{\pi_1+k,\pi_2}\ket{j_1\cdots j_k i'_1\cdots i'_k} :=
\Ket{i'_{\pi_2^{-1}(1)}\cdots
i'_{\pi_2^{-1}(k)}j_{\pi_1^{-1}(1)}\cdots j_{\pi_1^{-1}(k)}}.
\end{eqnarray}
\end{cor}

\begin{proof}
Clearly
\begin{eqnarray*}
&&\Innerm{\mathbf{ij}'}{\int_{\U(d)} \bsU^{\ot k}\ot \Pa{\bsU^{\ot
k}}^\dagger \dif\mu(\bsU)}{\mathbf{ji}'} = \int_{\U(d)} U_{i_1
j_1}\cdots U_{i_k j_k} \overline{U_{i'_1
j'_1}}\cdots\overline{U_{i'_k
j'_k}}\dif\mu(\bsU) \\
&&= \sum_{\sigma,\tau\in S_k}\mathrm{Wg}_{d,k}(\sigma\tau^{-1})
\iinner{i_1}{i'_{\sigma(1)}}\cdots
\iinner{i_k}{i'_{\sigma(k)}}\iinner{j'_1}{j_{\tau^{-1}(1)}}\cdots\iinner{j'_k}{j_{\tau^{-1}(k)}}.
\end{eqnarray*}
Next, by the definition of $\bP_{\pi_1+k,\pi_2}$, hence we get
\begin{eqnarray}
\bP_{\tau+k,\sigma^{-1}}\ket{j_1\cdots j_k i'_1\cdots i'_k} =
\Ket{i'_{\sigma(1)}\cdots i'_{\sigma(k)}j_{\tau^{-1}(1)}\cdots
j_{\tau^{-1}(k)}},
\end{eqnarray}
where $\pi+k$ means
$$
\pi+k\equiv\Pa{\begin{array}{ccc}
                 1 & \cdots & k \\
                 \pi(1)+k & \cdots & \pi(k)+k
               \end{array}
},
$$
implying that
\begin{eqnarray*}
\Innerm{\mathbf{ij}'}{\int_{\U(d)} \bsU^{\ot k}\ot \Pa{\bsU^{\ot
k}}^\dagger \dif\mu(\bsU)}{\mathbf{ji}'} =
\Innerm{\mathbf{ij}'}{\sum_{\sigma,\tau\in
S_k}\mathrm{Wg}_{d,k}(\sigma\tau^{-1})\bP_{\tau+k,\sigma^{-1}}}{\mathbf{ji}'}.
\end{eqnarray*}
The proof is complete.
\end{proof}

\begin{remark}
In recent papers \cite{Horodecki}, the authors modified the
Schur-Weyl duality in the sense that the commutant of $\bsU^{\ot
k-1}\ot \overline{\bsU}$ can be specifically computed. They make an
attempt in \cite{mhs,schm} to use the obtained new commutant theorem
investigate some questions in quantum information theory.
\end{remark}

\begin{center}
\fcolorbox{purple}{lightgray}{
\parbox{16.3cm}{
\begin{cor}
The uniform average of $\out{\psi}{\psi}^{\ot k}$ over unit vectors
$\ket{\psi}$ in $\complex^d$ is given by
\begin{eqnarray}
\int_{\complex^d} \out{\psi}{\psi}^{\ot k} \dif\mu(\psi) =
\frac1{s_{(k)}(1^{\times d})}\bsC_{(k)} =
\frac1{\binom{k+d-1}{k}}\bsC_{(k)},
\end{eqnarray}
where the meaning of $\bsC_\lambda$ can be referred to
Eq.~\eqref{eq:cproj} or Eq.~\eqref{eq:central-proj}, here
$\lambda=(k)$.
\end{cor}}}
\end{center}

\begin{proof}
In fact, this result is a direct consequence of
\eqref{eq:general-integral}. More explicitly, let us fix a vector
$\ket{\psi_0}$. Then every $\ket{\psi}$ can be generated by a
uniform unitary $\bsU$ such that $\ket{\psi}=\bsU\ket{\psi_0}$. Thus
\begin{eqnarray*}
\int_{\complex^d} \out{\psi}{\psi}^{\ot k} \dif\mu(\psi) &=&
\int_{\U(d)} \bsU^{\ot k}\out{\psi_0}{\psi_0}^{\ot k}\bsU^{\ot
k,\dagger}\dif\mu(\bsU).
\end{eqnarray*}
Taking $\bsA=\out{\psi_0}{\psi_0}^{\ot k}$ in
\eqref{eq:general-integral} gives rise to
\begin{eqnarray*}
&&\int_{\U(d)} \bsU^{\ot k}\out{\psi_0}{\psi_0}^{\ot k}\bsU^{\ot
k,\dagger}\dif\mu(\bsU) = \Pa{\frac1{k!}\sum_{\pi\in S_k}
\bP(\pi)}\Pa{\sum_{\lambda\vdash_d
k}\frac{f^\lambda}{s_{\lambda}(1^{\times d})}\bsC_\lambda}\\
&&= \bsC_{(k)}\Pa{\sum_{\lambda\vdash_d
k}\frac{f^\lambda}{s_{\lambda}(1^{\times d})}\bsC_\lambda} =
\frac{f^{(k)}}{s_{(k)}(1^{\times d})}\bsC_{(k)},
\end{eqnarray*}
implying the desired result.
\end{proof}

The following compact version of $\sE^{(k)}_d(\bsA)$ is given by
Audenaert in \cite{kmra}. The detailed presentation is shifted to
Appendix (see below).

\begin{prop}
Let $\cH_{\mathrm{in}}$ and $\cH_{\mathrm{out}}$ be two copies of
the Hilbert space $\cH = (\complex^d)^{\ot k}$. Let
$\bsC_{(k)}^\vee$ be the projector on the totally symmetric subspace
of $\cH_{\mathrm{out}}\ot \cH_{\mathrm{in}}$. Then it holds that
\begin{eqnarray}
\int_{\mathsf{U}(d)} \out{\bsU^{\ot k}}{\bsU^{\ot k}} \dif\mu(\bsU)=
\bsC_{(k)}^\vee\Pa{\Br{\Ptr{\mathrm{in}}{\bsC_{(k)}^\vee}}^{-1}\ot\I_{\mathrm{in}}}.
\end{eqnarray}
\end{prop}

\begin{cor}
Let $\bsA\in\End((\complex^d)^{\ot k})$. Then it holds that
\begin{eqnarray}
\int_{\mathsf{U}(d)} \Pa{\bsU^{\ot k}}\bsA\Pa{\bsU^{\ot k}}^\dagger
\dif\mu(\bsU)= \frac1{k!}\Br{\sum_{\pi\in
S_k}\Tr{\bsA\bP(\pi^{-1})}\bP(\pi)}\Br{\Ptr{\mathrm{in}}{\bsC_{(k)}^\vee}}^{-1}.
\end{eqnarray}
\end{cor}

\begin{center}
\fcolorbox{purple}{lightgray}{
\parbox{16.3cm}{
\begin{cor}
Assume $\bsX\in\End(\complex^d)$ with spectrum
$\set{x_j:1,\ldots,d}$. It holds that
\begin{eqnarray}
\int_{\mathsf{U}(d)} \Pa{\bsU\bsX\bsU^\dagger}^{\ot k}
\dif\mu(\bsU)=\sum_{\lambda\vdash_dk}\frac{s_\lambda(x_1,\ldots,x_d)}{s_\lambda(1^{\times
d})}\bsC_\lambda=\sum_{\lambda\vdash_dk}\frac{\Tr{\bsC_\lambda
\bsX^{\ot k}}}{\Tr{\bsC_\lambda}}\bsC_\lambda.
\end{eqnarray}
\end{cor}}}
\end{center}

\begin{proof}
We give a very simple derivation of this identity via Schur-Weyl duality , i.e. Theorem~\ref{th:S-W-Duality}. Indeed, the mentioned integral can be rewritten as
\begin{eqnarray}
\int_{\U(d)} \Pa{\bsU^{\ot k}}\bsX^{\ot k}\Pa{\bsU^{\ot k}}^\dagger
\dif\mu(\bsU)=\int_{\U(d)} \bQ(\bsU)\bQ(\bsX) \bQ^\dagger(\bsU)
\dif\mu(\bsU).
\end{eqnarray}
Now by Schur-Weyl duality, we have
\begin{eqnarray*}
\bQ(\bsU) \cong \bigoplus_{\lambda\vdash_dk}
\bQ_\lambda(\bsU)\ot\I_{\bP_\lambda},~~\bQ(\bsX) \cong
\bigoplus_{\lambda\vdash_dk}
\bQ_\lambda(\bsX)\ot\I_{\bP_\lambda},~~\bQ^\dagger(\bsU) \cong
\bigoplus_{\lambda\vdash_dk}
\bQ^\dagger_\lambda(\bsU)\ot\I_{\bP_\lambda}.
\end{eqnarray*}
Thus
\begin{eqnarray*}
\int_{\U(d)} \bQ(\bsU)\bQ(\bsX) \bQ^\dagger(\bsU) \dif\mu(\bsU)
&\cong& \sum_{\lambda\vdash_dk} \Pa{\int_{\U(d)}
\bQ_\lambda(\bsU)\bQ_\lambda(\bsX)
\bQ^\dagger_\lambda(\bsU)\dif\mu(\bsU)} \ot \I_{\bP_\lambda} \\
&=& \sum_{\lambda\vdash_dk} \Pa{\frac1{\dim(\bQ_\lambda)}\Tr{\bQ_\lambda(\bsX)}\I_{\bQ_\lambda}} \ot \I_{\bP_\lambda}\\
&=& \sum_{\lambda\vdash_dk}
\frac1{\dim(\bQ_\lambda)}\Tr{\bQ_\lambda(\bsX)} \I_{\bQ_\lambda}\ot
\I_{\bP_\lambda},
\end{eqnarray*}
which implies the desired result, where we have used the facts that
\begin{eqnarray*}
\Tr{\bQ_\lambda(\bsX)} &=& s_\lambda(x_1,\ldots,x_d), \\
\dim(\bQ_\lambda)&=& s_\lambda(1^{\times d}),~ \bsC_\lambda \cong
\I_{\bQ_\lambda} \ot \I_{\bP_\lambda}.
\end{eqnarray*}
This completes the proof.
\end{proof}
\begin{exam}
Note that we get the following decomposition via Schur-Weyl duality
\begin{eqnarray}
(\complex^d)^{\ot3}\cong
\bQ_{(3)}\ot\bP_{(3)}\bigoplus\bQ_{(2,1)}\ot\bP_{(2,1)}\bigoplus\bQ_{(1,1,1)}\ot\bP_{(1,1,1)}
\end{eqnarray}
where
\begin{eqnarray}
\dim(\bQ_\lambda)=
\begin{cases}
\frac{d(d+1)(d+2)}6,&\text{if}~
\lambda=(3),\\
\frac{(d-1)d(d+1)}3, &\text{if}~ \lambda=(2,1),\\
\frac{(d-2)(d-1)d}6,&\text{if}~ \lambda=(1,1,1),
\end{cases}~\text{and}~\dim(\bP_\lambda)=
\begin{cases}
1,&\text{if}~
\lambda=(3),\\
2, &\text{if}~ \lambda=(2,1),\\
1,&\text{if}~ \lambda=(1,1,1).
\end{cases}
\end{eqnarray}
Hence
\begin{eqnarray}
\bsC_\lambda =
\begin{cases}
\frac16\Pa{\bP_{(1)}+\bP_{(12)}+\bP_{(13)}+\bP_{(23)}+\bP_{(123)}+\bP_{(132)}},&\text{if}~
\lambda=(3),\\
\frac13\Pa{2\bP_{(1)}-\bP_{(123)}-\bP_{(132)}},&\text{if}~\lambda=(2,1),\\
\frac16\Pa{\bP_{(1)}-\bP_{(12)}-\bP_{(13)}-\bP_{(23)}+\bP_{(123)}+\bP_{(132)}},&\text{if}~\lambda=(1,1,1).
\end{cases}
\end{eqnarray}
It follows that
\begin{eqnarray}
\Tr{\bsC_\lambda} =
\begin{cases}
\frac{d(d+1)(d+2)}6,&\text{if}~
\lambda=(3),\\
\frac{2(d-1)d(d+1)}3,&\text{if}~\lambda=(2,1),\\
\frac{(d-2)(d-1)d}6,&\text{if}~\lambda=(1,1,1)
\end{cases}
\end{eqnarray}
and
\begin{eqnarray}
\Tr{\bsX^{\ot 3}\bsC_\lambda} =
\begin{cases}
\frac16\Br{\Tr{\bsX}^3+3\Tr{\bsX^2}\Tr{\bsX}+2\Tr{\bsX^3}},&\text{if}~
\lambda=(3),\\
\frac23\Br{\Tr{\bsX}^3-\Tr{\bsX^3}},&\text{if}~\lambda=(2,1),\\
\frac16\Br{\Tr{\bsX}^3-3\Tr{\bsX^2}\Tr{\bsX}+2\Tr{\bsX^3}},&\text{if}~\lambda=(1,1,1).
\end{cases}
\end{eqnarray}
Therefore
\begin{eqnarray}
\int (\bsU\bsX\bsU^\dagger)^{\ot3}\dif\mu(\bsU) =
\delta^{(3)}_3\bsC_{(3)} + \delta^{(2,1)}_3\bsC_{(2,1)}
+\delta^{(1,1,1)}_3\bsC_{(1,1,1)},
\end{eqnarray}
where
\begin{eqnarray*}
\delta^{(3)}_3 &:=& \frac{\Tr{\bsX}^3+3\Tr{\bsX^2}\Tr{\bsX}+2\Tr{\bsX^3}}{d(d+1)(d+2)},\\
\delta^{(2,1)}_3 &:=&\frac{\Tr{\bsX}^3-\Tr{\bsX^3}}{(d-1)d(d+1)},\\
\delta^{(1,1,1)}_3
&:=&\frac{\Tr{\bsX}^3-3\Tr{\bsX^2}\Tr{\bsX}+2\Tr{\bsX^3}}{(d-2)(d-1)d}.
\end{eqnarray*}
\end{exam}

\begin{exam}
Similar we get the following decomposition:
\begin{eqnarray}
(\complex^d)^{\ot 4}&\cong& \bQ_{(4)}\ot\bP_{(4)}\bigoplus
\bQ_{(3,1)}\ot\bP_{(3,1)}\bigoplus
\bQ_{(2,2)}\ot\bP_{(2,2)}\nonumber\\
&&\bigoplus \bQ_{(2,1,1)}\ot\bP_{(2,1,1)}\bigoplus
\bQ_{(1,1,1,1)}\ot\bP_{(1,1,1,1)},
\end{eqnarray}
where
\begin{eqnarray}
\dim(\bQ_\lambda)=
\begin{cases}
\frac{d(d+1)(d+2)(d+3)}{24},&\text{if}~
\lambda=(4),\\
\frac{(d-1)d(d+1)(d+2)}8, &\text{if}~ \lambda=(3,1),\\
\frac{(d-1)d^2(d+1)}{12},&\text{if}~
\lambda=(2,2),\\
\frac{(d-2)(d-1)d(d+1)}8,&\text{if}~ \lambda=(2,1,1),\\
\frac{(d-3)(d-2)(d-1)d}{24},&\text{if}~ \lambda=(1,1,1,1),
\end{cases}~\text{and}~\dim(\bP_\lambda)=
\begin{cases}
1,&\text{if}~
\lambda=(4),\\
3, &\text{if}~ \lambda=(3,1),\\
2,&\text{if}~
\lambda=(2,2),\\
3,&\text{if}~ \lambda=(2,1,1),\\
1,&\text{if}~ \lambda=(1,1,1,1).
\end{cases}
\end{eqnarray}
Hence
\begin{eqnarray*}
\bsC_{(4)} &=&
\frac1{24}\bP_{(1)}+\frac1{24}\Pa{\bP_{(12)}+\bP_{(13)}+\bP_{(14)}+\bP_{(23)}+\bP_{(24)}+\bP_{(34)}}\\
&&+\frac1{24}\Pa{\bP_{(12)(34)}+\bP_{(13)(24)}+\bP_{(14)(23)}}\\
&&+\frac1{24}\Pa{\bP_{(123)}+\bP_{(132)}+\bP_{(124)}+\bP_{(142)}+\bP_{(134)}+\bP_{(143)}+\bP_{(234)}+\bP_{(243)}}\\
&&+\frac1{24}\Pa{\bP_{(1234)}+\bP_{(1243)}+\bP_{(1324)}+\bP_{(1342)}+\bP_{(1423)}+\bP_{(1432)}}
\end{eqnarray*}

\begin{eqnarray*}
\bsC_{(3,1)} &=&
\frac38\bP_{(1)}+\frac18\Pa{\bP_{(12)}+\bP_{(13)}+\bP_{(14)}+\bP_{(23)}+\bP_{(24)}+\bP_{(34)}}\\
&&-\frac18\Pa{\bP_{(12)(34)}+\bP_{(13)(24)}+\bP_{(14)(23)}}\\
&&-\frac18\Pa{\bP_{(1234)}+\bP_{(1243)}+\bP_{(1324)}+\bP_{(1342)}+\bP_{(1423)}+\bP_{(1432)}}
\end{eqnarray*}

\begin{eqnarray*}
\bsC_{(2,2)} &=& \frac16\bP_{(1)}+\frac16\Pa{\bP_{(12)(34)}+\bP_{(13)(24)}+\bP_{(14)(23)}}\\
&&-\frac1{12}\Pa{\bP_{(123)}+\bP_{(132)}+\bP_{(124)}+\bP_{(142)}+\bP_{(134)}+\bP_{(143)}+\bP_{(234)}+\bP_{(243)}}
\end{eqnarray*}

\begin{eqnarray*}
\bsC_{(2,1,1)} &=&
\frac38\bP_{(1)}-\frac18\Pa{\bP_{(12)}+\bP_{(13)}+\bP_{(14)}+\bP_{(23)}+\bP_{(24)}+\bP_{(34)}}\\
&&-\frac18\Pa{\bP_{(12)(34)}+\bP_{(13)(24)}+\bP_{(14)(23)}}\\
&&+\frac18\Pa{\bP_{(1234)}+\bP_{(1243)}+\bP_{(1324)}+\bP_{(1342)}+\bP_{(1423)}+\bP_{(1432)}}
\end{eqnarray*}

\begin{eqnarray*}
\bsC_{(1,1,1,1)} &=&
\frac1{24}\bP_{(1)}-\frac1{24}\Pa{\bP_{(12)}+\bP_{(13)}+\bP_{(14)}+\bP_{(23)}+\bP_{(24)}+\bP_{(34)}}\\
&&+\frac1{24}\Pa{\bP_{(12)(34)}+\bP_{(13)(24)}+\bP_{(14)(23)}}\\
&&+\frac1{24}\Pa{\bP_{(123)}+\bP_{(132)}+\bP_{(124)}+\bP_{(142)}+\bP_{(134)}+\bP_{(143)}+\bP_{(234)}+\bP_{(243)}}\\
&&-\frac1{24}\Pa{\bP_{(1234)}+\bP_{(1243)}+\bP_{(1324)}+\bP_{(1342)}+\bP_{(1423)}+\bP_{(1432)}}
\end{eqnarray*}

\begin{eqnarray*}
\int (\bsU\bsX\bsU^\dagger)^{\ot4}\dif\mu(\bsU) =
\delta^{(4)}_4\bsC_{(4)}+\delta^{(3,1)}_4\bsC_{(3,1)}
+\delta^{(2,2)}_4\bsC_{(2,2)}+\Delta^{(2,1,1)}_4\bsC_{(2,1,1)}
+\delta^{(1,1,1,1)}_4\bsC_{(1,1,1,1)},
\end{eqnarray*}
where
\begin{eqnarray*}
\delta^{(4)}_4&:=&\frac{\Tr{\bsX}^4+6\Tr{\bsX^2}\Tr{\bsX}^2+3\Tr{\bsX^2}^2
+8\Tr{\bsX^3}\Tr{\bsX}+6\Tr{\bsX^4}}{d(d+1)(d+2)(d+3)},\\
\delta^{(3,1)}_4&:=&\frac{\Tr{\bsX}^4+2\Tr{\bsX^2}\Tr{\bsX}^2-\Tr{\bsX^2}^2-2\Tr{\bsX^4}}{(d-1)d(d+1)(d+2)},\\
\delta^{(2,2)}_4&:=&\frac{\Tr{\bsX}^4+3\Tr{\bsX^2}^2-4\Tr{\bsX^3}\Tr{\bsX}}{(d-1)d^2(d+1)},\\
\delta^{(2,1,1)}_4&:=&\frac{\Tr{\bsX}^4-2\Tr{\bsX^2}\Tr{\bsX}^2-\Tr{\bsX^2}^2+2\Tr{\bsX^4}}{(d-2)(d-1)d(d+1)},\\
\delta^{(1,1,1,1)}_4&:=&\frac{\Tr{\bsX}^4-6\Tr{\bsX^2}\Tr{\bsX}^2+3\Tr{\bsX^2}^2+8\Tr{\bsX^3}\Tr{\bsX}-6\Tr{\bsX^4}}{(d-3)(d-2)(d-1)d}.
\end{eqnarray*}
\end{exam}

\subsection{The case where $k=3$}

In general, the explicit formula of $\cE^{(3)}_d$ is relatively
seldom used, compared with the formulae $\cE^{(1)}_d$ and $\cE^{(2)}_d$. But it
still deserves to write it down explicitly.
\begin{center}
\fcolorbox{purple}{lightgray}{
\parbox{16.3cm}{
\begin{prop}\label{prop:uuu-integral}
Let
$S_3=\set{\pi_1=(1),\pi_2=(12),\pi_3=(13),\pi_4=(23),\pi_5=(123),\pi_6=(132)}$.
Denote by $a_j=\Tr{\bsA\bP(\pi^{-1}_j)}$. It holds that
\begin{eqnarray}
\int_{\mathsf{U}(d)}\bsU^{\ot 3}\bsA(\bsU^{\ot 3})^\dagger
\dif\mu(\bsU) = \sum_{\pi\in S_3} c_\pi(\bsA)\bP(\pi) =
\sum^6_{j=1}c_j(\bsA)\bP(\pi_j),
\end{eqnarray}
where $\bsA\in M_{d^3}(\complex)$, and
\begin{eqnarray*}
\begin{cases}
c_1(\bsA) = a_1\mathrm{Wg}_{d,3}(K_1)+(a_2+a_3+a_4)\mathrm{Wg}_{d,3}(K_2)+(a_5+a_6)\mathrm{Wg}_{d,3}(K_3),\\
c_2(\bsA) = a_2\mathrm{Wg}_{d,3}(K_1)+(a_1+a_5+a_6)\mathrm{Wg}_{d,3}(K_2)+(a_3+a_4)\mathrm{Wg}_{d,3}(K_3),\\
c_3(\bsA) = a_3\mathrm{Wg}_{d,3}(K_1)+(a_1+a_5+a_6)\mathrm{Wg}_{d,3}(K_2)+(a_2+a_4)\mathrm{Wg}_{d,3}(K_3),\\
c_4(\bsA) = a_4\mathrm{Wg}_{d,3}(K_1)+(a_1+a_5+a_6)\mathrm{Wg}_{d,3}(K_2)+(a_2+a_3)\mathrm{Wg}_{d,3}(K_3),\\
c_5(\bsA) = a_5\mathrm{Wg}_{d,3}(K_1)+(a_2+a_3+a_4)\mathrm{Wg}_{d,3}(K_2)+(a_1+a_6)\mathrm{Wg}_{d,3}(K_3).\\
c_6(\bsA) =
a_6\mathrm{Wg}_{d,3}(K_1)+(a_2+a_3+a_4)\mathrm{Wg}_{d,3}(K_2)+(a_1+a_5)\mathrm{Wg}_{d,3}(K_3).
\end{cases}
\end{eqnarray*}
Here
\begin{eqnarray*}
\mathrm{Wg}_{d,3}(K_1)&=&\frac{d^2-2}{(d-2) (d-1) d (d+1) (d+2)},\\
\mathrm{Wg}_{d,3}(K_2)&=&-\frac1{(d-2)(d-1)(d+1) (d+2)},\\
\mathrm{Wg}_{d,3}(K_3)&=&\frac2{(d-2)(d-1)d(d+1) (d+2)}.
\end{eqnarray*}
\end{prop}}}
\end{center}

\begin{proof}
From Proposition~\ref{prop:Deltadk}, we see that
\begin{eqnarray*}
\int_{\mathsf{U}(d)}\bsU^{\ot 3}\bsA(\bsU^{\ot 3})^\dagger
\dif\mu(\bsU) = \Pa{\sum_{\pi\in S_3}
\Tr{\bsA\bP(\pi^{-1})}\bP(\pi)}\Pa{\sum_{\pi\in
S_3}\mathrm{Wg}_{d,3}(\pi)\bP(\pi)},
\end{eqnarray*}
where Weingarten function $\mathrm{Wg}_{d,3}$ is given by
\begin{eqnarray*}
\mathrm{Wg}_{d,3}(K_j) &=& \frac1{(3!)^2}\sum_{\lambda\vdash
3}\frac{\dim(\bP_\lambda)^2}{\dim(\bQ_\lambda)}\chi_\lambda(K_j),\quad
(j=1,2,3).
\end{eqnarray*}
By hook-length formula and hook-content formula, we get that
\begin{eqnarray*}
\frac{\dim(\bP_{\lambda_1})^2}{\dim(\bQ_{\lambda_1})} &=&
\frac{6}{d(d+1)(d+2)},\\
\frac{\dim(\bP_{\lambda_2})^2}{\dim(\bQ_{\lambda_2})} &=&
\frac{12}{(d-1)d(d+1)},\\
\frac{\dim(\bP_{\lambda_3})^2}{\dim(\bQ_{\lambda_3})} &=&
\frac{6}{(d-2)(d-1)d}.
\end{eqnarray*}
From Table~\ref{tab:S3}, we substitute these into Weingarten
function, we get that
\begin{eqnarray*}
\mathrm{Wg}_{d,3}(K_1) &=&
\frac1{36}\Pa{1\frac{6}{d(d+1)(d+2)}+2\frac{12}{(d-1)d(d+1)}+1\frac{6}{(d-2)(d-1)d}}\\
&=&\frac{d^2-2}{(d-2) (d-1) d (d+1) (d+2)},\\
\mathrm{Wg}_{d,3}(K_2) &=&
\frac1{36}\Pa{1\frac{6}{d(d+1)(d+2)}+0\frac{12}{(d-1)d(d+1)}+(-1)\frac{6}{(d-2)(d-1)d}}\\
&=&-\frac1{(d-2)(d-1)(d+1) (d+2)},\\
\mathrm{Wg}_{d,3}(K_3) &=&
\frac1{36}\Pa{1\frac{6}{d(d+1)(d+2)}+(-1)\frac{12}{(d-1)d(d+1)}+1\frac{6}{(d-2)(d-1)d}}\\
&=&\frac2{(d-2)(d-1)d(d+1) (d+2)}.
\end{eqnarray*}
Let $a_j=\Tr{\bsA\bP(\pi^{-1}_j)}$ and
$w_j=\mathrm{Wg}_{d,3}(\pi_j)$, where
$w_1=\mathrm{Wg}_{d,3}(K_1),w_2=w_3=w_4=\mathrm{Wg}_{d,3}(K_2)$ and
$w_5=w_6=\mathrm{Wg}_{d,3}(K_3)$. Now from the following
multiplication table of $S_3$, we see that
\begin{table}[h]
\centering
\begin{tabular}{|c|c|c|c|c|c|c|}
  \hline & $\pi_1$ & $\pi_2$ & $\pi_3$ & $\pi_4$ & $\pi_5$ & $\pi_6$\\
 \hline $\pi_1$ & $\pi_1$ & $\pi_2$ & $\pi_3$ & $\pi_4$ & $\pi_5$ & $\pi_6$\\
  \hline $\pi_2$ & $\pi_2$ & $\pi_1$ & $\pi_6$ & $\pi_5$ & $\pi_4$ & $\pi_3$\\
  \hline $\pi_3$ & $\pi_3$ & $\pi_5$ & $\pi_1$ & $\pi_6$ & $\pi_2$ & $\pi_4$\\
  \hline $\pi_4$ & $\pi_4$ & $\pi_6$ & $\pi_5$ & $\pi_1$ & $\pi_3$ & $\pi_2$\\
  \hline $\pi_5$ & $\pi_5$ & $\pi_3$ & $\pi_4$ & $\pi_2$ & $\pi_6$ & $\pi_1$\\
  \hline $\pi_6$ & $\pi_6$ & $\pi_4$ & $\pi_2$ & $\pi_3$ & $\pi_1$ & $\pi_5$\\
  \hline
\end{tabular}\caption{The multiplication table of $S_3$.}
\end{table}

\begin{eqnarray*}
c_1(\bsA) &=&a_1w_1+a_2w_2+a_3w_3+a_4w_4+a_5w_6+a_6w_5\\
&=&a_1\mathrm{Wg}_{d,3}(K_1)+(a_2+a_3+a_4)\mathrm{Wg}_{d,3}(K_2)+(a_5+a_6)\mathrm{Wg}_{d,3}(K_3),\\
c_2(\bsA) &=& a_1w_2+a_2w_1+a_3w_5+a_4w_6+a_5w_4+a_6w_3\\
&=&a_2\mathrm{Wg}_{d,3}(K_1)+(a_1+a_5+a_6)\mathrm{Wg}_{d,3}(K_2)+(a_3+a_4)\mathrm{Wg}_{d,3}(K_3),\\
c_3(\bsA) &=& a_1 w_3+a_2 w_6+a_3 w_1+a_4 w_5+a_5 w_2+a_6 w_4\\
&=&a_3\mathrm{Wg}_{d,3}(K_1)+(a_1+a_5+a_6)\mathrm{Wg}_{d,3}(K_2)+(a_2+a_4)\mathrm{Wg}_{d,3}(K_3),\\
c_4(\bsA) &=& a_1 w_4 +a_2 w_5+a_3 w_6+a_4 w_1+a_5 w_3+a_6 w_2\\
&=& a_4\mathrm{Wg}_{d,3}(K_1)+(a_1+a_5+a_6)\mathrm{Wg}_{d,3}(K_2)+(a_2+a_3)\mathrm{Wg}_{d,3}(K_3),\\
c_5(\bsA) &=& a_1 w_5 +a_2 w_4+a_3 w_2+a_4 w_3+a_5 w_1+a_6 w_6 \\
&=& a_5\mathrm{Wg}_{d,3}(K_1)+(a_2+a_3+a_4)\mathrm{Wg}_{d,3}(K_2)+(a_1+a_6)\mathrm{Wg}_{d,3}(K_3),\\
c_6(\bsA) &=& a_1 w_6 +a_2 w_3+a_3 w_4+a_4 w_2+a_5 w_5+a_6 w_1 \\
&=&
a_6\mathrm{Wg}_{d,3}(K_1)+(a_2+a_3+a_4)\mathrm{Wg}_{d,3}(K_2)+(a_1+a_5)\mathrm{Wg}_{d,3}(K_3).
\end{eqnarray*}
By simplifying them, we get the desired result.
\end{proof}

\begin{center}
\fcolorbox{purple}{lightgray}{
\parbox{16.3cm}{
\begin{cor}
Let
$S_3=\set{\pi_1=(1),\pi_2=(12),\pi_3=(13),\pi_4=(23),\pi_5=(123),\pi_6=(132)}$.
Denote by $a_j=\Tr{(\bsX\ot\bsY\ot\bsZ)\bP(\pi^{-1}_j)}$. It holds
that
\begin{eqnarray}
\int_{\mathsf{U}(d)}\bsU\bsX\bsU^\dagger\ot
\bsU\bsY\bsU^\dagger\ot\bsU\bsZ\bsU^\dagger \dif\mu(\bsU) =
\sum^6_{j=1}c_j(\bsX,\bsY,\bsZ)\bP(\pi_j),
\end{eqnarray}
where $\bsX,\bsY,\bsZ\in M_d(\complex)$, and
\begin{eqnarray*}
\begin{cases}
c_1(\bsX,\bsY,\bsZ) = a_1\mathrm{Wg}_{d,3}(K_1)+(a_2+a_3+a_4)\mathrm{Wg}_{d,3}(K_2)+(a_5+a_6)\mathrm{Wg}_{d,3}(K_3),\\
c_2(\bsX,\bsY,\bsZ) = a_2\mathrm{Wg}_{d,3}(K_1)+(a_1+a_5+a_6)\mathrm{Wg}_{d,3}(K_2)+(a_3+a_4)\mathrm{Wg}_{d,3}(K_3),\\
c_3(\bsX,\bsY,\bsZ) = a_3\mathrm{Wg}_{d,3}(K_1)+(a_1+a_5+a_6)\mathrm{Wg}_{d,3}(K_2)+(a_2+a_4)\mathrm{Wg}_{d,3}(K_3),\\
c_4(\bsX,\bsY,\bsZ) = a_4\mathrm{Wg}_{d,3}(K_1)+(a_1+a_5+a_6)\mathrm{Wg}_{d,3}(K_2)+(a_2+a_3)\mathrm{Wg}_{d,3}(K_3),\\
c_5(\bsX,\bsY,\bsZ) = a_5\mathrm{Wg}_{d,3}(K_1)+(a_2+a_3+a_4)\mathrm{Wg}_{d,3}(K_2)+(a_1+a_6)\mathrm{Wg}_{d,3}(K_3).\\
c_6(\bsX,\bsY,\bsZ) =
a_6\mathrm{Wg}_{d,3}(K_1)+(a_2+a_3+a_4)\mathrm{Wg}_{d,3}(K_2)+(a_1+a_5)\mathrm{Wg}_{d,3}(K_3).
\end{cases}
\end{eqnarray*}
Here
\begin{eqnarray*}
\mathrm{Wg}_{d,3}(K_1)&=&\frac{d^2-2}{(d-2) (d-1) d (d+1) (d+2)},\\
\mathrm{Wg}_{d,3}(K_2)&=&-\frac1{(d-2)(d-1)(d+1) (d+2)},\\
\mathrm{Wg}_{d,3}(K_3)&=&\frac2{(d-2)(d-1)d(d+1) (d+2)}.
\end{eqnarray*}
\end{cor}}}
\end{center}
It is easily seen that
\begin{eqnarray*}
a_1=\Tr{\bsX}\Tr{\bsY}\Tr{\bsZ},a_2=\Tr{\bsX\bsY}\Tr{\bsZ},a_3=\Tr{\bsX\bsZ}\Tr{\bsY},\\
a_4=\Tr{\bsY\bsZ}\Tr{\bsX},a_5=\Tr{\bsX\bsZ\bsY},a_6=\Tr{\bsX\bsY\bsZ}.
\end{eqnarray*}

\section{The generalized Schur-Weyl duality}

Consider the local unitary group
$\U(\bsd)=\U(d_1)\ot\cdots\ot\U(d_n)$, where $\bsd=(d_1,\ldots,d_n)$
are positive integer dimensions, which is a subgroup of
$\GL(\bsd)=\GL(d_1,\complex)\ot\cdots\ot\GL(d_n,\complex)$. Let
$V_i$ be a $d_i$-dimensional complex Hilbert space and
$V=V_1\ot\cdots\ot V_n$. Then $\U(\bsd)$ acts on the vector space
$\End(V)=\ot^n_{i=1}\End(V_i)$ by
\begin{eqnarray}
\bsM\longmapsto(\bsU_1\ot\cdots\ot \bsU_n)\bsM(\bsU_1\ot\cdots\ot
\bsU_n)^\dagger
\end{eqnarray}
which is obtained by linear extension of the action:
$\ot^n_{i=1}\bsX_i\mapsto \ot^n_{i=1}\bsU_i\bsX_i\bsU^\dagger_i$,
where $\bsX_i\in\End(V_i)$ and $\bsU_i\in\U(d_i)$. This in turn can
be extended to an action on $\End(V)^{\oplus m}$ by simultaneous
conjugation. That is, for a $m$-tuple $(\bsM_1,\ldots,\bsM_m)\in
\End(V)^{\oplus m}$, where $\bsM_i\in\End(V)$, we get that
\begin{eqnarray*}
(\bsM_1,\ldots,\bsM_m)\mapsto ((\bsU_1\ot\cdots\ot
\bsU_n)\bsM_1(\bsU_1\ot\cdots\ot
\bsU_n)^\dagger,\ldots,(\bsU_1\ot\cdots\ot
\bsU_n)\bsM_m(\bsU_1\ot\cdots\ot \bsU_n)^\dagger).
\end{eqnarray*}

\begin{prop}[\cite{Turner2017}]
Let $G$ be a group acts on a vector space $V$ rationally. If $H$ is
a Zariski dense subgroup of $G$, then it holds that
$$
\complex[V]^G=\complex[V]^H.
$$
Here all elements in $\complex[V]$ are interpreted as the polynomial
ring $\complex[\bsv_1,\ldots,\bsv_N]$, where $\bsv_1,\ldots,\bsv_N$
form a basis for $N$-dimensional complex vector space $V$;
$\complex[V]^G$ means the set of all $G$-invariant polynomials in
$\complex[V]$, similarly, $\complex[V]^H$ means the set of all
$H$-invariant polynomials in $\complex[V]$.
\end{prop}
It is well-known that $\U(d_i)$ is a Zariski dense subgroup of
$\GL(d_i,\complex)$. From this, we get that $\U(\bsd)$ is Zariski
dense in $\GL(\bsd,\complex)$, so
\begin{eqnarray*}
\complex[\End(V)^{\oplus m}]^{\U(\bsd)}=\complex[\End(V)^{\oplus
m}]^{\GL(\bsd,\complex)}.
\end{eqnarray*}
Moreover we have
\begin{eqnarray*}
\complex[\End(V)^{\oplus
m}]^{\GL(\bsd,\complex)}=\complex[\End(V)^{\oplus
m}]^{\SL(\bsd,\complex)}=\complex[\End(V)^{\oplus m}]^{\SU(\bsd)}.
\end{eqnarray*}
\begin{prop}[\cite{Turner2017}]
Two Hermitian matrices are in the same $\GL(\bsd,\complex)$-orbit
\blue{if and only if} they are in the same $\mathsf{U}(\bsd)$-orbit.
\end{prop}
Consider the representation of $\GL(\bsd,\complex)$ on $\End(V^{\ot
m})$, defined by
\begin{eqnarray}
\widetilde\bQ(g_1,\ldots,g_n)\defeq \bQ(g_1)\ot\cdots\bQ(g_n),
\end{eqnarray}
where $\bQ(g_i)=g^{\ot m}_i$ for $g_i\in \GL(d_i,\complex)$. Clearly
$\Pi(g_1,\ldots,g_n)\cong (g_1\ot\cdots\ot g_n)^{\ot m}$. Denote the
$n$-fold Cartesian product $S^n_m:=S_m\times\cdots\times S_m$ of the
symmetric group $S_m$ of order $m$. The action of $S^n_m$ on
$\End(V^{\ot m})$ is defined by
\begin{eqnarray}
\widetilde\bP(\pi_1,\ldots,\pi_n)\defeq
\bP(\pi_1)\ot\cdots\ot\bP(\pi_n),
\end{eqnarray}
where $\bP(\pi_i)\in \End(V^{\ot m}_i)$ for $\pi_i\in  S_m$ with its
definition taken from Eq.~\eqref{eq:P}.
\begin{thrm}[The generalized Schur-Weyl duality, \cite{Grassl1998,Turner2017}]
Let
\begin{eqnarray}
\widetilde\cA&=&\spn_{\complex}\Set{\widetilde\bP(\pi_1,\ldots,\pi_n):
(\pi_1,\ldots,\pi_m)\in S^n_m},\\
\widetilde\cB&=&\spn_{\complex}\Set{\widetilde\bQ(g_1,\ldots,g_n):
(g_1,\ldots,g_m)\in \GL(\bsd,\complex)}.
\end{eqnarray}
Then it holds that
\begin{eqnarray}
\widetilde\cA'=\widetilde\cB\quad\text{and}\quad
\widetilde\cB'=\widetilde\cA.
\end{eqnarray}
\end{thrm}

\begin{proof}
Note that $\End(V^{\ot m})\cong \End(V)^{\ot m}$. Now for
$g_i\in\GL(d_i,\complex)$, $\widetilde\bQ(g_1,\ldots,g_n)=g^{\ot
m}_1\ot\cdots \ot g^{\ot m}_n$ commutes with the action of
$\widetilde\bP(\pi_1,\ldots,\pi_n)$ for $\pi_i\in S_m$, that is,
$\widetilde\bP\widetilde\bQ = \widetilde\bQ\widetilde\bP$. This
implies that
$$
\widetilde\cA\subseteq\widetilde\cB'\quad\text{and}\quad\widetilde\cB\subseteq\widetilde\cA'.
$$
Now if $\widetilde\cB=\widetilde\cA'$, then
$\widetilde\cA=\widetilde\cB'$ by \emph{dual theorem} (i.e.,
Proposition~\ref{double-commutants}). It suffices to show
$\widetilde\cA'=\widetilde\cB$ only. Because
$\widetilde\cB\subseteq\widetilde\cA'$, next we show that
$\widetilde\cA'\subseteq\widetilde\cB$. Take any element in
$\widetilde\cA'$, say a simple tensor
$\ot^n_{i=1}\ot^m_{j=1}\bsM_{ij}$, so let $(\pi_1,\ldots,\pi_n)\in
S^n_m$ and consider
\begin{eqnarray*}
\widetilde\bP(\pi_1,\ldots,\pi_n)\Pa{\ot^n_{i=1}\ot^m_{j=1}\bsM_{ij}}\Pa{\widetilde\bP^{-1}(\pi_1,\ldots,\pi_n)\Pa{\ot^n_{i=1}\ot^m_{j=1}\bsv_{ij}}},
\end{eqnarray*}
where $\bsv_{ij}\in V_i$ for $j=1,\ldots,d_i$. We find that
\begin{eqnarray*}
&&\widetilde\bP(\pi_1,\ldots,\pi_n)\Pa{\ot^n_{i=1}\ot^m_{j=1}\bsM_{ij}}\Pa{\widetilde\bP^{-1}(\pi_1,\ldots,\pi_n)\Pa{\ot^n_{i=1}\ot^m_{j=1}\bsv_{ij}}}\\
&&=\widetilde\bP(\pi_1,\ldots,\pi_n)\Pa{\ot^n_{i=1}\ot^m_{j=1}\bsM_{ij}\bsv_{i\pi_i(j)}}=\ot^n_{i=1}\ot^m_{j=1}\bsM_{i\pi^{-1}_i(j)}\bsv_{i\pi_i(\pi^{-1}_i(j))}\\
&&=\ot^n_{i=1}\ot^m_{j=1}\bsM_{i\pi^{-1}_i(j)}\bsv_{ij} =
\Pa{\ot^n_{i=1}\ot^m_{j=1}\bsM_{i\pi^{-1}_i(j)}}\Pa{\ot^n_{i=1}\ot^m_{j=1}\bsv_{ij}},
\end{eqnarray*}
implying that
\begin{eqnarray*}
\widetilde\bP(\pi_1,\ldots,\pi_n)\Pa{\ot^n_{i=1}\ot^m_{j=1}\bsM_{ij}}\widetilde\bP^{-1}(\pi_1,\ldots,\pi_n)
=\ot^n_{i=1}\ot^m_{j=1}\bsM_{i\pi^{-1}_i(j)}.
\end{eqnarray*}
Thus the isomorphism $\End(V^{\ot m})\cong \End(V)^{\ot m}$ induces
the isomorphism $\widetilde\cA'(\subset\End(V^{\ot m}))$ to the
subalgebra $\Theta_{\bsd}\subset \End(V)^{\ot m}$ which is
$S^n_m$-invariant under the induced action. Now $\Theta_{\bsd}$ as a
$S^n_m$-module, we consider its decomposition. Because $S^n_m$ acts
trivially on $\Theta_{\bsd}$, it follows that every non-zero
irreducible submodule will be one dimensional. Recall a well-known
fact in Representation Theory that \blue{every irreducible
representation of $G_1\times\cdots\times G_n$, where $G_i$'s are
given groups, is the tensor product of $n$ irreducible
$G_i$-modules}. Hence we see that an irreducible $S^n_m$-submodule
of $\Theta_{\bsd}$ is the tensor product of $n$ irreducible
$S_m$-submodule of $\End(V_i)^{\ot m}$, which\footnote{The $S^n_m$
submodule spanned in this way is just one dimensional.} is spanned
by operators of the form $\bsM_1\ot\cdots\ot \bsM_n$, where each
$\bsM_i$ is a symmetric operator in $\End(V_i)^{\ot m}$.

So we get that
$$
\Theta_{\bsd}\cong\cB_1\ot\cdots\ot\cB_n,
$$
where $\cB_i$ are the symmetric operators in $\End(V_i)^{\ot m}$,
that is, $\cB_i=\spn_{\complex}\set{\bQ_m(g_i):g_i\in\GL(V_i)}$ by
Schur-Weyl duality. However, this algebra is clearly generated by
the form
$\bQ(g_1)\ot\cdots\ot\bQ(g_n)=\widetilde\bQ(g_1,\ldots,g_n)$, where
$g_i\in\GL(V_i)$. Therefore, we obtain that
$$
\widetilde\cA'=\widetilde\cB.
$$
We are done.
\end{proof}

\section{Some matrix integrals related to random matrix theory}

This section is written based on Taylor's Lectures on Lie groups
\cite{Tay}. In this section, we give a direct derivation of a
formula for
\begin{eqnarray*}
\int_{\U(d)} \abs{\Tr{\bsU^k}}^2\dif\mu(\bsU),
\end{eqnarray*}
of usage in random matrix theory. We also calculate a more refined
object,
\begin{eqnarray*}
\int_{\U(d)} \bsU^k\ot (\bsU^k)^\dagger \dif\mu(\bsU)=\int_{\U(d)}
\bsU^k\ot \bsU^{-k} \dif\mu(\bsU),
\end{eqnarray*}
which in turn yields a formula for
\begin{eqnarray*}
\int_{\U(d)} f(\bsU)\ot g(\bsU) \dif\mu(\bsU).
\end{eqnarray*}

Let $f:\mathbb{S}^1\to\complex$ be a bounded Borel function, where
$\mathbb{S}^1=\Set{e^{\sqrt{-1}\theta}:
\theta\in(0,2\pi]}=\set{z\in\complex: \abs{z}=1}$. Given $\bsU\in
\U(d)$, we define $f(\bsU)\in \End(\complex^d)$ by the spectral
decomposition: If
$\bsU=\sum^d_{j=1}e^{\sqrt{-1}\theta_j}\out{\bsu_j}{\bsu_j}$ with
$\set{\ket{\bsu_j}:j=1,\ldots,d}$ being an orthonormal basis for
$\complex^d$, then $f(\bsU)$ is defined as
\begin{eqnarray*}
f(\bsU) :=
\sum^d_{j=1}f\Pa{e^{\sqrt{-1}\theta_j}}\out{\bsu_j}{\bsu_j}.
\end{eqnarray*}
For instance, $\bsU^k=
\sum^d_{j=1}e^{\sqrt{-1}k\theta_j}\out{\bsu_j}{\bsu_j}$.

We are interested in formulae for
\begin{eqnarray*}
\int_{\U(d)} \Tr{f(\bsU)}\Tr{g(\bsU)} \dif\mu(\bsU).
\end{eqnarray*}
Note that the above is equal to the trace of
\begin{eqnarray}\label{eq:comput-of-integral}
\int_{\U(d)} f(\bsU)\ot g(\bsU) \dif\mu(\bsU).
\end{eqnarray}
The notion of Fourier series will be used here. On $\mathbb{S}^1$,
let $\dif\mu(z)$ is the uniform and normalized Haar measure. Thus
for $z=e^{\sqrt{-1}\theta}\in\mathbb{S}^1$,
$$
\dif\mu(z) = \frac{d\theta}{2\pi}.
$$
The functions $\chi_m(z)=z^m$ or $\theta\mapsto
e^{\sqrt{-1}m\theta}$ for $m\in\integer$ are just irreducible
characters of $\unitary{1}=\mathbb{S}^1$, and thus form an
orthonormal basis of complex $L^2(\mathbb{S}^1,\mu)$. For $f\in
L^1(\mathbb{S}^1,\mu)$, defining the coefficients
\begin{eqnarray*}
\widehat f(k):=\int_{\mathbb{S}^1}f(z)z^{-k}\dif\mu(z)
\end{eqnarray*}
yields the formal series
\begin{eqnarray}\label{eq:formal-series}
f \sim \sum_{k\in\integer} \widehat f(k)z^k.
\end{eqnarray}
If $f\in L^2(\mathbb{S}^1,\mu)$, the series converges
unconditionally to $f$ in $L^2(\mathbb{S}^1,\mu)$. For general $f\in
L^1(\mathbb{S}^1,\mu)$ and $z\in\mathbb{S}^1$, let
\begin{eqnarray*}
S_mf(z):=\sum_{k=-m}^m \widehat f(k)z^k,~~~m=0,1,2,\ldots.
\end{eqnarray*}
There is a one-to-one correspondence between functions
$F:\real\to\complex$, periodic of period $2\pi$, and functions
$f:\mathbb{S}^1\to\complex$, given by
$F(\theta)=f(e^{\sqrt{-1}\theta}),\theta\in\real$. We will set
$\widehat F(k)=\widehat f(k)$. For $F\in L^1([0,2\pi];\complex)$
(with respect to Lebesgue measure) the formal series
\eqref{eq:formal-series} corresponds to
\begin{eqnarray*}
F(\theta) \sim \sum_{k\in\integer} \widehat
F(k)e^{\sqrt{-1}k\theta},
\end{eqnarray*}
which has been called the \emph{exponential Fourier series} of $F$.
In terms of trigonometric functions we get another series
\begin{eqnarray*}
F(\theta)\sim c_0 +
\sum^\infty_{k=1}a_k\cos(k\theta)+b_k\sin(k\theta),
\end{eqnarray*}
called the \emph{Fourier series} of $F$ (or of $f$). Here
$a_k:=\widehat f(k)+\widehat f(-k)$ and $b_k=\sqrt{-1}(\widehat
f(k)-\widehat f(-k))$.

Specifically, the Fourier series of $F$ converges to $F$ at a given
$\theta$ if and only if
$\lim_{m\to\infty}S_mf(e^{\sqrt{-1}\theta})=f(e^{\sqrt{-1}\theta})$.

Now we find that
\begin{eqnarray*}
f(\bsU) = \sum^{+\infty}_{k=-\infty}\widehat f(k)\bsU^k,
\end{eqnarray*}
where
$$
\widehat f(k):= \frac1{2\pi}\int^{2\pi}_0
F(\theta)e^{-\sqrt{-1}k\theta}d\theta = \frac1{2\pi}\int^{2\pi}_0
f(e^{\sqrt{-1}\theta})e^{-\sqrt{-1}k\theta}d\theta.
$$
and
$$
F(\theta) = \sum_{k\in\integer}\widehat
F(k)e^{\sqrt{-1}k\theta}\Longleftrightarrow
f(e^{\sqrt{-1}\theta})=\sum_{k\in\integer}\widehat
f(k)e^{\sqrt{-1}k\theta}.
$$
Thus we have
\begin{eqnarray*}
\int_{\U(d)} f(\bsU)\ot g(\bsU) \dif\mu(\bsU) &=& \int_{\U(d)}
\Pa{\sum^{+\infty}_{i=-\infty}\widehat
f(i)\bsU^i}\ot\Pa{\sum^{+\infty}_{j=-\infty}\widehat g(j)\bsU^j}\dif\mu(\bsU)\\
&=&\sum_{i,j\in\integer} \widehat f(i)\widehat
g(j)\int_{\U(d)}\bsU^i\ot \bsU^j \dif\mu(\bsU)\\
&=&\sum_{i,j\in\integer} \widehat f(i)\widehat g(j) \bsM_{ij},
\end{eqnarray*}
where $\bsM_{ij} = \int_{\U(d)}\bsU^i\ot \bsU^j \dif\mu(\bsU)$.
Performing the measure-preserving transformation $\bsU\mapsto
e^{\sqrt{-1}\psi}\bsU$ on $\U(d)$, we see that
\begin{eqnarray*}
\bsM_{ij}= e^{\sqrt{-1}(i+j)\psi}\bsM_{ij}~~~\text{for
all}~\psi\in\real.
\end{eqnarray*}
Thus $\bsM_{ij}=0$ for $i\neq -j$. Hence
\begin{eqnarray*}
\int_{\U(d)} f(\bsU)\ot g(\bsU) \dif\mu(\bsU) = \sum_{k\in\integer}
\widehat f(k)\widehat g(-k) \bsM_k,
\end{eqnarray*}
where $\bsM_k:=\int_{\U(d)}\bsU^k\ot \bsU^{-k}\dif\mu(\bsU)$, which
implies that
\begin{eqnarray*}
\int_{\U(d)} \Tr{f(\bsU)}\Tr{g(\bsU)} \dif\mu(\bsU) =
\sum_{k\in\integer} \widehat f(k)\widehat g(-k) \Tr{\bsM_k}.
\end{eqnarray*}
It remains to compute the following integral
\begin{eqnarray*}
\Tr{\bsM_k}=\int_{\U(d)} \abs{\Tr{\bsU^k}}^2\dif\mu(\bsU).
\end{eqnarray*}
Here we establish the following identity:
\begin{prop}\label{th:u^k-int}
It holds that
\begin{eqnarray}
\int_{\mathsf{U}(d)} \abs{\Tr{\bsU^k}}^2\dif\mu(\bsU) = \min(k,d).
\end{eqnarray}
\end{prop}

\begin{proof}
Here we give a natural proof, based on Weyl's integration formula,
which implies that whenever $\varphi:\U(d)\to\complex$ invariant
under conjugation, then
\begin{eqnarray*}
\int_{\U(d)}\varphi(\bsU)\dif\mu(\bsU) =C_d
\int_{\mathbb{T}^d}\varphi(\bsD(\theta))J(\theta)\dif\bsD(\theta) =
\frac{C_d}{(2\pi)^d}\overbrace{\int^{2\pi}_0\cdots\int^{2\pi}_0}^{d}\varphi(\bsD(\theta))J(\theta)\dif\theta
\end{eqnarray*}
where
$\bsD(\theta)=\diag(e^{\sqrt{-1}\theta_1},\ldots,e^{\sqrt{-1}\theta_d})$,
and $J(\theta)=\prod_{i<j}\abs{e^{\sqrt{-1}\theta_i} -
e^{\sqrt{-1}\theta_j}}^2$. We will verify in calculations below that
$C_d=1/d!$.

Now
\begin{eqnarray*}
\int_{\U(d)} \abs{\Tr{\bsU^k}}^2\dif\mu(\bsU) = \frac{C_d}{(2\pi)^d}
\overbrace{\int^{2\pi}_0\cdots\int^{2\pi}_0}^{d}\abs{e^{\sqrt{-1}k\theta_1}+\cdots+e^{\sqrt{-1}k\theta_d}}^2
J(\theta)d\theta.
\end{eqnarray*}
We restate this as follows. Set $\zeta_j=e^{\sqrt{-1}\theta_j}$, so
\begin{eqnarray*}
\abs{e^{\sqrt{-1}k\theta_1}+\cdots+e^{\sqrt{-1}k\theta_d}}^2 =
\abs{\zeta^k_1+\cdots+\zeta^k_d}^2 = \sum_{p,q=1}^d
\zeta^k_p\zeta^{-k}_q
\end{eqnarray*}
and
\begin{eqnarray*}
J(\theta) &=& \prod_{i<j}\abs{\zeta_i-\zeta_j}^2 = \prod_{i<j}
(\zeta_i-\zeta_j)(\zeta^{-1}_i-\zeta^{-1}_j)\\
&=&
(\sign\tau)(\zeta_1\cdots\zeta_d)^{-(d-1)}\prod_{i<j}(\zeta_i-\zeta_j)^2,
\end{eqnarray*}
where $\tau=(d\cdots 21)$, i.e. $\tau(j)=d+1-j$ or $\tau$ is written
as
$$
\tau:=\left(
  \begin{array}{cccc}
    1 & 2 & \cdots & d \\
    d & d-1 & \cdots & 1 \\
  \end{array}
\right).
$$
Note that $\sign\tau=(-1)^{\frac{d(d-1)}2}$. We see that
$\int_{\U(d)} \abs{\Tr{\bsU^k}}^2\dif\mu(\bsU)$ is the constant term
in
\begin{eqnarray*}
C_d(\sign\tau)(\zeta_1\cdots\zeta_d)^{-(d-1)}\Pa{\sum_{p,q=1}^d
\zeta^k_p\zeta^{-k}_q}\prod_{i<j}(\zeta_i-\zeta_j)^2.
\end{eqnarray*}
Thus our task is to identify the constant term in this Laurent
polynomial. To work on the last factor, we recognize
\begin{eqnarray*}
V(\zeta) = \prod_{i<j} (\zeta_i-\zeta_j)
\end{eqnarray*}
as a Vandermonde determinant; hence
\begin{eqnarray*}
V(\zeta) = \sum_{\sigma\in S_d}
(\sign\sigma)\zeta^{\sigma(1)-1}_1\cdots\zeta^{\sigma(d)-1}_d.
\end{eqnarray*}
Hence
\begin{eqnarray*}
\prod_{i<j}(\zeta_i-\zeta_j)^2 = V(\zeta)^2 = \sum_{\sigma,\pi\in
S_d}
(\sign\sigma)(\sign\pi)\zeta^{\sigma(1)+\pi(1)-2}_1\cdots\zeta^{\sigma(d)+\pi(d)-2}_d.
\end{eqnarray*}
Let us first identify the constant term in
\begin{eqnarray*}
J(\theta) = (\sign\tau)(\zeta_1\cdots\zeta_d)^{-(d-1)}V(\zeta)^2.
\end{eqnarray*}
We see this constant term is equal to
\begin{eqnarray*}
&&\frac1{(2\pi)^d}
\overbrace{\int^{2\pi}_0\cdots\int^{2\pi}_0}^{d}J(\theta)\dif\theta \\
&&= (\sign\tau)\frac1{(2\pi)^d}
\overbrace{\int^{2\pi}_0\cdots\int^{2\pi}_0}^{d}\Pa{\sum_{\sigma,\pi\in
S_d}(\sign\sigma)(\sign\pi)\zeta^{\sigma(1)+\pi(1)-d-1}_1\cdots\zeta^{\sigma(d)+\pi(d)-d-1}_d}\dif\theta\\
&&=(\sign\tau) \sum_{\sigma,\pi\in
S_d}(\sign\sigma)(\sign\pi)\Pa{\frac1{2\pi}\int^{2\pi}_0
\zeta^{\sigma(1)+\pi(1)-d-1}_1\dif\theta_1}\times\cdots\times
\Pa{\frac1{2\pi}\int^{2\pi}_0
\zeta^{\sigma(d)+\pi(d)-d-1}_1\dif\theta_d}\\
&&=(\sign\tau) \sum_{(\sigma,\pi)\in S_d\times S_d:\forall
j,\sigma(j)+\pi(j)=d+1}(\sign\sigma)(\sign\pi)=(\sign\tau)
\sum_{(\sigma,\pi)\in S_d\times
S_d:\pi=\tau\sigma}(\sign\sigma)(\sign\pi).
\end{eqnarray*}
Note that the sum is over all $(\sigma,\pi)\in S_d\times S_d$ such
that $\sigma(j)+\pi(j)=d+1$ for each $j\in\set{1,\ldots,d}$. In
other words, we get $\pi(j)=d+1-\sigma(j)=\tau(\sigma(j))$ for all
$j\in\set{1,\ldots,d}$, i.e. $\pi=\tau\sigma$. Thus the sum is equal
to
$$
(\sign\tau)\sum_{\sigma\in S_d}(\sign\sigma)(\sign\tau\sigma) = d!,
$$
which gives rise to $C_d=1/d!$.

Clearly
\begin{eqnarray*}
&&C_d(\sign\tau)(\zeta_1\cdots\zeta_d)^{-(d-1)}\Pa{\sum_{p,q=1}^d
\zeta^k_p\zeta^{-k}_q}\prod_{i<j}(\zeta_i-\zeta_j)^2\\
&&= C_d(\sign\tau)(\zeta_1\cdots\zeta_d)^{-(d-1)}
\sum^d_{p,q=1}\sum_{(\sigma,\pi)\in S_d\times S_d}(\sign\sigma)(\sign\pi)\\
&&~~~\times \zeta^k_{p}\zeta^{-k}_{q}\zeta^{\sigma(1)+\pi(1)-2}_1\cdots\zeta^{\sigma(d)+\pi(d)-2}_d,\\
&&=C_d(\sign\tau)(V_1(\zeta)+V_2(\zeta)),
\end{eqnarray*}
where
\begin{eqnarray*}
V_1(\zeta) &:=& (\zeta_1\cdots\zeta_d)^{-(d-1)}
\sum^d_{p=1}\sum_{(\sigma,\pi)\in S_d\times S_d}(\sign\sigma)(\sign\pi)\zeta^{\sigma(1)+\pi(1)-2}_1\cdots\zeta^{\sigma(d)+\pi(d)-2}_d,\\
V_2(\zeta) &:=& (\zeta_1\cdots\zeta_d)^{-(d-1)} \sum_{p\neq
q}\sum_{(\sigma,\pi)\in S_d\times
S_d}(\sign\sigma)(\sign\pi)\zeta^k_{p}\zeta^{-k}_{q}\zeta^{\sigma(1)+\pi(1)-2}_1\cdots\zeta^{\sigma(d)+\pi(d)-2}_d.
\end{eqnarray*}
Now
\begin{eqnarray*}
V_1(\zeta) &:=& d\times \sum_{(\sigma,\pi)\in S_d\times
S_d}(\sign\sigma)(\sign\pi)\zeta^{\sigma(1)+\pi(1)-d-1}_1\cdots\zeta^{\sigma(d)+\pi(d)-d-1}_d,
\end{eqnarray*}
implying
\begin{eqnarray*}
\frac1{(2\pi)^d} \overbrace{\int^{2\pi}_0\cdots\int^{2\pi}_0}^{d}
V_1(\zeta)\dif\theta :=d \cdot d!(\sign\tau).
\end{eqnarray*}
It remains to consider the integral involved in $V_2(\zeta)$. That
is,
\begin{eqnarray*}
\frac1{(2\pi)^d} \overbrace{\int^{2\pi}_0\cdots\int^{2\pi}_0}^{d}
V_2(\zeta)\dif\theta.
\end{eqnarray*}
We see that for a given $p\neq q$, a pair $(\sigma,\pi)\in S_d\times
S_d$ contributes to the constant term in $V_2(\zeta)$ if and only if
\begin{eqnarray*}
\sigma(j)+\pi(j) =
\begin{cases}
d+1,&\text{if}~j\in\set{1,\ldots,d}\backslash\set{p,q}\\
d+1-k,&\text{if}~j=p,\\
d+1+k,&\text{if}~j=q.
\end{cases}
\end{eqnarray*}
That is,
\begin{eqnarray*}
\pi(j) = \begin{cases}
d+1-\sigma(j),&\text{if}~j\in\set{1,\ldots,d}\backslash\set{p,q}\\
d+1-\sigma(j)-k,&\text{if}~j=p,\\
d+1-\sigma(j)+k,&\text{if}~j=q.
\end{cases}
\end{eqnarray*}
By the definition of $\tau$, $d+1-\sigma(j) = \tau(\sigma(j))$ for
all $j$. Thus
\begin{eqnarray*}
\pi(j) = \begin{cases}
\tau(\sigma(j)),&\text{if}~j\in\set{1,\ldots,d}\backslash\set{p,q}\\
\tau(\sigma(j))-k,&\text{if}~j=p,\\
\tau(\sigma(j))+k,&\text{if}~j=q.
\end{cases}
\end{eqnarray*}
Define
\begin{eqnarray*}
\omega_{pq}(j) = \begin{cases}
j,&\text{if}~j\in\set{1,\ldots,d}\backslash\set{j_p,j_q}\\
j-k,&\text{if}~j=j_p,\\
j+k,&\text{if}~j=j_q,
\end{cases}
\end{eqnarray*}
where $j_p=\tau(\sigma(p))$ and $j_q=\tau(\sigma(q))$. Therefore
$\pi=\omega_{pq}\tau\sigma$, where $\omega_{pq}=(j_pj_q)$ with
$j_p-j_q=k$. Note that all the possible choices of $\omega_{pq}$
depends on all the possible values of positive integer $j_p$, i.e.
totally $d-k$ since $k+1\leqslant j_p\leqslant d$.

Now
\begin{eqnarray*}
V_2(\zeta) := \sum_{p\neq q}\sum_{(\sigma,\pi)\in S_d\times
S_d}(\sign\sigma)(\sign\pi)\zeta^k_{p}\zeta^{-k}_{q}\zeta^{\sigma(1)+\pi(1)-d-1}_1\cdots\zeta^{\sigma(d)+\pi(d)-d-1}_d,
\end{eqnarray*}
implying that if $1\leqslant k\leqslant d-1$,
\begin{eqnarray*}
\frac1{(2\pi)^d} \overbrace{\int^{2\pi}_0\cdots\int^{2\pi}_0}^{d}
V_2(\zeta)\dif\theta &=& \sum_{p\neq q}\sum_{(\sigma,\pi)\in
S_d\times
S_d:\pi=\omega_{pq}\tau\sigma}(\sign\sigma)(\sign\pi)\\
&=& \sum_{p\neq q}\sum_{\sigma\in
S_d}(\sign\sigma)(\sign\omega_{pq}\tau\sigma)\\
&=& d!(\sign\tau)\sum_{p\neq q}\sign\omega_{pq} = -(d-k)\cdot
d!(\sign\tau),
\end{eqnarray*}
where we used the fact that $\sum_{p\neq q}\sign\omega_{pq} =
-(d-k)$. The reason is for some pairs $(p,q)$ with $p\neq q$,
$\omega_{pq}$ does not exist, but just exist for $d-k$ pairs $(p,q)$
with $p\neq q$. We also note that if $k\geqslant d$, the choice of
$\omega_{pq}$ is empty. Therefore the integral involved in
$V_2(\zeta)$ is zero.
\end{proof}

\begin{cor}
For $k\geqslant 1,d\geqslant 2$, we have
\begin{eqnarray}
\int_{\mathsf{U}(d)}\bsU^k\ot \bsU^{-k} \dif\mu(\bsU) =
\frac{\min(k,d)-1}{d^2-1}\I_{d^2} +
\frac{d^2-\min(k,d)}{d(d^2-1)}\bsF,
\end{eqnarray}
where $\bsF=\sum^d_{i,j=1}\out{ij}{ji}$ is the swap operator on
$\complex^d\ot\complex^d$.
\end{cor}

\begin{proof}
Apparently $[\bsM_k,\bsV\ot \bsV]=0$ for all $\bsV\in \U(d)$. It
follows from Proposition~\ref{prop:uu-integral} that
\begin{eqnarray*}
\bsM_k &=& \int_{\U(d)} (\bsV\ot \bsV)\bsM_k(\bsV\ot \bsV)^{-1}\dif\mu(\bsV)\\
&=& \Pa{\frac{\Tr{\bsM_k}}{d^2-1} -
\frac{\Tr{\bsM_k\bsF}}{d(d^2-1)}}\I_{d^2} +
\Pa{\frac{\Tr{\bsM_k\bsF}}{d^2-1} -
\frac{\Tr{\bsM_k}}{d(d^2-1)}}\bsF.
\end{eqnarray*}
It suffices to compute $\tr{\bsM_k}$ and $\tr{\bsM_k\bsF}$. Clearly
$\tr{\bsM_k\bsF}=d$. By Proposition~\ref{th:u^k-int}, we have
$\tr{\bsM_k}=\min(k,d)$. Therefore the desired conclusion is
obtained.
\end{proof}

\begin{cor}
For $k\geqslant 1,d\geqslant 2$, we have
\begin{eqnarray}
\int_{\mathsf{U}(d)}\bsU^k \bsA(\bsU^k)^\dagger \dif\mu(\bsU) =
\frac{\min(k,d)-1}{d^2-1}\bsA +
\frac{d^2-\min(k,d)}{d(d^2-1)}\Tr{\bsA}\I_d,
\end{eqnarray}
where $\bsF=\sum^d_{i,j=1}\out{ij}{ji}$ is the swap operator on
$\complex^d\ot\complex^d$.
\end{cor}

\begin{proof}
Firstly we have
\begin{eqnarray*}
\int_{\U(d)}\bsU^k \ot\overline{\bsU}^k \dif\mu(\bsU) =
\frac{\min(k,d)-1}{d^2-1}\I_{d^2} +
\frac{d^2-\min(k,d)}{d(d^2-1)}\out{\vec(\I_d)}{\vec(\I_d)},
\end{eqnarray*}
which indicates that
\begin{eqnarray*}
&&\Pa{\int_{\U(d)}\bsU^k \ot\overline{\bsU}^k \dif\mu(\bsU)}\ket{\vec(\bsA)} \\
&&= \frac{\min(k,d)-1}{d^2-1}\I_{d^2}\ket{\vec(\bsA)} +
\frac{d^2-\min(k,d)}{d(d^2-1)}\ket{\vec(\I_d)}\Inner{\vec(\I_d)}{\vec(\bsA)}.
\end{eqnarray*}
Therefore
\begin{eqnarray*}
\int_{\U(d)}\bsU^k \bsA(\bsU^k)^\dagger \dif\mu(\bsU) =
\frac{\min(k,d)-1}{d^2-1}\bsA +
\frac{d^2-\min(k,d)}{d(d^2-1)}\Tr{\bsA}\I_d,
\end{eqnarray*}
implying the desired result. When $k=1$, the result of the present
proposition is reduced to Proposition~\ref{prop:u-integral}.
\end{proof}

Up to now, we can finish the computation of
\eqref{eq:comput-of-integral}. We obtain
\begin{prop}
It holds that
\begin{eqnarray}
\int_{\mathsf{U}(d)} f(\bsU)\ot g(\bsU) \dif\mu(\bsU) =
\frac{h(0)-\cF_dh(0)}{d^2-1}\Pa{\I_{d^2}-\frac1d\bsF} -
\frac{h(0)-\widehat h(0)}{d^2-1}(\I_{d^2}-d\bsF) + \widehat
h(0)\I_{d^2},
\end{eqnarray}
where
\begin{eqnarray}
h(\theta) = \frac1{2\pi}\int^{2\pi}_0f(t)g(t-\theta)\dif t
\end{eqnarray}
and $\cF_dh$ denotes the $d$-th Fej\'{e}r mean of the Fourier series
of $h$:
\begin{eqnarray}
\cF_dh(\theta) = \sum^d_{j=-d}\Pa{1-\frac{\dim(\abs{j},d)}d}\widehat
h(j)e^{\sqrt{-1}j\theta}.
\end{eqnarray}
\end{prop}

\begin{prop}
It holds that
\begin{eqnarray}
\int_{\mathsf{U}(d)} \abs{\tr{\bsU^k}}^4 \dif\mu(\bsU) =
\begin{cases}
2k^2, &\text{if}~1\leqslant 2k\leqslant d-1; \\
2k^2+2k-d,&\text{if}~d\leqslant 2k\leqslant 2(d-1),\\
d(2d-1), &\text{if}~k\geqslant d.
\end{cases}
\end{eqnarray}
\end{prop}

\begin{proof}
With notation in Proposition~\ref{th:u^k-int}, we have
\begin{eqnarray*}
\int_{\U(d)} \abs{\tr{\bsU^k}}^4 \dif\mu(\bsU) = \frac1{(2\pi)^dd!}
\overbrace{\int^{2\pi}_0\cdots\int^{2\pi}_0}^{d}\abs{e^{\sqrt{-1}k\theta_1}+\cdots+e^{\sqrt{-1}k\theta_d}}^4
J(\theta)\dif\theta.
\end{eqnarray*}
Thus
\begin{eqnarray*}
\abs{e^{\sqrt{-1}k\theta_1}+\cdots+e^{\sqrt{-1}k\theta_d}}^4 =
\abs{\zeta^k_1+\cdots+\zeta^k_d}^4 = \sum_{p,q,r,s=1}^d
\zeta^k_p\zeta^{-k}_q\zeta^k_r\zeta^{-k}_s
\end{eqnarray*}
and
\begin{eqnarray*}
J(\theta) =(\sign\tau)\sum_{\sigma,\pi\in S_d}
(\sign\sigma)(\sign\pi)\zeta^{\sigma(1)+\pi(1)-d-1}_1\cdots\zeta^{\sigma(d)+\pi(d)-d-1}_d.
\end{eqnarray*}
We see that $\int_{\U(d)} \abs{\tr{\bsU^k}}^4\dif\mu(\bsU)$ is the
constant term in
\begin{eqnarray*}
&&\frac1{d!}(\sign\tau)\sum^d_{p,q,r,s=1}\sum_{(\sigma,\pi)\in
S_d\times S_d}(\sign\sigma)(\sign\pi)
\zeta^k_{p}\zeta^{-k}_{q}\zeta^k_{r}\zeta^{-k}_{s}\cdot\zeta^{\sigma(1)+\pi(1)-d-1}_1\cdots\zeta^{\sigma(d)+\pi(d)-d-1}_d\\
&&=\frac1{d!}(\sign\tau)(V_{11}(\zeta)+V_{12}(\zeta)+V_{21}(\zeta)+V_{22}(\zeta)),
\end{eqnarray*}
where
\begin{eqnarray*}
V_{11}(\zeta) &:=& \sum^d_{p=q=1}\sum^d_{r=s=1}
\sum_{(\sigma,\pi)\in S_d\times S_d}(\sign\sigma)(\sign\pi)
\zeta^k_{p}\zeta^{-k}_{q}\zeta^k_{r}\zeta^{-k}_{s}\cdot\zeta^{\sigma(1)+\pi(1)-d-1}_1\cdots\zeta^{\sigma(d)+\pi(d)-d-1}_d,\\
V_{12}(\zeta) &:=& \sum^d_{p=q=1}\sum_{r\neq s}
\sum_{(\sigma,\pi)\in S_d\times S_d}(\sign\sigma)(\sign\pi)
\zeta^k_{p}\zeta^{-k}_{q}\zeta^k_{r}\zeta^{-k}_{s}\cdot\zeta^{\sigma(1)+\pi(1)-d-1}_1\cdots\zeta^{\sigma(d)+\pi(d)-d-1}_d,\\
V_{21}(\zeta) &:=& \sum_{p\neq q}\sum^d_{r=s=1}\sum_{(\sigma,\pi)\in
S_d\times S_d}(\sign\sigma)(\sign\pi)
\zeta^k_{p}\zeta^{-k}_{q}\zeta^k_{r}\zeta^{-k}_{s}\cdot\zeta^{\sigma(1)+\pi(1)-d-1}_1\cdots\zeta^{\sigma(d)+\pi(d)-d-1}_d,\\
V_{22}(\zeta) &:=& \sum_{p\neq q}\sum_{r\neq s}\sum_{(\sigma,\pi)\in
S_d\times S_d}(\sign\sigma)(\sign\pi)
\zeta^k_{p}\zeta^{-k}_{q}\zeta^k_{r}\zeta^{-k}_{s}\cdot\zeta^{\sigma(1)+\pi(1)-d-1}_1\cdots\zeta^{\sigma(d)+\pi(d)-d-1}_d.
\end{eqnarray*}
Next our task is to identify the constant term in this Laurent
polynomial. Now
\begin{eqnarray*}
V_{11}(\zeta) &:=& d^2\times \sum_{(\sigma,\pi)\in S_d\times
S_d}(\sign\sigma)(\sign\pi)\zeta^{\sigma(1)+\pi(1)-d-1}_1\cdots\zeta^{\sigma(d)+\pi(d)-d-1}_d,
\end{eqnarray*}
implying
\begin{eqnarray*}
\frac1{(2\pi)^d} \overbrace{\int^{2\pi}_0\cdots\int^{2\pi}_0}^{d}
V_{11}(\zeta)\dif\theta :=d^2 \cdot d!(\sign\tau).
\end{eqnarray*}
Then we consider the integral involved in $V_{12}(\zeta)$.
Apparently,
\begin{eqnarray*}
V_{12}(\zeta) &:=& d\times \sum_{r\neq s}\sum_{(\sigma,\pi)\in
S_d\times
S_d}(\sign\sigma)(\sign\pi)\zeta^k_r\zeta^{-k}_s\zeta^{\sigma(1)+\pi(1)-d-1}_1\cdots\zeta^{\sigma(d)+\pi(d)-d-1}_d,
\end{eqnarray*}
implying that
\begin{eqnarray*}
\frac1{(2\pi)^d} \overbrace{\int^{2\pi}_0\cdots\int^{2\pi}_0}^{d}
V_{12}(\zeta)\dif\theta =
\begin{cases}
- d(d-k)d!(\sign\tau), &\text{if}~1\leqslant k\leqslant d-1; \\
0,&\text{if}~k\geqslant d.
\end{cases}
\end{eqnarray*}
Similarly,
\begin{eqnarray*}
V_{21}(\zeta) &:=& d\times \sum_{p\neq q}\sum_{(\sigma,\pi)\in
S_d\times
S_d}(\sign\sigma)(\sign\pi)\zeta^k_p\zeta^{-k}_q\zeta^{\sigma(1)+\pi(1)-d-1}_1\cdots\zeta^{\sigma(d)+\pi(d)-d-1}_d,
\end{eqnarray*}
it follows that
\begin{eqnarray*}
\frac1{(2\pi)^d} \overbrace{\int^{2\pi}_0\cdots\int^{2\pi}_0}^{d}
V_{21}(\zeta)\dif\theta =
\begin{cases}
- d(d-k)d!(\sign\tau), &\text{if}~1\leqslant k\leqslant d-1; \\
0,&\text{if}~k\geqslant d.
\end{cases}
\end{eqnarray*}
It remains to consider the integral involved in $V_{22}(\zeta)$. We
have that
\begin{eqnarray*}
V_{22}(\zeta) &:=& \sum_{p\neq q}\sum_{r\neq s}
\sum_{(\sigma,\pi)\in S_d\times
S_d}(\sign\sigma)(\sign\pi)\zeta^k_p\zeta^{-k}_q\zeta^k_r\zeta^{-k}_s\zeta^{\sigma(1)+\pi(1)-d-1}_1\cdots\zeta^{\sigma(d)+\pi(d)-d-1}_d.
\end{eqnarray*}
We still need to split $V_{22}(\zeta)$ into some parts. In order to
be convenience, we introduce the following notation: $\cI:=
\Set{(\mu,\nu): \mu,\nu\in[d]~\text{and}~\mu\neq \nu}$, where
$[d]:=\set{1,2,\ldots,d}$. We also denote
\begin{eqnarray*}
\Lambda_1&:=&\Set{((p,q),(r,s)):(p,q),(r,s)\in\cI~\text{and}~(p,q)=(r,s)},\\
\Lambda_2&:=&\Set{((p,q),(r,s)):(p,q),(r,s)\in\cI~\text{and}~(p,q)=(s,r)},\\
\Lambda_3&:=&\Set{((p,q),(r,s)):(p,q),(r,s)\in\cI~\text{and}~(p,q)\neq(r,s)~\text{and}~(p,q)\neq(s,r)}.
\end{eqnarray*}
Thus we can get a partition of
$\cI\times\cI=\Lambda_1\cup\Lambda_2\cup\Lambda_3$
\begin{eqnarray*}
V_{22}(\zeta) = V^{(1)}_{22}(\zeta) + V^{(2)}_{22}(\zeta) +
V^{(3)}_{22}(\zeta),
\end{eqnarray*}
where
\begin{eqnarray*}
V^{(1)}_{22}(\zeta) &:=&
\sum_{((p,q),(r,s))\in\Lambda_1}\sum_{(\sigma,\pi)\in S_d\times
S_d}(\sign\sigma)(\sign\pi)\zeta^k_p\zeta^{-k}_q\zeta^k_r\zeta^{-k}_s\zeta^{\sigma(1)+\pi(1)-d-1}_1\cdots\zeta^{\sigma(d)+\pi(d)-d-1}_d\\
&=&\sum_{(p,q)\in\cI}\sum_{(\sigma,\pi)\in S_d\times
S_d}(\sign\sigma)(\sign\pi)\zeta^{2k}_p\zeta^{-2k}_q\zeta^{\sigma(1)+\pi(1)-d-1}_1\cdots\zeta^{\sigma(d)+\pi(d)-d-1}_d,\\
V^{(2)}_{22}(\zeta) &:=&
\sum_{((p,q),(r,s))\in\Lambda_2}\sum_{(\sigma,\pi)\in S_d\times
S_d}(\sign\sigma)(\sign\pi)\zeta^k_p\zeta^{-k}_q\zeta^k_r\zeta^{-k}_s\zeta^{\sigma(1)+\pi(1)-d-1}_1\cdots\zeta^{\sigma(d)+\pi(d)-d-1}_d\\
&=&\sum_{(p,q)\in\cI}\sum_{(\sigma,\pi)\in S_d\times
S_d}(\sign\sigma)(\sign\pi)\zeta^{\sigma(1)+\pi(1)-d-1}_1\cdots\zeta^{\sigma(d)+\pi(d)-d-1}_d\\
&=&\binom{d}{2}2!\sum_{(\sigma,\pi)\in S_d\times
S_d}(\sign\sigma)(\sign\pi)\zeta^{\sigma(1)+\pi(1)-d-1}_1\cdots\zeta^{\sigma(d)+\pi(d)-d-1}_d,\\
V^{(3)}_{22}(\zeta) &:=&
\sum_{((p,q),(r,s))\in\Lambda_3}\sum_{(\sigma,\pi)\in S_d\times
S_d}(\sign\sigma)(\sign\pi)\zeta^k_p\zeta^{-k}_q\zeta^k_r\zeta^{-k}_s\zeta^{\sigma(1)+\pi(1)-d-1}_1\cdots\zeta^{\sigma(d)+\pi(d)-d-1}_d.
\end{eqnarray*}
This indicates that
\begin{eqnarray*}
\frac1{(2\pi)^d} \overbrace{\int^{2\pi}_0\cdots\int^{2\pi}_0}^{d}
V^{(1)}_{22}(\zeta)\dif\theta &=&
\begin{cases}
- (d-2k)d!(\sign\tau), &\text{if}~1\leqslant 2k\leqslant d-1; \\
0,&\text{if}~2k\geqslant d
\end{cases}\\
\frac1{(2\pi)^d} \overbrace{\int^{2\pi}_0\cdots\int^{2\pi}_0}^{d}
V^{(2)}_{22}(\zeta)\dif\theta &=&d(d-1)d!(\sign\tau).
\end{eqnarray*}
Moreover we separate the index set $\Lambda_3$ into some disjoint
unions: $\Lambda_3=
\Lambda^{(1)}_3\cup\Lambda^{(2)}_3\cup\Lambda^{(3)}_3\cup\Lambda^{(4)}_3\cup\Lambda^{(5)}_3$,
where
\begin{eqnarray*}
\Lambda^{(1)}_3 &:=& \Set{((p,q),(r,s))\in\Lambda_3: p=r},\\
\Lambda^{(2)}_3 &:=& \Set{((p,q),(r,s))\in\Lambda_3: p=s},\\
\Lambda^{(3)}_3 &:=& \Set{((p,q),(r,s))\in\Lambda_3: q= r},\\
\Lambda^{(4)}_3 &:=& \Set{((p,q),(r,s))\in\Lambda_3: q= s},\\
\Lambda^{(5)}_3 &:=& \Set{((p,q),(r,s))\in\Lambda_3: p\neq r,p\neq
s,q\neq r,q\neq s}.
\end{eqnarray*}
Thus $V^{(3)}_{22}(\zeta)$ is partitioned as five subparts:
\begin{eqnarray*}
V^{(3)}_{22}(\zeta)=
V^{(31)}_{22}(\zeta)+V^{(32)}_{22}(\zeta)+V^{(33)}_{22}(\zeta)+V^{(34)}_{22}(\zeta)+V^{(35)}_{22}(\zeta).
\end{eqnarray*}
We see that for a given $p\neq q$ and $r\neq s$, if $p=r$, then a
pair $(\sigma,\pi)\in S_d\times S_d$ contributes to the constant
term in $V^{(31)}_{22}(\zeta)$ if and only if
\begin{eqnarray*}
\sigma(j)+\pi(j) =
\begin{cases}
d+1,&\text{if}~j\in\set{1,\ldots,d}\backslash\set{p,q,s}\\
d+1-2k,&\text{if}~j=p,\\
d+1+k,&\text{if}~j=q,s.
\end{cases}
\end{eqnarray*}
That is,
\begin{eqnarray*}
\pi(j) = \begin{cases}
d+1-\sigma(j),&\text{if}~j\in\set{1,\ldots,d}\backslash\set{p,q,s}\\
d+1-\sigma(j)-2k,&\text{if}~j=p,\\
d+1-\sigma(j)+k,&\text{if}~j=q,s.
\end{cases}
\end{eqnarray*}
By the definition of $\tau$, $d+1-\sigma(j) = \tau(\sigma(j))$ for
all $j$. Thus
\begin{eqnarray*}
\pi(j) = \begin{cases}
\tau(\sigma(j)),&\text{if}~j\in\set{1,\ldots,d}\backslash\set{p,q,s}\\
\tau(\sigma(j))-2k,&\text{if}~j=p,\\
\tau(\sigma(j))+k,&\text{if}~j=q,s.
\end{cases}
\end{eqnarray*}
Define
\begin{eqnarray*}
\omega_{pqs}(j) = \begin{cases}
j,&\text{if}~j\in\set{1,\ldots,d}\backslash\set{j_p,j_q,j_s}\\
j-2k,&\text{if}~j=j_p,\\
j+k,&\text{if}~j=j_q,j_s,
\end{cases}
\end{eqnarray*}
where $j_p=\tau(\sigma(p)),j_q=\tau(\sigma(q))$ and
$j_s=\tau(\sigma(s))$. Therefore $\pi=\omega_{pqs}\tau\sigma$, where
$\omega_{pqs}=(j_pj_qj_s)$ or $(j_pj_sj_q)$. Note that all the
possible choices of $\omega_{pqs}$ depends on all the possible
values of positive integer $j_p$. If $\omega_{pqs}=(j_pj_qj_s)$,
then $j_p=j_q+2k$ and $j_s=j_q+k$, thus $1\leqslant j_q\leqslant
d-2k$. If $\omega_{pqs}=(j_pj_sj_q)$, then $j_p=j_s+2k$ and
$j_q=j_s+k$, thus $1\leqslant j_s\leqslant d-2k$. This implies that
\begin{eqnarray*}
\frac1{(2\pi)^d} \overbrace{\int^{2\pi}_0\cdots\int^{2\pi}_0}^{d}
V^{(31)}_{22}(\zeta)\dif\theta &=&\begin{cases}
(d-2k)d!(\sign\tau), &\text{if}~1\leqslant 2k\leqslant d-1; \\
0,&\text{if}~2k\geqslant d.
\end{cases}
\end{eqnarray*}
Similarly the above analysis goes for $V^{(34)}_{22}(\zeta)$. We see
that for a given $p\neq q$ and $r\neq s$, if $q=s$, then a pair
$(\sigma,\pi)\in S_d\times S_d$ contributes to the constant term in
$V^{(34)}_{22}(\zeta)$ if and only if
\begin{eqnarray*}
\sigma(j)+\pi(j) =
\begin{cases}
d+1,&\text{if}~j\in\set{1,\ldots,d}\backslash\set{p,q,r}\\
d+1+2k,&\text{if}~j=q,\\
d+1-k,&\text{if}~j=p,r.
\end{cases}
\end{eqnarray*}
That is,
\begin{eqnarray*}
\pi(j) = \begin{cases}
d+1-\sigma(j),&\text{if}~j\in\set{1,\ldots,d}\backslash\set{p,q,r}\\
d+1-\sigma(j)+2k,&\text{if}~j=q,\\
d+1-\sigma(j)-k,&\text{if}~j=p,r.
\end{cases}
\end{eqnarray*}
By the definition of $\tau$, $d+1-\sigma(j) = \tau(\sigma(j))$ for
all $j$. Thus
\begin{eqnarray*}
\pi(j) = \begin{cases}
\tau(\sigma(j)),&\text{if}~j\in\set{1,\ldots,d}\backslash\set{p,q,r}\\
\tau(\sigma(j))+2k,&\text{if}~j=q,\\
\tau(\sigma(j))-k,&\text{if}~j=p,r.
\end{cases}
\end{eqnarray*}
Define
\begin{eqnarray*}
\omega_{pqs}(j) = \begin{cases}
j,&\text{if}~j\in\set{1,\ldots,d}\backslash\set{j_p,j_q,j_r}\\
j+2k,&\text{if}~j=j_q,\\
j-k,&\text{if}~j=j_p,j_r,
\end{cases}
\end{eqnarray*}
where $j_p=\tau(\sigma(p)),j_q=\tau(\sigma(q))$ and
$j_r=\tau(\sigma(r))$. Therefore $\pi=\omega_{pqr}\tau\sigma$, where
$\omega_{pqr}=(j_pj_qj_r)$ or $(j_pj_rj_q)$. Note that all the
possible choices of $\omega_{pqr}$ depends on all the possible
values of positive integer $j_p$. If $\omega_{pqr}=(j_pj_qj_r)$,
then $j_p=j_q+2k$ and $j_r=j_q+k$, thus $1\leqslant j_q\leqslant
d-2k$. If $\omega_{pqr}=(j_pj_rj_q)$, then $j_p=j_r+2k$ and
$j_q=j_r+k$, thus $1\leqslant j_r\leqslant d-2k$. This implies that
\begin{eqnarray*}
\frac1{(2\pi)^d} \overbrace{\int^{2\pi}_0\cdots\int^{2\pi}_0}^{d}
V^{(34)}_{22}(\zeta)\dif\theta &=&\begin{cases}
(d-2k)d!(\sign\tau), &\text{if}~1\leqslant 2k\leqslant d-1; \\
0,&\text{if}~2k\geqslant d.
\end{cases}
\end{eqnarray*}
It is easily obtained that the formulae for $V^{(32)}_{22}(\zeta)$
and $V^{(33)}_{22}(\zeta)$.
\begin{eqnarray*}
\frac1{(2\pi)^d} \overbrace{\int^{2\pi}_0\cdots\int^{2\pi}_0}^{d}
V^{(32)}_{22}(\zeta)\dif\theta &=&\begin{cases}
-d(d-1-k)d!(\sign\tau), &\text{if}~1\leqslant k\leqslant d-1; \\
0,&\text{if}~k\geqslant d
\end{cases}
\end{eqnarray*}
and
\begin{eqnarray*}
\frac1{(2\pi)^d} \overbrace{\int^{2\pi}_0\cdots\int^{2\pi}_0}^{d}
V^{(33)}_{22}(\zeta)\dif\theta &=&\begin{cases}
-d(d-1-k)d!(\sign\tau), &\text{if}~1\leqslant k\leqslant d-1; \\
0,&\text{if}~k\geqslant d.
\end{cases}
\end{eqnarray*}
It remains to compute the integral involved in
$V^{(35)}_{22}(\zeta)$. We see that for a given $p\neq q$ and $r\neq
s$, if $p\neq r,p\neq s,q\neq r,q\neq s$, then a pair
$(\sigma,\pi)\in S_d\times S_d$ contributes to the constant term in
$V^{(35)}_{22}(\zeta)$ if and only if
\begin{eqnarray*}
\sigma(j)+\pi(j) =
\begin{cases}
d+1,&\text{if}~j\in\set{1,\ldots,d}\backslash\set{p,q,r,s}\\
d+1-k,&\text{if}~j=p,r\\
d+1+k,&\text{if}~j=q,s.
\end{cases}
\end{eqnarray*}
That is,
\begin{eqnarray*}
\pi(j) = \begin{cases}
d+1-\sigma(j),&\text{if}~j\in\set{1,\ldots,d}\backslash\set{p,q,r,s}\\
d+1-\sigma(j)-k,&\text{if}~j=p,r,\\
d+1-\sigma(j)+k,&\text{if}~j=q,s.
\end{cases}
\end{eqnarray*}
By the definition of $\tau$, $d+1-\sigma(j) = \tau(\sigma(j))$ for
all $j$. Thus
\begin{eqnarray*}
\pi(j) = \begin{cases}
\tau(\sigma(j)),&\text{if}~j\in\set{1,\ldots,d}\backslash\set{p,q,r,s}\\
\tau(\sigma(j))-k,&\text{if}~j=p,r,\\
\tau(\sigma(j))+k,&\text{if}~j=q,s.
\end{cases}
\end{eqnarray*}
Define
\begin{eqnarray*}
\omega_{pqrs}(j) = \begin{cases}
j,&\text{if}~j\in\set{1,\ldots,d}\backslash\set{j_p,j_q,j_r,j_s}\\
j-k,&\text{if}~j=j_p,j_r,\\
j+k,&\text{if}~j=j_q,j_s,
\end{cases}
\end{eqnarray*}
where $j_p=\tau(\sigma(p)),j_q=\tau(\sigma(q))$ and
$j_r=\tau(\sigma(r)),j_s=\tau(\sigma(s))$. Therefore
$\pi=\omega_{pqrs}\tau\sigma$, where
$\omega_{pqrs}=(j_pj_q)(j_rj_s)$ or $(j_pj_s)(j_qj_r)$. Therefore
\begin{eqnarray*}
\frac1{(2\pi)^d} \overbrace{\int^{2\pi}_0\cdots\int^{2\pi}_0}^{d}
V^{(35)}_{22}(\zeta)\dif\theta &=&\begin{cases}
2(d-k)(d-k-1)d!(\sign\tau), &\text{if}~1\leqslant k\leqslant d-1; \\
0,&\text{if}~k\geqslant d.
\end{cases}
\end{eqnarray*}
Finally we get that
\begin{eqnarray*}
\int_{\U(d)} \abs{\tr{\bsU^k}}^4 \dif\mu(\bsU) =
\begin{cases}
2k^2, &\text{if}~1\leqslant 2k\leqslant d-1; \\
2k^2+2k-d,&\text{if}~d\leqslant 2k\leqslant 2(d-1),\\
d(2d-1), &\text{if}~k\geqslant d.
\end{cases}
\end{eqnarray*}
We are done.
\end{proof}
In fact, when $k>d$, $\int_{\U(d)} \abs{\tr{\bsU^k}}^4 \dif\mu(\bsU)
= d(2d-1)$ can be seen again in \cite{Pastur2004}. Apparently, we
get more in this proposition.

\begin{remark}
Based on the above discussion, we can consider the following
computations:
\begin{enumerate}[(i)]
\item $\int_{\U(d)} (\bsU^k)^{\ot n}\ot (\bsU^{-k})^{\ot n}\dif\mu(\bsU)$;
\item $\int_{\U(d)} (\bsU^k)^{\ot n}\bsA(\bsU^{-k})^{\ot n}\dif\mu(\bsU)$;
\item $\int_{\U(d)} \abs{\Tr{\bsU^k}}^{2n}\dif\mu(\bsU)$.
\end{enumerate}
Indeed, for (iii), we see from the results in \cite{Diaconis2001}
that if the integer $k$ satisfies the condition $1\leqslant kn
\leqslant d$, then
\begin{eqnarray}
\int_{\U(d)} \abs{\Tr{\bsU^k}}^{2n}\dif\mu(\bsU) = k^n\cdot n!.
\end{eqnarray}
What happened for $kn>d$? We leave them open for future research.
\end{remark}

\section{Discussion and concluding remarks}\label{sect:concluding-remarks}

We see that the integrals considered in this paper, where all the
underlying domain of integrals are just $\U(d)$. As a matter of
fact, analogous problems can be considered when the unitary group
$\U(d)$ can be replaced by a compact Lie group $G$ of some
particular property, for instance, we may assume that $G$ is a
\emph{gauge group} (see \cite{Mashhad,Spekkens}), a some kind of
subgroup of $\U(d)$.

In addition, we can derive some similar results from Schur Orthogonality Relations. Recall that for a compact Lie group $G$, let $\set{g\to V^{(\mu)}(g)}$ be the set of all inequivalent unitary irreps on the underlying vector space $\cV$. Consider the matrix entries of all these unitary matrices as a set of functions from $G$ to $\complex$, denoted by $\set{V^{(\mu)}_{i,j}}$. Then, they satisfy the following Schur-Orthogonality Relations:
\begin{eqnarray}\label{eq:Schur-Ortho-Relation}
\int_G V^{(\mu)}_{i,j}(g) \overline{V}^{(\nu)}_{k,l}(g) \dif\mu(g) =
\frac1{d_\mu}\delta_{\mu\nu}\delta_{ik}\delta_{jl},
\end{eqnarray}
where $\dif\mu(g)$ is the uniform probability Haar measure on $G$,
bar means the complex conjugate and $d_\mu$ is the dimension of
irrep $\mu$. We can make analysis about
\eqref{eq:Schur-Ortho-Relation} as follows:  For the orthonormal
base $\set{\ket{i}:i=1,\ldots,d_\mu}$ and $\set{\ket{k}:
k=1,\ldots,d_\nu}$, we have
\begin{eqnarray}
V^{(\mu)}_{i,j}(g) = \Innerm{i}{V^{(\mu)}(g)}{j},~~~\overline{V}^{(\nu)}_{k,l}(g) = \Innerm{k}{\overline{V}^{(\nu)}(g)}{l}.
\end{eqnarray}
Then
\begin{eqnarray*}
\int_G V^{(\mu)}(g)\ot \overline{V}^{(\nu)}(g)\dif\mu(g) =
\frac1{d_\mu}\delta_{\mu\nu}\sum_{i,j=1}^{d_\mu}
\sum_{k,l=1}^{d_\nu}\delta_{ik}\delta_{jl}\out{ik}{jl}.
\end{eqnarray*}
That is
\begin{eqnarray}
\int_G V^{(\mu)}(g)\ot \overline{V}^{(\nu)}(g)\dif\mu(g) =
\begin{cases}
0,&~\text{if}~\mu\neq \nu, \\
\frac1{d_\mu}\out{\vec(\I_\mu)}{\vec(\I_\mu)}, &~\text{if}~\mu= \nu.
\end{cases}
\end{eqnarray}
Here $\vec(\I_\mu):=\sum_{i,j=1}^{d_\mu}\ket{ii}$. This indicates that
\begin{eqnarray}
\int_G V^{(\mu)}(g)\bsX V^{(\mu),\dagger}(g)\dif\mu(g) =
\frac1{d_\mu}\Tr{\bsX}\I_\mu
\end{eqnarray}
is a completely depolarizing channel. Therefore for $\mu\neq\nu$,
$\int_G V^{(\mu)}(g)\ot V^{(\nu),\dagger}(g)\dif\mu(g) = 0$, and
\begin{eqnarray}
\int_G V^{(\mu)}(g)\ot V^{(\mu),\dagger}(g)\dif\mu(g) =
\frac1{d_\mu} \bsF^{(\mu)},
\end{eqnarray}
where $\bsF^{(\mu)}$ is the swap operator on the 2-fold tensor space
of irrep $\mu$. In view of this point, we naturally want to know if
the integral
\begin{eqnarray}\label{eq:g-g*}
\int_G V(g) \ot V^\dagger(g)\dif\mu(g)
\end{eqnarray}
can be computed explicitly, where $\set{g\to V(g)}$ is any unitary
representation of $G$. In particular, when $G = \U(d)$ and $V(g) =
\bQ(g)$, the integral \eqref{eq:g-g*} is reduced to the form:
\begin{eqnarray}
\int_{\U(d)} \bQ(g) \ot \bQ^\dagger(g)\dif\mu(g),
\end{eqnarray}
for which we have derived explicit formula in the present paper. We
leave these topics for future research.


\subsubsection*{Acknowledgement}
The author would also like to thank Haijiang Yu for his useful
conversations, and thank Nan Li for bringing
Corollary~\ref{cor:uu*u*u} to my attention.


\newpage
\appendix
\appendixpage
\addappheadtotoc

\section{Appendix}

To better understand Schur-Weyl duality, i.e. irreps of unitary
group and permutation group, we collect some relevant materials. The
details presented in the Appendix are written based on Notes of
Audenaert \cite{kmra}.

\subsection{Partitions}
A partition is a sequence
$\lambda=(\lambda_1,\lambda_2,\ldots,\lambda_r,\ldots)$ of
non-negative integers in non-increasing order
$$
\lambda_1\geqslant\lambda_2\geqslant\cdots\geqslant
\lambda_r\geqslant\cdots
$$
and containing finitely many non-zero terms. The non-vanishing terms
$\lambda_j$ are called the \emph{parts} of $\lambda$. The
\emph{length} (also called \blue{height}) of $\lambda$, denoted
$\ell(\lambda)$, is the number of parts of $\lambda$. The
\emph{weight} of $\lambda$, denoted $\abs{\lambda}$, is the sum of
the parts: $\abs{\lambda} := \sum_j \lambda_j$. A partition
$\lambda$ with weight $\abs{\lambda}=k$ is also called a
\blue{partition} of $k$, and this is denoted $\lambda\vdash k$. We
will also use the notation $\lambda\vdash_d k$ to indicate that
$\lambda\vdash k$ and $\ell(\lambda)\leqslant d$ in one statement.

For $\lambda\vdash k$, we use the shorthand $\bar \lambda:=
\frac\lambda k$. For $j\geqslant1$, the $j$-th element of $\lambda$
is denoted by $\lambda_j$. This element is a part if $j\leqslant
\ell(\lambda)$, otherwise it is 0. It is frequently convenient to
use a different notation that indicates the number of times each
integer $j=1,2,\ldots,\abs{\lambda}$ occurs as a part, the so-called
\emph{multiplicity} $m_j$ of $j$:
$$
\lambda = (1^{m_1}2^{m_2}\ldots r^{m_r}\ldots).
$$
As a shorthand we will use a superscripted index: $\lambda_j =
m_j(\lambda)$.

Now one has the relations
\begin{eqnarray*}
\begin{cases}
\sum^k_{j=1}\lambda^j &= \ell(\lambda),\\
\sum^k_{j=1}j\lambda^j &= \abs{\lambda} = k.
\end{cases}
\end{eqnarray*}

When dealing with numerical calculations it is necessary to impose
an ordering on the set of partitions. We will adhere here to the
lexicographic ordering, in which $\lambda$ precedes $\mu$, denoted
$\lambda>\mu$, if and only if the first non-zero difference
$\lambda_j-\mu_j$ is positive.

\begin{exam}
With the above convention, the partitions of $5$ are ordered as
follows:
$$
(5),~(41),~(32),~(31^2),~(2^21),~(21^3),~(1^5).
$$
It is seen easily that lexicographic ordering is a total order.
\end{exam}

\subsection{Young frames and Young tableaux}

\blue{Partitions can be graphically represented by Young diagrams},
which are Young tableaux with empty boxes. The $j$-th part
$\lambda_j$ corresponds to the $j$-th row of the diagram, consisting
of $\lambda_j$ boxes. Conversely, the Young diagrams of $k$ boxes
can be uniquely labeled by a partition $\lambda\vdash k$. \blue{We
will therefore identify a Young diagram with the partition labeling
it}.

A \emph{Young tableau} (YT) of $d$ objects and of shape
$\lambda\vdash k$ is a Young diagram $\lambda$ in which the boxes
are labeled by numbers $\set{1,\ldots,d}$.

A \emph{standard Young tableau} (SYT) of shape $\lambda\vdash k$ is
a Young tableau of $d=k$ objects such that the labels appear
\emph{increasing} in every row from left to right, and
\emph{increasing} in every column downwards; hence every number
occurs exactly once.

A \emph{Semi-standard Young tableau} (SSYT) of shape $\lambda\vdash
k$ is a Young tableau such that the labels appear
\emph{non-decreasing} in every row from left to right, and
\emph{increasing} in every column downwards.

The number of SSYTs of $d$ objects and of shape $\lambda\vdash k$
(imposing the condition $\ell(\lambda)\leqslant d$) is given by
$s_\lambda(1^{\times d}) \equiv s_{\lambda,d}(1)$; see below for an
explanation.

The number $f^\lambda$ of SYTs of shape $\lambda \vdash_d k$ is
\begin{center}
\fcolorbox{purple}{lightgray}{
\parbox{16cm}{
\begin{eqnarray}
f^\lambda = k!\frac{V(\mu_1,\ldots,\mu_d)}{\mu_1!\cdots
\mu_d!},~d=\ell(\lambda),
\end{eqnarray}}}
\end{center}
where $V(\mu_1,\ldots,\mu_d)$ denotes the \emph{difference product}
of a non-increasing sequence
$$
V(\mu_1,\ldots,\mu_d) := \prod_{1\leqslant i<j\leqslant d} (\mu_i -
\mu_j),
$$
and the numbers $\mu_j = \mu_j(\lambda)$ are defined by
$$
\mu_j(\lambda) := \lambda_j + \ell(\lambda) -
j,\text{~for~}j=1,2,\ldots, \ell(\lambda).
$$

\subsection{Permutations}

We can display a permutation $\pi$ using \emph{cycle notation}.
Given $j\in\set{1,\ldots,k}:=[k]$, the elements of the sequence
$j,\pi(j),\ldots$ cannot be distinct. Taking the first power $n$
such that $\pi^n(j)=j$, we have the cycle
$$
(j,\pi(j),\ldots,\pi^{n-1}(j)).
$$
Equivalently, the cycle $(i,j,\ldots, l)$ means that $\pi$ sends $i$
to $j,\ldots $, and $l$ back to $i$. Now pick an element not in the
cycle containing $i$ and iterate this process until all members of
$[k]$ have been used. For example $\pi\in S_5$, $\pi=(1,2,3)(4)(5)$
in cycle notation. Note that cyclically permuting the elements
within a cycle or reordering the cycles themselves does not change
the permutation. Thus
$$
(1,2,3)(4)(5) = (2,3,1)(4)(5) = (4)(2,3,1)(5) = (4)(5)(3,1,2).
$$
A $k$-cycle, or \emph{cycle of length} $k$, is a cycle containing
$k$ elements. The \emph{cycle type}, or simply the \emph{type}, of
$\pi$ is an expression of the form
$$
(1^{m_1},2^{m_2},\ldots,k^{m_k}),
$$
where $m_k$ is the number of cycles of length $k$ in $\pi$. A
1-cycle of $\pi$ is called a \emph{fixed-point}. Fixed-points are
usually dropped from the cycle notation if no confusion will result.
\blue{It is easy to see that a permutation $\pi$ such that
$\pi^2=\mathrm{id}$ if and only if all of $\pi$'s cycles have length
1 or 2}.

Another way to give the cycle type is as a partition. A
\emph{partition} of $k$ is a sequence
$$
\lambda = (\lambda_1,\lambda_2,\ldots,\lambda_d),
$$
where the $\lambda_i$ are weakly decreasing and
$\sum^d_{i=1}\lambda_i=k$. Thus $\pi=(1,2,3)(4)(5)$ corresponds to a
partition $(3,1,1)$, and a cycle type $(1^2,2^0,3^1,4^0,5^0)$.

In $S_k$, it is not hard to see that if
\begin{center}
\fcolorbox{purple}{lightgray}{
\parbox{12cm}{
$$
\pi = \blue{(i_{11},i_{12},\ldots,i_{1j_1})}\cdots
\red{(i_{m1},i_{m2},\ldots,i_{mj_m})}
$$}}
\end{center}
in cycle notation, then for any $\sigma\in S_k$
\begin{center}
\fcolorbox{purple}{lightgray}{
\parbox{12cm}{
$$
\sigma\pi \sigma^{-1} =
\blue{(\sigma(i_{11}),\sigma(i_{12}),\ldots,\sigma(i_{1j_1}))}\cdots
\red{(\sigma(i_{m1}),\sigma(i_{m2}),\ldots,\sigma(i_{mj_m}))}.
$$}}
\end{center}
It follows that two permutations are in the same conjugate class if
and only if they have the \emph{same cycle type}. Thus \blue{there
is a natural one-to-one correspondence between partitions of $k$ and
conjugate classes of $S_k$}.

We can compute the size of a conjugate class in the following
manner. Let $G$ be any group and consider the \emph{centralizer} of
$g\in G$ defined by
\begin{center}
\fcolorbox{purple}{lightgray}{
\parbox{6cm}{
$$
Z_g := \Set{h\in G: hgh^{-1}=g},
$$}}
\end{center}
i.e., the set of all elements that commute with $g$. Now, there is a
bijection between the cosets of $Z_g$ and the elements of $K_g$,
where $K_g$ is the conjugate class of $g$---the set of all elements
conjugate to a given $g$ , so that
\begin{center}
\fcolorbox{purple}{lightgray}{
\parbox{6cm}{
$$
\abs{K_g} = \frac{\abs{G}}{\abs{Z_g}}.
$$}}
\end{center}
Now let $G=S_k$ and use $K_\gamma$ for $K_g$ when $g$ has type
$\gamma$. Thus if $\gamma=(1^{m_1},2^{m_2},\ldots,k^{m_k})$ and
$g\in S_k$ has type $\gamma$, then $\abs{Z_g}$ depends only on
$\gamma$ and
\begin{center}
\fcolorbox{purple}{lightgray}{
\parbox{7cm}{
$$
z_\gamma\defeq \abs{Z_g}= 1^{m_1}m_1!2^{m_2}m_2!\cdots k^{m_k}m_k!.
$$}}
\end{center}
The number $\abs{K_g}$ of elements in a conjugacy class $\gamma$ of
$S_k$, denoted $h_\gamma$, is given by
$$
h_\gamma = \frac{k!}{z_\gamma}.
$$
We know that every permutation $\pi\in S_k$ decomposes uniquely as a
product of disjoint cycles. The orders of the cycles, sorted in
non-increasing order, determine the cycle type of the permutation.
Evidently, the cycle type of a permutation $\pi\in S_k$ is a
partition of $k$. We will denote the cycle type of a permutation
$\pi\in S_k$ by $\gamma = \gamma(\pi)\vdash k$. We shall identify
the conjugacy classes with their cycle type, and even write
$\pi\in\gamma$ for a permutation $\pi$ with cycle type $\gamma$.

For instance, $h_{(k)} = (k-1)!$ and $h_{(1^k)}=1$. Obviously, we
need to have $\sum_{\gamma\vdash k} \frac1{z_\gamma} =1$.

\subsection{Products of power sums}

For an integer $r\geqslant1$, the $r$-th \emph{power sum} in the
variables $x_j$ is $p_r = \sum_j x^r_j$. For a partition
$\gamma\vdash_rk$, the \emph{power sum products} $p_\gamma$ are
defined by
\begin{eqnarray}
p_\gamma &:=& p_{\gamma_1}p_{\gamma_2}\cdots p_{\gamma_r}=\Pa{\sum_j
x^{\gamma_1}_j}\Pa{\sum_j x^{\gamma_2}_j}\cdots\Pa{\sum_j
x^{\gamma_r}_j}.
\end{eqnarray}
As a special case, $p_\gamma(1^{\times d}) = d^r$, where
$r=\ell(\gamma)$ is nothing but the number of cycles in $\gamma$.

\subsection{Schur functions}

To define the Schur symmetric functions, or S-functions, it is best
to start with the polynomial case, i.e. with a finite number $d$ of
variables $x_1,\ldots, x_d$. The complete set of S-functions is
obtained by letting $d$ tend to infinity. The S-functions
$s_\lambda$ of $d$ variables and of homogeneity order $k$ are
labeled by partitions $\lambda\vdash_dk$, and are defined by
\begin{center}
\fcolorbox{purple}{lightgray}{
\parbox{16.3cm}{
\begin{eqnarray}
s_\lambda(x_1,\ldots,x_d):=
\frac{\det\Pa{x^{\lambda_j+d-j}_i}^d_{i,j=1}}{\det\Pa{x^{d-j}_i}^d_{i,j=1}}
\end{eqnarray}}}
\end{center}
(recall again that for $j>\ell(\lambda),\lambda_j=0$). For
$\ell(\lambda)>d$, one again has $s_\lambda(x_1,\ldots,x_d)=0$. If
some variables assume equal values, a limit has to be taken, since
both numerator and denominator vanish in that case.

The denominator in the definition of the S-function is a Vandermonde
determinant and is thus equal to $V(x_1,\ldots,x_d)$. The numerator
is divisible (in the ring of polynomials) by each of the differences
$x_i-x_j$, and therefore also by the denominator; hence the
S-functions in a finite number of variables really are polynomials.

For the important case where all $d$ variables assume the value 1
(i.e. giving the number of semi-standard Young tableaux of $d$
objects and of shape $\lambda$), we get, for $\ell(\lambda)\leqslant
d$:
\begin{center}
\fcolorbox{purple}{lightgray}{
\parbox{16.3cm}{
\begin{eqnarray}
s_\lambda(1^{\times d}) =
\frac{V(\lambda_1+d-1,\lambda_2+d-2,\ldots,\lambda_d)}{V(d-1,d-2,\ldots,0)},
\end{eqnarray}}}
\end{center}
and, again, $s_\lambda(1^{\times d})=0$ for $\ell(\lambda)>d$.
Note that $V(d-1,d-2,\ldots,0)=1!2!\cdots (d-1)!$. In particular, if
$\lambda=(k)$, one finds that $s_{(k)}(1^{\times d}) =
\binom{k+d-1}{k}$.

\subsection{Characters of the symmetric group and unitary group}

In the case of the symmetric group, the irreps are labeled by Young
diagrams $\lambda$. The character of a permutation $\pi\in S_k$ in
irrep $\lambda$ is denoted $\chi_\lambda(\pi)$. Since characters are
class functions, one only needs to find the characters of any
representative of a conjugacy class, so that one can use the symbol
$\chi_{\lambda,\gamma}$, with
$$
\chi_{\lambda,\gamma} = \chi_\lambda(\pi),~~~\forall \pi\in\gamma.
$$
The character table is the matrix with elements
$\chi_{\lambda,\gamma}$, where $\lambda$ is the row index and
$\gamma$ the column index (assuming lexicographic ordering for
both). As the conjugacy classes of $S_k$ are labeled by partitions
of $k$, there are as many rows as columns, hence the character table
is a square matrix.

The character of the identity permutation $e$ equals the degree of
the representation in the given irrep. One can show that this degree
is equal to the number of standard Young tableaux of shape $\lambda$
\begin{center}
\fcolorbox{purple}{lightgray}{
\parbox{6cm}{
$$
\chi_\lambda(e) = f^\lambda.
$$}}
\end{center}
The characters in irrep $\lambda=(k)$ are all 1:
$$
\chi_{(k),\gamma} = 1,~~~\forall \gamma\vdash k.
$$
Thus $f^{(k)}=1$. For $\gamma$ consisting of one cycle,
$\gamma=(k)$, the characters are
$$
\chi_{\lambda,(k)} =
\begin{cases}
(-1)^d,&\lambda = (k-d,1^d),0\leqslant d\leqslant k\\
0,&\text{otherwise}.
\end{cases}
$$
In what follows, We now briefly consider the irreducible polynomial
representations of the full linear group $\rG\rL(d,\complex)$ (note
that both the full linear group $\rG\rL(d,\complex)$ and the unitary
group $\U(d)$ embrace the same irreps). There representations get
their name from the fact that their matrix elements are polynomials
in the elements of the represented matrix. Just like the irreps of
the symmetric group, the polynomial irreps of $\rG\rL(d,\complex)$
are labeled by Young diagrams. The conjugacy classes of
$\rG\rL(d,\complex)$ consist of all matrices
$\bsA\in\rG\rL(d,\complex)$ have the same eigenvalues
$(a_1,\ldots,a_d)$ and thus can be labeled by these eigenvalues. The
simple characters (known, in this context, as characteristics) are
denoted $\phi_\lambda(\bsA) = \phi_\lambda(a_1,\ldots,a_d)$.
According to a famous result by Schur, these characters are the
Schur functions (polynomials) of the eigenvalues
$$
\phi_\lambda(a_1,\ldots,a_d) = s_\lambda(a_1,\ldots,a_d).
$$

\subsection{Representations of $S_k$ and $\rG\rL(d,\complex)$ on the tensor product space $(\complex^d)^{\ot k}$}

Here we have denoted the dimension of the subspace $\bQ_\lambda$ by
$t^\lambda(d)$, and the dimension of $\bP_\lambda$ by $f^\lambda$.
The matrix $\bQ_\lambda(\bsA)$ is an irrep of
$A\in\rG\rL(d,\complex)$ of degree $t^\lambda(d)$, operating on
$\bQ_\lambda$. The matrix $\bP_\lambda(\pi)$ is an irrep of $\pi\in
S_k$ of degree $f^\lambda$, operating on $\bP_\lambda$.

Taking traces yields the corresponding simple characters
\begin{eqnarray}
\begin{cases}
\Tr{\bQ_\lambda(\bsA)} &= s_\lambda(a_1,\ldots,a_d),\\
\Tr{\bP_\lambda(\pi)} &= \chi_\lambda(\pi) =
\chi_{\lambda,\gamma(\pi)},
\end{cases}
\end{eqnarray}
where $a_1,\ldots,a_d$ are the eigenvalues of $\bsA$. For the
dimensions one finds
\begin{eqnarray}
\begin{cases}
t^\lambda(d) &= \Tr{\bQ_\lambda(\I_d)} = s_\lambda(1^{\times d}),\\
f^\lambda &= \Tr{\bP_\lambda(e)} = \chi_\lambda(e),
\end{cases}
\end{eqnarray}
i.e. $t^\lambda(d)$ is the number of semi-standard Young tableaux
$\lambda$ of $d$ objects, and $f^\lambda$ is the number of standard
Young tableaux $\lambda$.

In accordance with these decompositions, the tensor space
$(\complex^d)^{\ot k}$ splits up into invariant subspaces. The
subspaces $\bQ_\lambda\ot\bP_\lambda$ are invariant under all
$\bsA^{\ot k}$ and all $\bP(\pi)$. They are further reducible into
direct sums of $f^\lambda$ subspaces of dimension $t^\lambda(d)$,
invariant under the transformations $\bsA^{\ot k}$ but no longer
invariant under permutations $\bP(\pi)$. These irreducible invariant
subspaces are called the symmetry classes of the tensor space. They
are labeled by standard Young tableaux of shape $\lambda$.

We now consider the invariant subspaces $\bQ_\lambda\ot
\bP_\lambda$, corresponding to the Young diagrams $\lambda$. Their
dimension is $f^\lambda s_\lambda(1^{\times d})$. \blue{We will
denote the projectors on these subspaces by $\bsC_\lambda$ (which is
equivalently $\I_{\bQ_\lambda}\ot\I_{\bP_\lambda}$ under the
Schur-transform)}. They are the sum of the Young projectors
corresponding to the standard Young tableaux $\lambda$. We will
consider the Young projectors themselves in the next subsection. The
projectors $\bsC_\lambda$ form an orthogonal set and add up to the
identity on the full tensor space:
\begin{center}
\fcolorbox{purple}{lightgray}{
\parbox{16.3cm}{
\begin{eqnarray}
\bsC_\lambda \bsC_{\lambda'} =
\delta_{\lambda\lambda'}\bsC_\lambda,~~~\sum_{\lambda\vdash_d
k}\bsC_\lambda = \I_{(\complex^d)^{\ot k}},~~~\Tr{\bsC_\lambda} =
f^\lambda s_\lambda(1^{\times d}).
\end{eqnarray}}}
\end{center}
Consider the conjugacy classes $\gamma$ of $S_k$ with cycle type
$\gamma\vdash k$. We define the "class average" of all permutation
matrices with cycle type $\gamma$ as
\begin{eqnarray*}
\bsC^\gamma:= \frac1{h_\gamma} \sum_{\pi\in\gamma}\bP(\pi).
\end{eqnarray*}
Note the distinction between the notations $\bsC_\lambda$, where the
subscript $\lambda$ labels an irrep, and $\bsC^\gamma$, where the
superscript $\gamma$ labels a conjugacy class. Alternatively, we can
write
\begin{eqnarray*}
\bsC^\gamma = \frac1{k!} \sum_{\sigma\in
S_k}\bP(\sigma\pi\sigma^{-1}).
\end{eqnarray*}
The projectors $\bsC_\lambda$ can be expressed in terms of the
permutations $\bP(\pi)$, according to a general relation, as:
\begin{center}
\fcolorbox{purple}{lightgray}{
\parbox{16.3cm}{
\begin{eqnarray}\label{eq:central-proj}
\bsC_\lambda = \frac{f^\lambda}{k!} \sum_{\pi\in
S_k}\chi_\lambda(\pi)\bP(\pi),
\end{eqnarray}}}
\end{center}
and in terms of $p^\gamma$ as:
\begin{eqnarray}\label{eq:central-proj}
\bsC_\lambda = f^\lambda \sum_{\gamma\vdash_d k}
\frac1{z_\gamma}\chi_{\lambda,\gamma}\bsC^\gamma.
\end{eqnarray}
Let $\bsA$ be a matrix with eigenvalues $(a_1,\ldots,a_d)$. Taking
the trace of one $\lambda$-term in the following expression:
$$
\bsA^{\ot k} \cong \bigoplus_{\lambda\vdash_d k}
\bQ_\lambda(\bsA)\ot \I_{\bP_\lambda}
$$
immediately yields $\bsC_\lambda \bsA^{\ot k}\bsC_\lambda \cong
\bQ_\lambda(\bsA)\ot \I_{\bP_\lambda}$, and
\begin{eqnarray*}
\Tr{\bsC_\lambda \bsA^{\ot k}} = f^\lambda
s_\lambda(a_1,\ldots,a_d).
\end{eqnarray*}
For $\pi\in \gamma\vdash k$, it is easy to see that
\begin{eqnarray*}
\Tr{\bP(\pi) \bsA^{\ot k}} = \Tr{\bsC^\gamma \bsA^{\ot k}} =
p_\gamma(a_1,\ldots,a_d).
\end{eqnarray*}
Combining this with \eqref{eq:central-proj} gives the famous
Frobenius formula, relating the characteristics of the full linear
group to the characters of the symmetric group
\begin{eqnarray*}
s_\lambda(a_1,\ldots,a_d) = \sum_{\gamma\vdash_d k}
\frac1{z_\gamma}\chi_{\lambda,\gamma}p_\gamma(a_1,\ldots,a_d).
\end{eqnarray*}
As this holds for any $A$, and thus for any set of values $a_j$ of
whatever dimension, it yields the transition matrix from the
$p_\gamma$ symmetric functions to the S-functions
\begin{eqnarray*}
s_\lambda = \sum_{\gamma\vdash_d k} \frac1{z_\gamma}
\chi_{\lambda,\gamma}p_\gamma.
\end{eqnarray*}
Using the orthogonality relations of the characters, we find
\begin{eqnarray*}
\bsC^\gamma = \sum_{\lambda\vdash_d k}
\frac1{f^\lambda}\chi_{\lambda,\gamma}\bsC_\lambda,~~~ p_\gamma =
\sum_{\lambda\vdash_d k} \frac1{f^\lambda}\chi_{\lambda,\gamma}
s_\lambda.
\end{eqnarray*}

\subsection{Symmetric functions and representations of tensor products}

A property of index permutation matrices that is both simple and
powerful is that index permutation matrices over tensor products of
Hilbert spaces are tensor products themselves. With a minor abuse of
notation we identify $(\cH_A\ot\cH_B)^{\ot k}$ with $\cH_A^{\ot
k}\ot\cH_B^{\ot k}$ and write
\begin{eqnarray*}
\widehat\bP(\pi)(\cH_A\ot\cH_B) = \bP(\pi)(\cH_A)\ot
\bP(\pi)(\cH_B).
\end{eqnarray*}
Here $\bP(\pi)(\cH_A)$ acts on $\cH^{\ot k}_A$, and
$\bP(\pi)(\cH_B)$ acts on $\cH^{\ot k}_B$. Clearly
$\widehat\bP(\pi)(\cH_A\ot\cH_B)$ acts on $(\cH_A\ot\cH_B)^{\ot k}$.
As a short hand, the above equation can be written as
$$
\bP^{AB}(\pi) = \bP^A(\pi)\ot\bP^B(\pi).
$$
This corresponds to considering symmetric functions of tensor
products of variables. If $x=(x_1,x_2,\ldots)$ and
$y=(y_1,y_2,\ldots)$, then their tensor product, which is
customarily denote $xy$ rather than $x\ot y$, consists of all
possible products $x_iy_j$. For power product sums one immediately
sees
\begin{eqnarray*}
p_\gamma(xy) =p_\gamma(x)p_\gamma(y).
\end{eqnarray*}
This yields for Schur functions
\begin{eqnarray}\label{eq:KC}
s_\lambda(xy) = \sum_{\mu,\nu\vdash k}
g_{\lambda\mu\nu}s_\mu(x)s_\nu(y),
\end{eqnarray}
where $g_{\lambda\mu\nu}$ are the so-called \emph{Kronecker
coefficients}
\begin{eqnarray*}
g_{\lambda\mu\nu} := \frac1{k!} \sum_{\pi\in S_k}
\chi_\lambda(\pi)\chi_\mu(\pi)\chi_\nu(\pi) = \sum_{\gamma\vdash_d
k}\frac1{z_\gamma}\chi_{\lambda,\gamma}\chi_{\mu,\gamma}\chi_{\nu,\gamma}.
\end{eqnarray*}
One of the rare cases in which a closed formula can be given for the
Kronecker coefficients, is $\lambda = (k)$. One finds
$$
g_{(k)\mu\nu} =\delta_{\mu\nu}~~~\text{and}~~~ s_{(k)}(xy) =
\sum_{\lambda\vdash k} s_\lambda(x)s_\lambda(y).
$$
A consequence of \eqref{eq:KC} is that for $\bsX$ and $\bsY$, acting
on $\cH_A$ and $\cH_B$, respectively,
\begin{eqnarray*}
\frac1{f^\lambda}\Tr{\bsC_\lambda(\bsX\ot \bsY)^{\ot k}} =
\sum_{\mu,\nu\vdash k}g_{\lambda\mu\nu}
\Pa{\frac1{f^\mu}\Tr{\bsC_\mu \bsX^{\ot
k}}}\Pa{\frac1{f^\nu}\Tr{\bsC_\nu \bsY^{\ot k}}},
\end{eqnarray*}
where $\bsC_\lambda$ acts on $(\cH_A\ot\cH_B)^{\ot k}$, $\bsC_\mu$
on $\cH^{\ot k}_A$, and $C_\nu$ on $\cH^{\ot k}_B$. In terms of the
irrpes of $\rG\rL(d,\complex)$ we have
\begin{eqnarray*}
\bQ_\lambda(\bsX\ot \bsY) \cong \bigoplus_{\mu,\nu\vdash k}
g_{\lambda\mu\nu} \bQ_\mu(\bsX) \ot \bQ_\nu(\bsY),
\end{eqnarray*}
where $g_{\lambda\mu\nu}$ counts the number of copies of
$\bQ_\mu(\bsX) \ot \bQ_\nu(\bsY)$ in the direct sum.

Consider the computation about the partial trace of
$$
\Ptr{B}{\bsC^{AB}_\lambda(\I_{\cH^{\ot k}_A}\ot \bsC^B_\nu)},
$$
where $\bsC^{AB}_\lambda$ acts on $(\cH_A\ot\cH_B)^{\ot k}$ and
$C^B_\nu$ on $\cH^{\ot k}_B$.

Since
\begin{eqnarray*}
\bsC^{AB}_\lambda = \frac{f^\lambda}{k!}\sum_{\pi\in
S_k}\chi_\lambda(\pi)\bP^{AB}(\pi) =
\frac{f^\lambda}{k!}\sum_{\pi\in
S_k}\chi_\lambda(\pi)\bP^A(\pi)\ot\bP^B(\pi),
\end{eqnarray*}
which, together with $\bsC^B_\nu\bP^B(\pi)\bsC^B_\nu =
\I_{\bQ_\nu}\ot \bP_\nu(\pi)$, implies that
\begin{eqnarray*}
\Ptr{B}{\bsC^{AB}_\lambda(\I_{\cH^{\ot k}_A}\ot \bsC^B_\nu)} &=&
\frac{f^\lambda}{k!}\sum_{\pi\in
S_k}\chi_\lambda(\pi)\bP^A(\pi)\Tr{\bP^B(\pi)\bsC^B_\nu}\\
&=&\frac{f^\lambda}{k!}\sum_{\pi\in
S_k}\chi_\lambda(\pi)\bP^A(\pi)s_\nu(1^{\times d_B})\chi_\nu(\pi)\\
&=&\frac{f^\lambda s_\nu(1^{\times d_B})}{k!}\sum_{\pi\in
S_k}\chi_\lambda(\pi)\chi_\nu(\pi)\bP^A(\pi)\\
&=& f^\lambda s_\nu(1^{\times d_B})\sum_{\gamma\vdash
k}\frac1{z_\gamma}\chi_{\lambda,\gamma}\chi_{\nu,\gamma}\bsC^\gamma\\
&=& f^\lambda s_\nu(1^{\times d_B})\sum_{\gamma\vdash
k}\frac1{z_\gamma}\chi_{\lambda,\gamma}\chi_{\nu,\gamma}\Pa{\sum_{\mu\vdash k}\frac1{f^\mu}\chi_{\mu,\gamma}\bsC^A_\mu}\\
&=&f^\lambda s_\nu(1^{\times d_B})\sum_{\mu\vdash
k}\frac1{f^\mu}\Pa{\sum_{\gamma\vdash k}\frac1{z_\gamma}\chi_{\lambda,\gamma}\chi_{\mu,\gamma}\chi_{\nu,\gamma}}\bsC^A_\mu\\
&=&f^\lambda s_\nu(1^{\times d_B})\sum_{\mu\vdash
k}\frac{g_{\lambda\mu\nu}}{f^\mu}\bsC^A_\mu.
\end{eqnarray*}
Therefore we have
\begin{eqnarray*}
\Ptr{B}{\bsC^{AB}_\lambda(\I_{\cH^{\ot k}_A}\ot \bsC^B_\nu)} =
f^\lambda s_\nu(1^{\times d_B})\sum_{\mu\vdash
k}\frac{g_{\lambda\mu\nu}}{f^\mu}\bsC^A_\mu.
\end{eqnarray*}
This fact implies that
\begin{eqnarray*}
\Ptr{B}{\bsC^{AB}_\lambda} = \sum_{\nu\vdash
k}\Ptr{B}{\bsC^{AB}_\lambda(\I_{\cH^{\ot k}_A}\ot \bsC^B_\nu)} =
\sum_{\nu\vdash k}f^\lambda s_\nu(1^{\times d_B})\sum_{\mu\vdash
k}\frac{g_{\lambda\mu\nu}}{f^\mu}\bsC^A_\mu.
\end{eqnarray*}
In particular, for $\lambda=(k)$,
\begin{eqnarray*}
\Ptr{B}{\bsC^{AB}_{(k)}} = \sum_{\mu\vdash k}\frac{s_\mu(1^{\times
d_B})}{f^\mu}\bsC^A_\mu.
\end{eqnarray*}
In addition, we also have
\begin{eqnarray*}
\Ptr{B}{\bsC^{AB}_\lambda(\bsC^A_\mu\ot \bsC^B_\nu)} = f^\lambda
s_\nu(1^{\times d_B})\frac{g_{\lambda\mu\nu}}{f^\mu}\bsC^A_\mu,
\end{eqnarray*}
implying
\begin{eqnarray*}
\Tr{\bsC^{AB}_\lambda(\bsC^A_\mu\ot \bsC^B_\nu)} = f^\lambda
g_{\lambda\mu\nu}s_\mu(1^{\times d_A})s_\nu(1^{\times d_B}).
\end{eqnarray*}
Summing over all $\lambda\vdash k$ gives rise to
\begin{eqnarray*}
\Pa{f^\mu s_\mu(1^{\times d_A})} \Pa{f^\nu s_\nu(1^{\times d_B})}
&=& \Tr{\bsC^A_\mu\ot \bsC^B_\nu} = \sum_{\lambda\vdash
k}\Tr{\bsC^{AB}_\lambda(\bsC^A_\mu\ot \bsC^B_\nu)}
\\
&=&\Pa{\sum_{\lambda\vdash k}f^\lambda
g_{\lambda\mu\nu}}s_\mu(1^{\times d_A})s_\nu(1^{\times d_B}),
\end{eqnarray*}
implying that $f^\mu f^\nu = \sum_{\lambda\vdash k}f^\lambda
g_{\lambda\mu\nu}$.

\section{Weyl integration formula}
This section is written based on Bump's book \cite{Bump}.
\subsection{Haar measure}

If $G$ is a locally compact group, there is, up to a constant
multiple, a unique regular Borel measure $\mu_L$ that is invariant
under left translation. Here \emph{left translation invariance}
means that $\mu(M) = \mu(gM)$ for all measurable sets $M$.
\emph{Regularity} means that
\begin{eqnarray*}
\mu(M) &=& \inf\set{\mu(\cO): M\subseteq \cO, \cO~\text{open}}\\
&=&\sup\set{\mu(\cC): M\supseteq \cC, \cC~\text{compact}}.
\end{eqnarray*}
Such a measure is called a \emph{left Haar measure}. It has the
properties that any compact set has finite measure and any nonempty
open set has positive measure.

I will not prove the existence and uniqueness of the Haar measure,
which has already established. Left-invariance of the measure
amounts to left-invariance of the corresponding integral,
\begin{eqnarray}
\int_G f(g'g)\dif\mu_L(g) = \int_G f(g)\dif\mu_L(g),
\end{eqnarray}
for any Haar integral function $g$ on $G$.

There is also a right-invariant measure $\mu_R$, unique up to
constant multiple, called a \emph{right Haar measure}. Left and
right Haar measures may or may not coincide. For example, if
$$
G=\Set{\left(
         \begin{array}{cc}
           y & x \\
           0 & 1 \\
         \end{array}
       \right): x,y\in \real,y>0
},
$$
then it is easy to see that the left- and right-invariant measures
are, respectively,
$$
\dif\mu_L=y^{-2}\dif x\dif y,~~~\dif\mu_R=y^{-1}\dif x\dif y.
$$
They are not the same. However, there are many cases where they do
coincide, and if the left Haar measure is also right-invariant, we
call $G$ \emph{unimodular}.

Conjugation is an automorphism of $G$, and so it takes a left Haar
measure to another left Haar measure, which must be a constant
multiple of the first. Indeed,
\begin{eqnarray*}
\int_G f(x^{-1}gx)\dif\mu_L(g) = \int_G f(g)\dif\mu_L(xgx^{-1}) =
\int_G f(g)\dif\mu_L(gx^{-1}).
\end{eqnarray*}
Clearly $d\mu^x_L(g) := d\mu_L(gx^{-1})$ defines a new left Haar
measure. By the uniqueness of left Haar measure, up to constant
multiple, $d\mu^x_L(g)=\delta(x)d\mu_L(g)$, which implies that
\begin{eqnarray}\label{eq:one-dimen-homo}
\int_G f(x^{-1}gx)\dif\mu_L(g) = \delta(x)\int_G f(g)\dif\mu_L(g).
\end{eqnarray}
\begin{prop}
The function $\delta:G\to\real^\times_+$ is a continuous
homomorphism. The measure $\delta(g)\mu_L(g)$ is a right-invariant,
denoted $\mu_R(g)$.
\end{prop}

\begin{proof}
Conjugation by first $x_1$ and then $x_2$ is the same as conjugation
by $x_1x_2$ in one step. This can be seen from the following
reasoning: Let $x=x_1x_2$ in \eqref{eq:one-dimen-homo}, we have
\begin{eqnarray*}
\int_G f(x_2^{-1}x_1^{-1}gx_1x_2)\dif\mu_L(g) = \delta(x_1x_2)\int_G
f(g)\dif\mu_L(g)
\end{eqnarray*}
and
\begin{eqnarray*}
&&\int_G f(x_2^{-1}x_1^{-1}gx_1x_2)\dif\mu_L(g) = \int_G
f_{x_2}(x_1^{-1}gx_1)\dif\mu_L(g)
= \delta(x_1)\int_G f_{x_2}(g)\dif\mu_L(g)\\
&&=\delta(x_1)\int_G f(x_2^{-1}gx_2)\dif\mu_L(g) =
\delta(x_1)\delta(x_2)\int_G f(g)\dif\mu_L(g)
\end{eqnarray*}
where $f_{x_2}(g):= f(x_2^{-1}gx_2)$. That is
$$
\delta(x_1x_2)=\delta(x_1)\delta(x_2).
$$
Replace $f$ by $f\delta$ in the following
\begin{eqnarray*}
\int_G f(gx)\dif\mu_L(g) = \delta(x)\int_G f(g)\dif\mu_L(g)
\end{eqnarray*}
we get
\begin{eqnarray*}
\int_G f(gx)\delta(gx)\dif\mu_L(g) = \delta(x)\int_G
f(g)\delta(g)\dif\mu_L(g),
\end{eqnarray*}
which gives rise to
\begin{eqnarray*}
\int_G f(gx)\delta(g)\dif\mu_L(g) = \int_G
f(g)\delta(g)\dif\mu_L(g),
\end{eqnarray*}
that is
\begin{eqnarray*}
\int_G f(gx)\dif\mu_R(g) = \int_G f(g)\dif\mu_R(g),
\end{eqnarray*}
completing the proof.
\end{proof}

\begin{prop}
If $G$ is compact, then $G$ is unimodular and $\mu_L(G)<\infty$.
\end{prop}

\begin{proof}
Since $\delta$ is a homomorphism, the image of $\delta$ is a
subgroup of $\real^\times_+$. Since $G$ is compact, $\delta(G)$ is
also compact, and the only compact subgroup of $\real^\times_+$ is
just $\set{1}$. Thus $\delta$ is trivial, so a left Haar measure is
right-invariant. We have mentioned as assumed fact that the Haar
volume of any compact subset of a locally compact group is finite,
so if $G$ is finite, its Haar volume is finite.
\end{proof}
If $G$ is compact, then it is natural to normalize the Haar measure
so that $G$ has volume 1. To simplify our notation, we will denote
$\int_G f(g)\dif\mu_L(g)$ by $\int_Gf(g)\dif\mu(g)$.

\begin{prop}
If $G$ is unimodular, then the map $g\to g^{-1}$ is an isometry.
\end{prop}

\begin{proof}
It is easy to see that $g\to g^{-1}$ turns a left Haar measure into
a right Haar measure. If left and right Haar measures agree, then
$g\to g^{-1}$ multiplies the left Haar measure by a positive
constant, which must be 1 since the map has order 2.
\end{proof}

\subsection{Weyl integration formula}

Let $G$ be a compact, connected Lie group, and let $T$ be a maximal
torus. It is already known that every conjugacy class meets $T$.
Thus we should be able to compute the Haar integral over $G$. The
following formula that allows this, the \emph{Weyl integration
Formula}, is therefore fundamental in representation theory and in
other areas, such as random matrix theory.
\begin{eqnarray}
\int_G f(g)\dif\mu(g) = \frac1{\abs{W(G)}}\int_T
\Pa{\int_{G/T}f(gtg^{-1})\abs{\det(\I-\Ad(t))}\dif\mu(gT)}\dif t.
\end{eqnarray}

If $G$ is a locally compact group and $H$ a closed subgroup, then
the quotient space $G/H$ consisting of all cosets $gH$ with $g\in
G$, given the quotient topology, is a locally compact Hausdorff
space.

If $X$ is a locally compact Hausdorff space let $C_c(X)$ be the
space of continuous, compactly supported functions on $X$. If $X$ is
a locally compact Hausdorff space, a linear functional $\cI$ on
$C_c(X)$ is called \emph{positive} if $\cI(f)\geqslant0$ if $f$ is
nonnegative. According to the \emph{Riesz representation theorem},
every such $\cI$ is of the form
\begin{eqnarray*}
\cI(f) = \int_X f \dif\mu
\end{eqnarray*}
for some regular Borel measure $\dif\mu$.

\begin{prop}
Let $G$ be a locally compact group, and let $H$ be a compact
subgroup. Let $d\mu_G$ and $d\mu_H$ be left Haar measures on $G$ and
$H$, respectively. Then there exists a regular Borel measure
$d\mu_{G/H}$ on $G/H$ which is invariant under the action of $G$ by
left translation. The measure $d\mu_{G/H}$ may be normalized so
that, for $f\in C_c(G)$, we have
\begin{eqnarray}
\int_{G/H} \Pa{\int_H f(gh)d\mu_H(h)}\dif\mu_{G/H}(gH).
\end{eqnarray}
Here the function $g\mapsto \int_Hf(gh)\dif\mu_H(h)$ is constant on
the cosets $gH$, and we are therefore identifying it with a function
on $G/H$.
\end{prop}

\begin{proof}
We may choose the normalization of $\dif\mu_H$ so that $H$ has total
volume 1. We define a map $\Lambda: C_c(G)\to C_c(G/H)$ by
\begin{eqnarray*}
(\Lambda f)(g) = \int_H f(gh)\dif\mu_H(h).
\end{eqnarray*}
Note that $\Lambda f$ is a function on $G$ which is right invariant
under translation by elements of $H$, so it may be regarded as a
function on $G/H$. Since $H$ is compact, $\Lambda f$ is compactly
supported. If $\phi\in C_c(G/H)$, regarding $\phi$ as a function on
$G$, we have $\Lambda\phi=\phi$ because
\begin{eqnarray*}
(\Lambda \phi)(g) = \int_H \phi(gh)\dif\mu_H(h) = \int_H
\phi(g)\dif\mu_H(h) = \phi(g).
\end{eqnarray*}
This shows that $\Lambda$ is surjective. We may therefore define a
linear functional $\cI$ on $C_c(G/H)$ by
\begin{eqnarray*}
\cI(\Lambda f) = \int_G f(g)\dif\mu_G(g),~~~f\in C_c(G)
\end{eqnarray*}
provided we check that this is well-defined. We must show that if
$\Lambda f=0$, then
\begin{eqnarray*}
\cI(\Lambda f) = 0,
\end{eqnarray*}
i.e. $\int_G f(g)d\mu_G(g)=0$. We note that the function
$(g,h)\mapsto f(gh)$ is compactly supported and continuous on
$G\times H$, so if $\Lambda f=0$, we may use Fubini's theorem to
write
\begin{eqnarray*}
0&=&\int_G (\Lambda f)(g) \dif\mu_G(g) = \int_G\Pa{\int_H
f(gh)d\mu_H(h)}\dif\mu_G(g)\\
&=&\int_H\Pa{\int_G f(gh)\dif\mu_G(g)}\dif\mu_H(h).
\end{eqnarray*}
In the inner integral on the right-hand side we make the variable
change $g\mapsto gh^{-1}$. Recalling that $d\mu_G(g)$ is \emph{left}
Haar measure, this produces a factor of $\delta_G(h)$, where
$\delta_G(h)$ is the modular homomorphism. Thus
$$
0=\int_H \delta_G(h) \Pa{\int_G f(g)d\mu_G(g)}\dif\mu_H(h).
$$
Now the group $H$ is compact, so its image under $\delta_G$ is a
compact subgroup of $\real^\times_+$, which must be $\set{1}$. Thus
$\delta_G(h)=1$ for all $h\in H$, and we obtain $\int_G
f(g)d\mu_G(g)=0$, justifying the definition of the functional $\cI$.
The existence of the measure on $G/H$ now follows from the Riesz
representation theorem.
\end{proof}

\begin{exam}
Suppose that $G=\U(d)$. A maximal torus is
$$
\mathbb{T} = \Set{\mathrm{diag}(t_1,\ldots,t_n):
\abs{t_1}=\cdots=\abs{t_n}=1 }.
$$
Its normalizer $N(\mathbb{T})$ consists of all monomial matrices
(matrices with a single nonzero entry in each row and column) so
that the quotient $N(\mathbb{T})/\mathbb{T}\cong S_n$.
\end{exam}

\begin{prop}
Let $T$ be a maximal torus in the compact connected Lie group $G$,
and let $\liet,\g$ be the Lie algebras of $T$ and $G$, respectively.
\begin{enumerate}[(i)]
\item Any vector in $\g$ fixed by $\Ad(T)$ is in $\liet$.
\item We have $\g=\liet\oplus \liet^\perp$, where $\liet^\perp$ is
invariant under $\Ad(T)$. Under the restriction of $\Ad$ to $T$,
$\liet^\perp$ decomposes into a direct sum of two-dimensional real
irreps of $T$.
\end{enumerate}
\end{prop}
Let $W(G)$ be the Weyl group of $G$. The Weyl group acts on $T$ by
conjugation. Indeed, the elements of the Weyl group are cosets
$w=nT$ for $n\in N(T)$. If $t\in T$, the elements $ntn^{-1}$ depends
only on $w$ so by abuse of notation we denote it $wtw^{-1}$.

\begin{thrm}
(i) Two elements of $T$ are conjugate in $G$ if and only if they are
conjugate in $N(T)$.\\
(ii) The inclusion $T\to G$ induces a bijection between the orbits
of $W(G)$ on $T$ and the conjugacy classes of $G$.
\end{thrm}

\begin{proof}
Suppose that $t,u\in T$ are conjugate in $G$, say $gtg^{-1}=u$. Let
$H$ be the connected component of the identity in the centralizer of
$u$ in $G$. It is a closed Lie subgroup of $G$. Both $T$ and
$gTg^{-1}$ are contained in $H$ since they are connected commutative
groups containing $u$. As they are maximal tori in $G$, they are
maximal tori in $H$, and so they are conjugate in the compact
connected group $H$. If $h\in H$ such that $hTh^{-1}=gTg^{-1}$, then
$w=h^{-1}g\in N(T)$. Since $wtw^{-1}=h^{-1}uh=u$, we see that $t$
and $u$ are conjugate in $N(T)$.

Since $G$ is the union of the conjugates of $T$, (ii) is a
restatement of (i).
\end{proof}

\begin{prop}
The centralizer $C(T)=T$.
\end{prop}

\begin{prop}
There exists a dense open set $\Omega$ of $T$ such that the
$\abs{W(G)}$ elements $wtw^{-1}(w\in W(G))$ are all distinct for
$t\in \Omega$.
\end{prop}

\begin{proof}
If $w\in W(G)$, let
$$
\Omega_w=\set{t\in T: wtw^{-1}\neq t}.
$$
It is an open subset of $T$ since its complement is evidently
closed. If $w\neq\I$ and $t$ is a generator of $T$, then
$t\in\Omega_w$ because otherwise if $n\in N(T)$ represents $w$, then
$n\in C(t)=C(T)$, so $n\in T$. This is a contradiction since
$w\neq\I$. By Kronecker Theorem, it follows that $\Omega_w$ is a
dense open set. The finite intersection
$\Omega=\bigcap_{w\neq\I}\Omega_w$ thus fits our requirements.
\end{proof}

\begin{thrm}[Weyl]
If $f$ is a class function, and if $dg$ and $dt$ are Haar measures
on $G$ and $T$ (normalized so that $G$ and $T$ have volume 1), then
\begin{eqnarray}
\int_G f(g)\dif\mu(g) = \frac1{\abs{W(G)}}\int_T
f(t)\det\Pa{\Br{\Ad(t^{-1})-\I_{\liet^\perp}}|_{\liet^\perp}}\dif t
\end{eqnarray}
\end{thrm}

\begin{proof}
Let $\cX=G/T$. We give $\cX$ the measure $d_\cX$ invariant under
left translation by $G$ such that $\cX$ has volume 1. Consider the
map
$$
\phi: \cX\times T\to G,~~~\phi(xT,t)=xtx^{-1}.
$$
Both $\cX\times T$ and $G$ are orientable manifolds of the same
dimension. Of course, $G$ and $T$ both are given the Haar measures
such that $G$ and $T$ have volume 1.

We choose volume elements on the Lie algebras $\g$ and $\liet$ of
$G$ and $T$, respectively, so that the Jacobians of the exponential
maps $\g\to G$ and $\liet\to T$ at the identity are $\I$.

We compute the Jacobian $J\phi$ of $\phi$. Parameterize a
neighborhood of $xT$ in $\cX$ by a chart based on a neighborhood of
the origin in $\liet^\perp$. This chart is the map
$$
\liet^\perp \ni A\mapsto xe^AT.
$$
We also make use of the exponential map to parameterize a
neighborhood of $t\in T$. This is the chart $\liet\ni B\mapsto
te^B$. We therefore have the chart near the point $(xT,t)$ in
$\cX\times T$ mapping
$$
\liet^\perp\times \liet\ni (A,B)\to (xe^AT,te^B)\in\cX\times T
$$
and, in these coordinates, $\phi$ is the map
$$
(A,B)\mapsto xe^Ate^Be^{-A}x^{-1}.
$$
To compute the Jacobian of this map, we translate on the left by
$t^{-1}x^{-1}$ and on the right by $x$. There is no harm in this
because these maps are Haar isometries. We are reduced to computing
the Jacobian of the map
$$
(A,B)\mapsto t^{-1}e^Ate^Be^{-A} = e^{\Ad(t^{-1}A)}e^Be^{-A}.
$$
Identifying the tangent space of the real vector space
$\liet^\perp\times \liet$ with itself (that is, with
$\g=\liet^\perp\times \liet$), the differential of this map is
$$
A\oplus B\mapsto \Pa{\Ad(t^{-1})-\I_{\liet^\perp}}A\oplus B.
$$
The Jacobian is the determinant of the differential, so
\begin{eqnarray}\label{eq:Jacobian}
(J\phi)(xT,t) =
\det\Pa{\Br{\Ad(t^{-1})-\I_{\liet^\perp}}|_{\liet^\perp}}.
\end{eqnarray}
The map $\phi: \cX\times T\to G$ is a $\abs{W(G)}$-fold cover over a
dense open set and so, for any function $f$ on $G$, we have
\begin{eqnarray*}
\int_G f(g)\dif\mu(g) = \frac1{\abs{W(G)}}\int_{\cX\times T}
f(\phi(xT,t))J(\phi(xT,t))\dif_\cX\times \dif t.
\end{eqnarray*}
The integrand
$f(\phi(xT,t))J(\phi(xT,t))=f(t)\det\Pa{\Br{\Ad(t^{-1})-\I_{\liet^\perp}}|_{\liet^\perp}}$
 is independent of $x$ since $f$ is a class function, and the result
 follows.
\end{proof}

\begin{remark}
Let $G$ be a Lie group and $\g$ its Lie algebra. Identify both
$\rT_0\g$ and $\rT_eG$ with $\g$. Then, $(d\exp)_0: \rT_0\g\to
\rT_eG$ is the identity map. Indeed,
\begin{eqnarray*}
(d\exp)_0(A) = \left.\frac{d}{dt}\right|_{t=0} \exp(0+tA) = A.
\end{eqnarray*}
That is $(d\exp)_0$ is the identity map over $\g$.
\end{remark}

\begin{prop}
Let $G=\unitary{n}$, and let $\mathbb{T}$ be the diagonal torus.
Writing
$$
t=\mathrm{diag}(t_1,\ldots,t_n)\in \mathbb{T},
$$
and letting $\int_\mathbb{T}dt$ be the Haar measure on $\mathbb{T}$
normalized so that its volume is 1, we have
\begin{eqnarray}
\int_G f(g)\dif\mu(g) =
\frac1{n!}\int_\mathbb{T}f(t)\prod_{i<j}\abs{t_i-t_j}^2\dif t.
\end{eqnarray}
\end{prop}

\begin{proof}
We need to check that
$$
\det\Pa{\Br{\Ad(t^{-1})-\I_{\liet^\perp}}|_{\liet^\perp}} =
\prod_{i<j}\abs{t_i-t_j}^2.
$$
To compute this determinant, we may as well consider the linear
transformation induced by $\Ad(t^{-1})-\I_{\liet^\perp}$ on the
complexified vector space $\complex\ot \liet^\perp$. We may identify
$\complex\ot\u(n)$ with $\gl(n,\complex)=M_n(\complex)$. We recall
that $\complex\ot \liet^\perp$ is spanned by the $T$-eigenspaces in
$\complex\ot\u(n)$ corresponding to nontrivial characters of $T$.
There are spanned by the elementary matrices $E_{ij}$ with a 1 in
the $(i,j)$-th position and zeros elsewhere, where $1\leqslant
i,j\leqslant n$ and $i\neq j$. The eigenvalue of $t$ on $E_{ij}$ is
$t_it^{-1}_j$. Hence
\begin{eqnarray*}
\det\Pa{\Br{\Ad(t^{-1})-\I_{\liet^\perp}}|_{\liet^\perp}}
=\prod_{i\neq j}(t_it^{-1}_j - 1) = \prod_{i<j}(t_it^{-1}_j -
1)(t_jt^{-1}_i - 1).
\end{eqnarray*}
Since $\abs{t_i}=\abs{t_j}=1$, we have
$$
(t_it^{-1}_j - 1)(t_jt^{-1}_i - 1) =
(t_i-t_j)(t^{-1}_i-t^{-1}_j)=\abs{t_i-t_j}^2.
$$
This completes the proof.
\end{proof}

\begin{remark}
Let $G=\U(1)=\mathbb{S}^1$, $\rho_n:\mathbb{S}^1\to
\rG\rL(1,\complex)$ be given by $\rho_n(e^{\sqrt{-1}\theta}) =
e^{\sqrt{-1}n\theta}$. Then $\dif\mu(g)=\dif\theta/2\pi$ and
$$
\frac1{2\pi}\int^{2\pi}_{0}
e^{\sqrt{-1}n\theta}e^{-\sqrt{-1}m\theta}\dif\theta=\delta_{mn}.
$$
\end{remark}

\begin{cor}\label{cor:class-integration}
If $f$ is a class function over $\mathsf{U}(n)$, then
\begin{eqnarray}
\int_{\mathsf{U}(n)} f(\bsU)\dif\mu(\bsU) &=&
\frac1{n!}\int_{\mathbb{T}^n}f(\bsD(\theta))J(\theta)\dif\bsD(\theta)\\
&=&\frac1{(2\pi)^nn!}\overbrace{\int^{2\pi}_0\cdots\int^{2\pi}_0}^{n}
f(\bsD(\theta))J(\theta)\dif\theta_1\cdots \dif\theta_n,
\end{eqnarray}
where
$$
\bsD(\theta):=
\mathrm{diag}\Pa{e^{\sqrt{-1}\theta_1},\ldots,e^{\sqrt{-1}\theta_n}}~~\text{and}~~J(\theta)
:=
            \prod_{i<j}\abs{e^{\sqrt{-1}\theta_i}-e^{\sqrt{-1}\theta_j}}^2.
$$
\end{cor}

\begin{remark}
We know that for one-dimensional torus, the normalized Haar measure
is defined as $\dif\mu(u):=\frac{\dif\theta}{2\pi}$ over $\U(1)$.
This implies that for $n$-dimensional torus of $\U(n)$:
$$
\mathbb{T}^n=\overbrace{\U(1)\times\cdots\times \U(1)}^n,
$$
the normalized Haar measure is given by the product measure of $n$
one-dimensional measures of $\U(1)$. Thus for
$\bsD(\theta)\in\mathbb{T}^n$, described by $\bsD(\theta)=
u_1(\theta)\times\cdots u_n(\theta)$ with $\dif
u_j(\theta)=\dif\theta_j/2\pi$, the normalized Haar measure is
defined as
\begin{eqnarray*}
\dif \bsD(\theta) := \dif u_1\times\cdots \times \dif u_n =
\frac{\dif\theta_1}{2\pi}\times\cdots\times\frac{\dif\theta_n}{2\pi}
= \frac1{(2\pi)^n}\dif\theta_1\cdots \dif\theta_n =
\frac1{(2\pi)^n}\dif\theta,
\end{eqnarray*}
where $\dif\theta:= \dif\theta_1\cdots \dif\theta_n$.
\end{remark}
In Corollary~\ref{cor:class-integration}, assume that $f\equiv1$,
then we have
\begin{eqnarray*}
1 &=& \int_{\U(n)}\dif\mu(\bsU) =
\frac1{n!}\int_{\mathbb{T}^n}J(\theta)\dif \bsD(\theta)\\
&=&\frac1{(2\pi)^nn!}\overbrace{\int^{2\pi}_0\cdots\int^{2\pi}_0}^{n}
J(\theta)\dif\theta_1\cdots \dif\theta_n,
\end{eqnarray*}
implying
\begin{eqnarray}\label{eq:n!}
n!=\int_{\mathbb{T}^n}J(\theta)\dif\bsD(\theta) =
\frac1{(2\pi)^n}\overbrace{\int^{2\pi}_0\cdots\int^{2\pi}_0}^{n}
J(\theta)\dif\theta.
\end{eqnarray}
In what follows, we give a check on the identity in \eqref{eq:n!}.
Here is another way of writing $J(\theta)$, which is useful. Set
$e^{\sqrt{-1}\theta_j}=\zeta_j$. Then
\begin{eqnarray*}
J(\theta) = V(\zeta)V(\bar\zeta),
\end{eqnarray*}
where
$$
V(\zeta):=V(\zeta_1,\ldots,\zeta_n) = \prod_{1\leqslant i<j\leqslant
n} (\zeta_j-\zeta_i).
$$
Now $V(\zeta)$ is a Vandermonde determinant:
$$
V(\zeta) = \det\left(
             \begin{array}{cccc}
               1 & 1 & \cdots & 1 \\
               \zeta_1 & \zeta_2 & \cdots & \zeta_n \\
               \vdots & \vdots & \ddots & \vdots \\
               \zeta^{n-1}_1 & \zeta^{n-1}_2 & \cdots & \zeta^{n-1}_n \\
             \end{array}
           \right).
$$
Define $a_{ij}:=\zeta^{i-1}_j~(i,j\in\set{1,\ldots,n})$. We can form
a $n\times n$ matrix $\bsA=[a_{ij}]$ in terms of $a_{ij}$.
Apparently, $V(\zeta)=\det(\bsA)$. According to the definition of
determinant, the expansion of a determinant can be given by
\begin{eqnarray*}
\det(\bsA) = \sum_{\pi\in S_n}\sign(\pi)a_{\pi(1)1}\cdots
a_{\pi(n)n}.
\end{eqnarray*}
Now that $a_{\pi(i)j}=\zeta^{\pi(i)-1}_j$. We thus obtain that
\begin{eqnarray*}
V(\zeta) = \sum_{\pi\in S_n}
\sign(\pi)\zeta^{\pi(1)-1}_1\zeta^{\pi(2)-1}_2\cdots\zeta^{\pi(n)-1}_n.
\end{eqnarray*}
Now $\bar\zeta_j=\zeta^{-1}_j$ for $\zeta_j\in \unitary{1}$, so
\begin{eqnarray*}
J(\theta) = \sum_{(\pi,\sigma)\in S_n\times S_n}
\sign(\pi)\sign(\sigma)\zeta^{\pi(1)-\sigma(1)}_1\zeta^{\pi(2)-\sigma(2)}_2\cdots\zeta^{\pi(n)-\sigma(n)}_n.
\end{eqnarray*}
Hence
\begin{eqnarray*}
&&\int_{\mathbb{T}^n}J(\theta)\dif\bsD(\theta)=\sum_{(\pi,\sigma)\in
S_n\times S_n}\sign(\pi)\sign(\sigma)\int_{\mathbb{T}^n}\dif
\bsD(\theta)\Pa{
\zeta^{\pi(1)-\sigma(1)}_1\zeta^{\pi(2)-\sigma(2)}_2\cdots\zeta^{\pi(n)-\sigma(n)}_n}\\
&&=\sum_{(\pi,\sigma)\in S_n\times S_n}\sign(\pi)\sign(\sigma)
\Pa{\frac1{2\pi}\int^{2\pi}_0e^{\sqrt{-1}(\pi(1)-\sigma(1))\theta}\dif\theta_1}\times\cdots\times\Pa{\frac1{2\pi}\int^{2\pi}_0e^{\sqrt{-1}(\pi(n)-\sigma(n))\theta}\dif\theta_n}\\
&&=\sum_{(\pi,\sigma)\in S_n\times
S_n}\sign(\pi)\sign(\sigma)\delta_{\pi(1)\sigma(1)}\cdots\delta_{\pi(n)\sigma(n)}=\sum_{\pi\in
S_n}1 = n!,
\end{eqnarray*}
where we used the fact that
$$
\frac1{2\pi}\int^{2\pi}_0e^{\sqrt{-1}(\pi(k)-\sigma(k))\theta}\dif\theta_k=\delta_{\pi(k)\sigma(k)}.
$$
Denote $\theta = (\theta_1,\ldots,\theta_n)$ and define functionals
$\alpha_{ij}(\theta) = \theta_i-\theta_j$. We mention another way of
writing $J(\theta)$, i.e.
\begin{eqnarray*}
J(\theta) = A(\theta)\overline{A(\theta)},
\end{eqnarray*}
where
$$
A(\theta) := \prod_{i<j}\Pa{1 - e^{-\sqrt{-1}\alpha_{ij}(\theta)}}.
$$
\begin{remark}
In fact, by using Selberg's integral \cite[Eq.(17.7.1),
pp323]{Mehta}, one gets more in the following:
\begin{eqnarray*}
\frac1{(2\pi)^n}\int^{2\pi}_0\cdots \int^{2\pi}_0 \prod_{1\leqslant
i<j\leqslant
n}\abs{e^{\sqrt{-1}\theta_i}-e^{\sqrt{-1}\theta_j}}^{2\gamma}
\dif\theta_1\cdots \dif\theta_n = \frac{(n\gamma)!}{(\gamma!)^n},
\end{eqnarray*}
where $\gamma\in\natural$. Then, letting $\gamma=1$ gives that
\begin{eqnarray*}
\frac1{(2\pi)^n}\int^{2\pi}_0\cdots \int^{2\pi}_0 \prod_{1\leqslant
i<j\leqslant n}\abs{e^{\sqrt{-1}\theta_i}-e^{\sqrt{-1}\theta_j}}^2
\dif\theta_1\cdots \dif\theta_n = n!,
\end{eqnarray*}
as obtained in the above remark. In addition, there is another
integral of interest is the following \cite[Eq.(17.11.11),
pp331]{Mehta}:
\begin{eqnarray*}
\cI_n(k,\gamma):=\frac1{(2\pi)^n}\frac{(\gamma!)^n}{(n\gamma)!}\int^{2\pi}_0\cdots
\int^{2\pi}_0
\abs{\sum^n_{k=1}e^{\sqrt{-1}\theta_k}}^{2k}\prod_{1\leqslant
i<j\leqslant
n}\abs{e^{\sqrt{-1}\theta_i}-e^{\sqrt{-1}\theta_j}}^{2\gamma}
\dif\theta_1\cdots \dif\theta_n.
\end{eqnarray*}
In fact, if $\gamma=1$, then
\begin{eqnarray*}
\cI_n(k,1)= \int_{\U(n)} \abs{\Tr{\bsU}}^{2k}\dif\mu(\bsU).
\end{eqnarray*}
For $0\leqslant k\leqslant n$, we have $\cI_n(k,1)=k!$. But however
$\cI_n(k,\gamma)$ is not known for general $\gamma>1$.
\end{remark}

\section{Character tables of the permutation groups}\label{app:irrech}

If $K$ is a conjugacy class of a finite group $G$ and $\chi$ is a
character, we can define $\chi_K$ to be the value of the given
character on the given class: $\chi_K=\chi(\pi)$ for any $\pi\in K$.
\begin{definition}[Character table]
Let $G$ be a finite group. The \emph{character table} of $G$ is an
array with rows indexed by the inequivalent irreducible characters
of $G$ and columns indexed by the conjugacy classes. The table entry
in row $\chi$ and column $K$ is $\chi_K$:
\begin{table}[h]
\centering
\begin{tabular}{c|ccc}
   & $\cdots$ & $K$ & $\cdots$\\
 \hline $\vdots$ &  &  &  \\
  $\chi$ &  & $\chi_K$ & \\
  $\vdots$ &  &  &
\end{tabular}\caption{The character table of a group $G$.}
\end{table}

By convention, the first row corresponds to the trivial character,
and the first column corresponds to the class of the identity,
$K=\set{\mathrm{id}}$.
\end{definition}

Since the number of inequivalent irreducible representations of a
finite group $G$ is equal to the number of conjugacy classes, so
\blue{the character table is always square}. Due to this fact, we
can write such character table in terms of a square matrix $\bsT$
with some orders prescribed to rows and columns, respectively. We
will see this point in the following each example.

Next we turn to consider the character table for $G=S_k$, the
permutation group of $k$th order. We denote the character table of
$S_k$ by $\bsT_k$. All irreducible characters $\chi_{\lambda_i}$
(each corresponding to a partition $\lambda_i \vdash k$) of $S_k$
are abbreviated as $\chi_i$. The following properties of irreducible
characters will be used in the calculation:
\begin{itemize}
\item $\Inner{\chi_i}{\chi_j}=\delta_{ij}$, where $\chi_{\lambda_i}, \chi_{\lambda_j}$ are irreducible characters of
$S_k$.
\item
$d^{c(\pi)}=\sum_{\lambda\vdash_dk}d_\lambda\chi_\lambda(\pi)$,
where $d_\lambda=\dim(\bQ_\lambda)$.
\item $\sum_{\chi\in\widehat{S}_k}
\chi(K_i)\chi(K_j)=\frac{\abs{S_k}}{\abs{K_i}}\delta_{ij}$, where
$K_i,K_j$ are conjugacy classes of $S_k$.
\end{itemize}

Let $p_0(x_1,\ldots,x_k)\equiv1$ and
$$
p_j(x_1,\ldots,x_k):=x^j_1+x^j_2+\cdots+x^j_k
$$
for $j=1,2,\ldots$. If
$\lambda=(\lambda_1,\ldots,\lambda_d)\vdash k$, the power sum
symmetric polynomial associated with $\lambda$ is given by
\begin{eqnarray}
p_\lambda(x_1,\ldots,x_k) \defeq p_{\lambda_1}(x_1,\ldots,x_k)
p_{\lambda_2}(x_1,\ldots,x_k)\cdots p_{\lambda_d}(x_1,\ldots,x_k).
\end{eqnarray}
Once $k$ is fixed, $\delta$ denote the partition
$\delta=(k-1,k-2,\ldots,1,0)$. Then the partition
$\lambda=(\lambda_1,\lambda_2,\ldots,\lambda_k)$ satisfies the
condition
$\lambda_1>\lambda_2>\cdots>\lambda_{k-1}>\lambda_k\geqslant0$ if
and only if
$$
\lambda = \mu+\delta =
(\mu_1+k-1,\mu_2+k-2,\ldots,\mu_{k-1}+1,\mu_k),
$$
where $\mu=(\mu_1,\mu_2,\ldots,\mu_k)$ is an arbitrary $k$-parts
partition with $\mu_1\geqslant\mu_2\geqslant\cdots\geqslant
\mu_k\geqslant0$.

Denote $V_{\delta}(x_1,\ldots,x_k):=\prod_{1\leqslant i<j\leqslant
k}(x_i-x_j)$. In order to get character tables for the permutation
groups, we need the following result:
\begin{prop}[Frobenius character formula]
For all partitions $\lambda,\mu\vdash k$, we have that
\begin{eqnarray}
p_\mu(x_1,\ldots,x_k) V_{\delta}(x_1,\ldots,x_k)=
\sum_{\lambda\vdash
k}\chi_\lambda(K_\mu)V_{\lambda+\delta}(x_1,\ldots,x_k)\quad(\forall
\mu\vdash k),
\end{eqnarray}
where $K_\mu$ is the conjugacy class with given cycle type $\mu$;
and $\delta=(k-1,k-2,\ldots,1,0)$.
\end{prop}
Apparently $\chi_\lambda(K_\mu)$ is an entry from character table
matrix form $\bsT_k$. Note that by conventions in the character table, the first column $\chi_\lambda$ is odered lexicographically from top to bottom in terms of $\lambda$, i.e., $(n)>\cdots>(1,\ldots,1)$; and the first row $K_\mu$ is ordered lexicographically from left to right in terms of $\mu$, i.e., $(1^n)<\cdots<(n)$. We can see it below. In fact, we can use the mathematical software
\textsc{Mathematica} to obtain character tables of permutation
groups of various orders.

\subsection{The case where $k=2$}

For the permutation group $S_2$, let $K_{(1^2)}=\set{(1)(2)}$ and
$K_{(2^1)}=\set{(12)}$ such that $S_2=K_{(1^2)}\cup K_{(2^1)}$. We have
the following character table:
\begin{table}[h]
\centering
\begin{tabular}{c|cc}
  & $K_{(1^2)}$ & $K_{(2^1)}$ \\
 \hline $\chi_{(2)}$ & $1$ & $1$ \\
   $\chi_{(1,1)}$ & $1$ & $-1$
\end{tabular}\caption{The character table of $S_2$.}
\end{table}

\subsection{The case where $k=3$}

For the permutation group $S_3$, let $K_{(1^3)}=\set{(1)(2)(3)},
K_{(1^12^1)}=\set{(12),(13),(23)}$, and $K_{(3^1)}=\set{(123),(132)}$ such that
$S_2=K_{(1^3)}\cup K_{(1^12^1)}\cup K_{(3^1)}$. We have the following
character table:
\begin{table}[h]\label{tab:S3}
\centering
\begin{tabular}{c|ccc}
  & $K_{(1^3)}$ & $K_{(1^12^1)}$ & $K_{(3^1)}$\\
 \hline $\chi_{(3)}$ & $1$ & $1$ & $1$\\
   $\chi_{(2,1)}$ & $2$ & $0$ & $-1$\\
   $\chi_{(1,1,1)}$ & $1$ & $-1$ & $1$
\end{tabular}\caption{The character table of $S_3$.}
\end{table}

\subsection{The case where $k=4$}

For the permutation group $S_4$, let
\begin{eqnarray*}
K_1 &=& \set{(1)},\\
K_2 &=& \set{(12),(13),(14),(23),(24),(34)},\\
K_3 &=& \set{(12)(34),(13)(24),(14)(23)}, \\
K_4 &=& \Set{(123),(132),(124),(142),(134),(143),(234),(243)}, \\
K_5 &=& \Set{(1234),(1243),(1324),(1342),(1423),(1432)}.
\end{eqnarray*}
These five conjugacy classes are such that $S_4=K_1\cup K_2\cup
K_3\cup K_4\cup K_5$. Let
$\lambda_1=(4),\lambda_2=(3,1),\lambda_3=(2,2),\lambda_4=(2,1,1)$,
and $\lambda_5=(1,1,1,1)$ be the partitions of $4$. In fact,
$d^{c(\pi)}=\sum^5_{i=1}d_{\lambda_i}\chi_{\lambda_i}(\pi)$ implies
that
\begin{eqnarray*}
d^{c(\pi)}=\sum^5_{i=1}d_{\lambda_i}\chi_{\lambda_i}(\pi) =
(d_{\lambda_1},d_{\lambda_2},d_{\lambda_3},d_{\lambda_4},d_{\lambda_5})\Pa{\begin{array}{c}
                                                                  \chi_{\lambda_1}(\pi) \\
                                                                  \chi_{\lambda_2}(\pi) \\
                                                                  \chi_{\lambda_3}(\pi) \\
                                                                  \chi_{\lambda_4}(\pi) \\
                                                                  \chi_{\lambda_5}(\pi)
                                                                \end{array}
}.
\end{eqnarray*}
Let $\pi_1=(1),\pi_2=(12),\pi_3=(12)(34),\pi_4=(123),\pi_5=(1234)$.
Then $\pi_i\in K_i$ for $i=1,\ldots,5$. Thus
\begin{eqnarray*}
(d^{c(\pi_1)},\ldots,d^{c(\pi_5)}) =
(d_{\lambda_1},\ldots,d_{\lambda_5})\Pa{\begin{array}{ccccc}
                                                                  \chi_{\lambda_1}(\pi_1) &\chi_{\lambda_1}(\pi_2)&\chi_{\lambda_1}(\pi_3)&\chi_{\lambda_1}(\pi_4)&\chi_{\lambda_1}(\pi_5)\\
                                                                  \chi_{\lambda_2}(\pi_1) &\chi_{\lambda_2}(\pi_2)&\chi_{\lambda_2}(\pi_3)&\chi_{\lambda_2}(\pi_4)&\chi_{\lambda_2}(\pi_5)\\
                                                                  \chi_{\lambda_3}(\pi_1) &\chi_{\lambda_3}(\pi_2)&\chi_{\lambda_3}(\pi_3)&\chi_{\lambda_3}(\pi_4)&\chi_{\lambda_3}(\pi_5)\\
                                                                  \chi_{\lambda_4}(\pi_1) &\chi_{\lambda_4}(\pi_2)&\chi_{\lambda_4}(\pi_3)&\chi_{\lambda_4}(\pi_4)&\chi_{\lambda_4}(\pi_5)\\
                                                                  \chi_{\lambda_5}(\pi_1) &\chi_{\lambda_5}(\pi_2)&\chi_{\lambda_5}(\pi_3)&\chi_{\lambda_5}(\pi_4)&\chi_{\lambda_5}(\pi_5)\\
                                                                \end{array}
}.
\end{eqnarray*}
We have the following character table:
\begin{table}[h]
\centering
\begin{tabular}{c|ccccc}
   & $K_1$ & $K_2$ & $K_3$ & $K_4$ & $K_5$\\
 \hline $\chi_1$ & $1$ & $1$ & $1$ & $1$ & $1$\\
   $\chi_2$ & $3$ & $1$ & $-1$ & $0$ & $-1$\\
   $\chi_3$ & $2$ & $0$ & $2$ & $-1$ & $0$\\
   $\chi_4$ & $3$ & $-1$ & $-1$ & $0$ & $1$\\
   $\chi_5$ & $1$ & $-1$ & $1$ & $1$ & $-1$
\end{tabular}\caption{The character table of $S_4$.}
\end{table}

\subsection{The case where $k=5$}
For the permutation group $S_5$, let
\begin{eqnarray*}
K_1 &=& \set{(1)},\\
K_2 &=& \set{(12),(13),(14),(15),(23),(24),(25),(34),(35),(45)},\\
K_3 &=& \Big\{(12)(34),(12)(35),(12)(45),(13)(24),(13)(25),(13)(45),\\
&&~~~(14)(23),(14)(25),(14)(35),(15)(23),(15)(24),(15)(34),\\
&&~~~(23)(45),(24)(35),(25)(34)\Big\}, \\
K_4 &=& \Big\{(123),(132),(124),(142),(125),(152),(134),(143),(135),(153),\\
&&~~~(145),(154),(234),(243),(235),(253),(245),(254),(345),(354)\Big\},\\
K_5 &=& \Big\{(123)(45),(132)(45),(124)(35),(142)(35),(125)(34),(152)(34),(134)(25),\\
&&~~~(143)(25),(135)(24),(153)(24),(145)(23),(154)(23),(234)(15),(243)(15),\\
&&~~~(235)(14),(253)(14),(245)(13),(254)(13),(345)(12),(354)(12)\Big\}, \\
K_6 &=& \Big\{(1234),(1235),(1243),(1245),(1253),(1254),(1342),(1352),(1345),(1354),\\
&&~~~(1324),(1325),(1432),(1452),(1453),(1435),(1423),(1425),(1532),(1542),\\
&&~~~(1543),(1534),(1523),(1524),(2345),(2354),(2435),(2453),(2534),(2543)\Big\},\\
K_7 &=& \Big\{(12345),(12354),(12453),(12435),(12543),(12534),(13452),(13542),\\
&&~~~(13245),(13524),(13254),(13425),(14532),(14352),(14523),(14235),\\
&&~~~(14253),(14325),(15432),(15342),(15423),(15234),(15243),(15324)\Big\}
\end{eqnarray*}
These seven conjugacy classes are such that $S_5=\cup^7_{i=1}K_i$.
All partitions of $5$ are listed below:
\begin{eqnarray*}
\lambda_1=(5),\lambda_2=(4,1),\lambda_3=(3,2),\lambda_4=(3,1,1),\\
\lambda_5=(2,2,1),\lambda_6=(2,1,1,1),\lambda_7=(1,1,1,1,1).
\end{eqnarray*}
The character table of $S_5$ can be obtained immediately.
\begin{table}[h]
\centering
\begin{tabular}{c|ccccccc}
 & $K_1$ & $K_2$ & $K_3$ & $K_4$ & $K_5$ & $K_6$ & $K_7$\\
\hline $\chi_1$  & $1$ & $1$ & $1$ & $1$ & $1$ & $1$ & $1$ \\
 $\chi_2$  & $4$ & $2$ & $0$ & $1$ & $-1$ & $0$ & $-1$ \\
 $\chi_3$  & $5$ & $1$ & $1$ & $-1$ & $1$ & $-1$ & $0$ \\
 $\chi_4$  & $6$ & $0$ & $-2$ & $0$ & $0$ & $0$ & $1$ \\
 $\chi_5$  & $5$ & $-1$ & $1$ & $-1$ & $-1$ & $1$ & $0$ \\
 $\chi_6$  & $4$ & $-2$ & $0$ & $1$ & $1$ & $0$ & $-1$ \\
$\chi_7$  & $1$ & $-1$ & $1$ & $1$ & $-1$ & $-1$ & $1$
\end{tabular}\caption{The character table of $S_5$.}
\end{table}

\subsection{The case where $k=6$}

Since the order of the permutation group $S_6$ is $6!=720$, for
convenience, we just list a representative for each conjugacy class.
Let
\begin{eqnarray*}
K_1=\set{(1)}, K_2=\set{(12)}, K_3=\set{(12)(34)},
K_4=\set{(12)(34)(56)},\\
K_5=\set{(123)},K_6=\set{(123)(45)},K_7=\set{(123)(456)},K_8=\set{(1234)},\\
K_9=\set{(1234)(56)},K_{10}=\set{(12345)},K_{11}=\set{(123456)}.
\end{eqnarray*}
All partitions of $6$ are listed below:
\begin{eqnarray*}
\lambda_1=(6),\lambda_2=(5,1),\lambda_3=(4,2),\lambda_4=(4,1,1),\\
\lambda_5=(3,3),\lambda_6=(3,2,1),\lambda_7=(3,1,1,1),\lambda_8=(2,2,2)\\
\lambda_9=(2,2,1,1),\lambda_{10}=(2,1,1,1,1),\lambda_{11}=(1,1,1,1,1,1).
\end{eqnarray*}
The character table of $S_6$ can be given below:
\begin{table}[h]
\centering
\begin{tabular}{c|ccccccccccc}
 & $K_1$ & $K_2$ & $K_3$ &  $K_4$ & $K_5$ & $K_6$ & $K_7$ &
$K_8$ &  $K_9$ & $K_{10}$ &  $K_{11}$ \\
\hline $\chi_1$ & $1$ & $1$ & $1$ & $1$ & $1$ & $1$ & $1$ & $1$ & $1$ & $1$ & $1$  \\
 $\chi_2$ & $5$ & $3$ & $1$ & $-1$ & $2$ & $0$ & $-1$ & $1$ & $-1$ & $0$ & $-1$  \\
 $\chi_3$ & $9$ & $3$ & $1$ & $3$ & $0$ & $0$ & $0$ & $-1$ & $1$ & $-1$ & $0$  \\
 $\chi_4$ & $10$ & $2$ & $-2$ & $-2$ & $1$ & $-1$ & $1$ & $0$ & $0$ & $0$ & $1$  \\
 $\chi_5$ & $5$ & $1$ & $1$ & $-3$ & $-1$ & $1$ & $2$ & $-1$ & $-1$ & $0$ & $0$  \\
 $\chi_6$ & $16$ & $0$ & $0$ & $0$ & $-2$ & $0$ & $-2$ & $0$ & $0$ & $1$ & $0$  \\
 $\chi_7$ & $10$ & $-2$ & $-2$ & $2$ & $1$ & $1$ & $1$ & $0$ & $0$ & $0$ & $-1$  \\
 $\chi_8$ & $5$ & $-1$ & $1$ & $3$ & $-1$ & $-1$ & $2$ & $1$ & $-1$ & $0$ & $0$  \\
 $\chi_9$ & $9$ & $-3$ & $1$ & $-3$ & $0$ & $0$ & $0$ & $1$ & $1$ & $-1$ & $0$  \\
 $\chi_{10}$ & $5$ & $-3$ & $1$ & $1$ & $2$ & $0$ & $-1$ & $-1$ & $-1$ & $0$ & $1$  \\
 $\chi_{11}$ & $1$ & $-1$ & $1$ & $-1$ & $1$ & $-1$ & $1$ & $-1$ & $1$ & $1$ & $-1$
\end{tabular}\caption{The character table of $S_6$.}
\end{table}


\end{document}